\documentclass{article}

\usepackage{booktabs}
\usepackage{longtable}
\usepackage{array}
\usepackage[normalem]{ulem}
\usepackage{makecell}
\usepackage{algorithm}
\usepackage{algorithmic}
\usepackage{titletoc}
\usepackage[nice]{nicefrac}

\usepackage[final,main]{neurips_2026}
\usepackage[T1]{fontenc}
\usepackage{times}

\usepackage{amsmath}
\usepackage{amssymb}
\usepackage{mathtools}
\usepackage{amsthm}
\usepackage{enumitem}

\usepackage{tabularx}
\usepackage[table]{xcolor}
\usepackage{caption}
\usepackage{subcaption}
\usepackage{pifont}
\newcommand{\cmark}{\ding{51}}
\newcommand{\xmark}{\ding{55}}

\usepackage{siunitx}
\DeclareSIUnit{\rad}{rad}

\usepackage[colorlinks,citecolor=blue,linkcolor=blue]{hyperref}
\usepackage[capitalize,noabbrev,nameinlink]{cleveref}

\theoremstyle{plain}
\newtheorem{theorem}{Theorem}
\newtheorem{proposition}{Proposition}
\newtheorem{lemma}{Lemma}
\newtheorem{corollary}{Corollary}

\theoremstyle{definition}
\newtheorem{definition}{Definition}
\newtheorem{assumption}{Assumption}

\theoremstyle{remark}
\newtheorem{remark}{Remark}

\theoremstyle{plain}
\newtheorem{problem}{Problem}

\crefname{problem}{Problem}{Problems}
\Crefname{problem}{Problem}{Problems}

\crefname{assumption}{Assumption}{Assumptions}
\Crefname{assumption}{Assumption}{Assumptions}

\newcommand{\arxiv}[1]{}

\title{Submodular Multi-Agent Reinforcement Learning for Effective Online Distributed Task Allocation}

\author{%
  \makebox[\textwidth][c]{%
    \parbox{\textwidth}{%
      \centering
      Jing Liu$^{1}$\thanks{Equal contribution.}
      \quad
      Yangyang Yang$^{1}$\footnotemark[1]
      \quad
      Luca Ballotta$^{2}$
      \\[0.2em]
      Fangfei Li$^{1}$\thanks{Corresponding authors:
        Fangfei Li and Yang Tang.}
      \quad
      Yang Tang$^{3}$\footnotemark[2]
      \quad
      Ruggero Carli$^{2}$
      \\[0.45em]
      {\normalfont\small
        $^{1}$School of Mathematics,
        East China University of Science and Technology,
        Shanghai, China\\
        $^{2}$Department of Information Engineering,
        University of Padova, Padova, Italy\\
        $^{3}$School of Information Science and Engineering,
        East China University of Science and Technology,
        Shanghai, China
      }
      \\[0.3em]
      {\normalfont\small
        \texttt{\{y20220091,y20250091,lifangfei,yangtang\}@ecust.edu.cn}\\
        \texttt{\{luca.ballotta,ruggero.carli\}@unipd.it}
      }
    }%
  }%
}

\begin{document}
	
	\maketitle
	
	\begin{abstract}
		This paper studies online distributed multi-agent task allocation with submodular team utilities.
		Since each agent selects one action from a local categorical distribution, feasible joint actions form a partition matroid over agent-action pairs. The classical Multilinear Extension relaxation, based on independent Bernoulli sampling, does not align with categorical execution under partition-matroid constraints.
		To remedy this, we propose the \emph{Partition Multilinear Extension} (PME), which equals the expected stage utility under factorized categorical policies.
		We prove that submodular difference rewards provide unbiased PME marginal-gradient information and a stagewise policy-gradient estimator.
		Building on this connection, we propose \emph{SubMAPG}, a centralized-training, decentralized-execution framework using masked categorical policies and difference-reward signals.
		For projected stochastic-gradient ascent in the PME marginal space, we establish a stagewise $\nicefrac{1}{2}$-approximation guarantee and sublinear dynamic regret under slowly varying environments.
		We instantiate SubMAPG on graph neural network policy models that effectively handle open systems with time-varying agents and tasks.
		Experiments on cooperative coverage and multi-target tracking showcase superior performance compared with standard MARL that does not leverage submodular rewards, online submodular optimization, and greedy baselines.
	\end{abstract}
	
	\section{Introduction}
	\label{sec:intro}
	
	Distributed multi-agent task allocation arises in target tracking, active sensing, and coverage, where agents coordinate in a decentralized fashion under partial observability and dynamic environment. These problems are naturally sequential and time-varying, as targets move, observations change, and environments evolve over time. \autoref{fig:system} illustrates a representative sensing scenario with mobile robots, dynamic targets, and limited sensing and peer-to-peer communication; agents and tasks may even vary during prolonged deployments~\citep{zhang2018fully,deplano2026optimization}.	
	
	A central difficulty in these problems is that team utilities are typically non-additive across agents. For example, overlapping fields-of-view in \autoref{fig:system} reduce the total area covered by robots, such that their respective coverage contributions do not add linearly. This so-called diminishing-return structure is formally captured by monotone submodular utilities. Submodularity provides a principled basis for optimal combinatorial coordination since, if the feasible set is a matroid, they enable suboptimality  guarantees for greedy selection \citep{nemhauser1978analysis,fisher1978analysis} and continuous relaxations such as the Multilinear Extension (MLE) with rounding \citep{calinescu2011maximizing,zhang2025near}. 
	However, these optimization methods are online in nature and do not provide parametrized policies for sequential decision-making. 
	Moreover, even static submodular maximization over a matroid is NP-hard, making the direct optimization of joint actions exponentially costly in the size of agent and action sets \citep{prajapat2024submodular,chen2026multi}. 
	
	\begin{figure}[t]
		\centering
		\begin{minipage}[c]{0.58\linewidth}
			\centering
			\includegraphics[width=\linewidth]{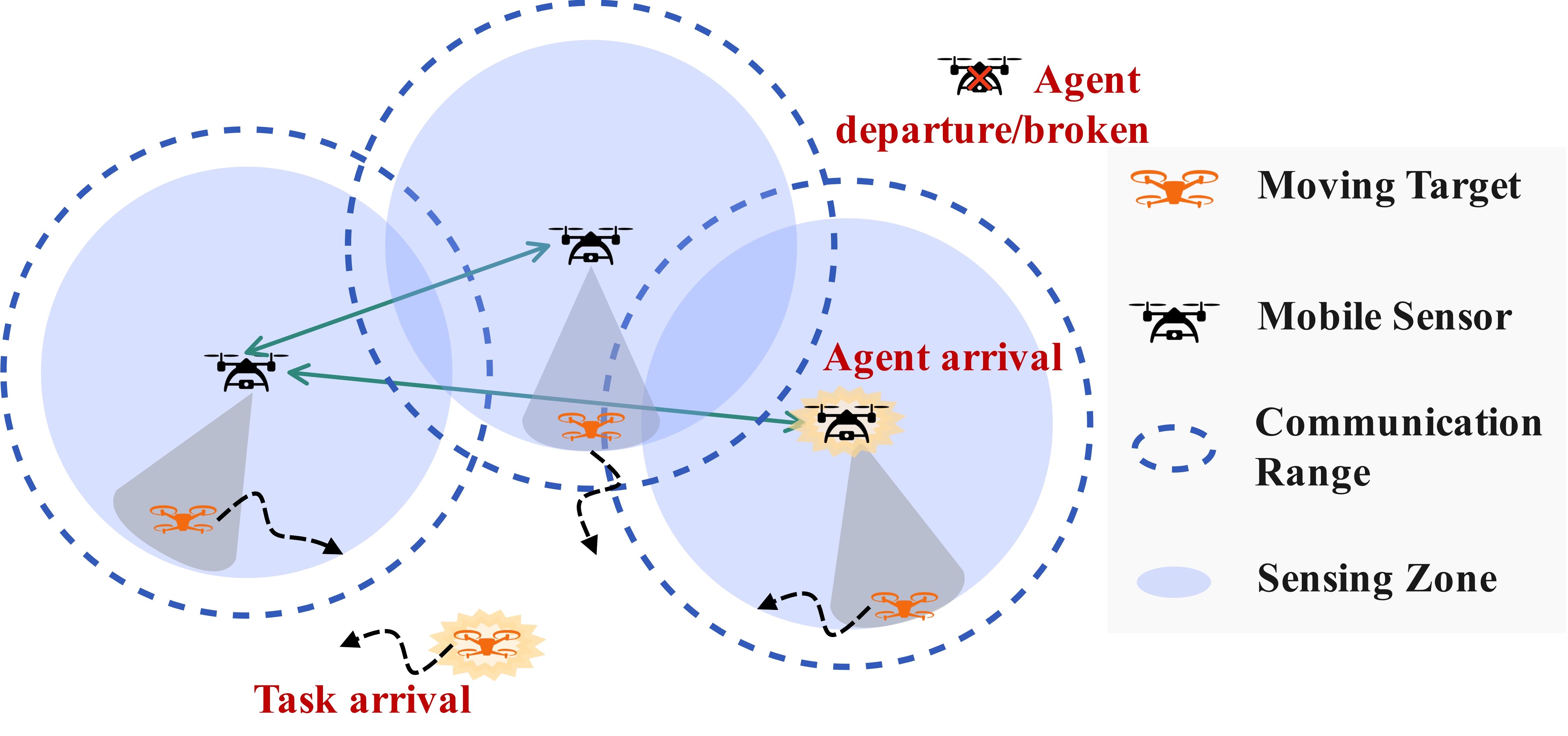}
		\end{minipage}
		\hfill
		\begin{minipage}[c]{0.38\linewidth}
			\caption{
				An open multi-agent system for task allocation. Mobile sensors collaboratively track mobile targets under time-varying agent populations, dynamic tasks, limited communication, and bounded sensing.}
			\label{fig:system}
		\end{minipage}
	\end{figure}
	
	In contrast, multi-agent reinforcement learning (MARL) offers a way to learn distributed parametrized policies for sequential cooperative decision-making under partial observability and dynamic environment. However, the relationship between submodular team utilities and policy optimization under multi-agent partition constraints remains insufficiently characterized. In particular, credit assignment under non-additive utilities is inherently harder than the standard additive case because each agent's marginal contribution depends critically on the actions of others. Difference rewards and counterfactual baselines can mitigate this difficulty, yet whether they estimate a consistent gradient direction for any continuous objective with submodular utilities remains unclear~\citep{foerster2018coma,castellini2025difference}. This motivates our central question:
	
	\emph{Can we establish a formal connection between credit assignment and MARL policy optimization for distributed task allocation under  submodular team utilities and partition matroid constraints?}
	
	Achieving this goal involves three challenges. First, consistent policy gradient requires a continuous objective that is mathematically equivalent to the expected utility of the executed distributed policy. The classical MLE constructs a continuous relaxation of set functions via independent Bernoulli sampling and typically relies on rounding post-hoc, creating a formal mismatch between the original optimization objective and the categorical policies executed by the agents. Second, while difference rewards for multi-agent credit assignment are mainly treated as variance-reduction or reward-shaping signals, it remains unclear if they estimate unbiased, and therefore formally consistent, gradient information for any  continuous objective under partition-matroid constraints. Third, task allocation requires scalable local policies that handle time-varying agent and task sets, whereas standard MARL assumes fixed-size inputs~\citep{lowe2017multi,yu2022surprising,anil2025mohito}.
	
	We address these challenges by proposing the \emph{Partition Multilinear Extension} (PME), a continuous relaxation of set functions that naturally accommodates partition matroid constraints. In particular, the PME equals the expected stage utility with \textit{executed} policies, unlike the MLE, aligning policy learning with the original submodular optimization problem. We prove that submodular difference rewards yield unbiased stagewise PME marginal-gradient information, thereby turning marginal-contribution credit assignment into consistent stochastic gradient information for partition-feasible objectives. Building on this connection, we propose \textit{SubMAPG}, a policy-gradient framework with centralized utility evaluations during training and decentralized masked-categorical execution that accommodates distributed multi-agent task allocation online. Our main contributions are:
	\vspace{-1em}
	\begin{enumerate}[label=(\arabic*), leftmargin=*]
		\item We formulate distributed multi-agent task allocation under monotone submodular utilities and partition-constrained categorical policies. While closely related works study RL with submodular rewards but assume fully centralized decision-making~\citep{prajapat2024submodular,desanti2024global} or globally coordinated sequential action selection~\cite{chen2026multi}, we study multi-agent policy-gradient learning with submodular stagewise team utilities under centralized training with decentralized execution (CTDE). In this setting, we introduce the PME, a continuous relaxation that equals the expected team utility under factorized categorical policies.
		
		\item We propose \textit{SubMAPG}, a CTDE MARL framework for submodular multi-agent task allocation with masked categorical policies.
		Submodular difference rewards provide unbiased marginal-gradient information for the partition-feasible PME objective, which is used to construct a gradient estimator used during training.
		To achieve formal insight into learning performance of SubMAPG,
		we study \emph{SubMAPG-$x$}, which implements projected stochastic gradient ascent directly in the PME marginal space and provides an analytically tractable way to perform a surrogate theoretical analysis of SubMAPG.
		For this algorithm, we prove a stagewise $\nicefrac{1}{2}$-approximation guarantee, matching the standard gap in static submodular optimization. Further,
		we establish a dynamic regret bound of $O(\sqrt{(1+\mathcal{P}_T)T})$ for SubMAPG-$x$, where the path length $\mathcal{P}_T$ quantifies the variation of optimal PME marginals across time.
		A key corollary is that this regret bound is sublinear with $T$ if $\mathcal{P}_T \in o(T)$,
		corresponding to slowly varying environments and tasks.
		Unlike online learning, which optimizes instantaneous allocation decisions instead of training parametrized allocation policies~\citep{zhang2025effective,zhang2025near},
		and centralized submodular RL and MARL~\citep{prajapat2024submodular,chen2026multi},
		SubMAPG marks the first attempt to connect formal suboptimality guarantees enabled by continuous relaxations with learning of decentralized multi-agent policies.
		
		\item We instantiate SubMAPG with tabular softmax~\citep{zhang2022on}, standard multilinear perceptron (MLP), and graph neural network (GNN) policy models.
		In particular,
		GNN models enable decentralized execution with time-varying agent and task sets, accommodating open systems and complementing recent works on task-open MARL~\citep{anil2025mohito}, which focuses on adaptability~\citep{anil2025mohito} rather than efficient credit assignment.
		Experiments on information coverage (cooperative sensing) and multi-target tracking show that SubMAPG outperforms shared-reward and online/local greedy baselines, is competitive with centralized myopic greedy methods in several settings, and generalizes across time-varying agent and task sets.
		
	\end{enumerate}
	
    \paragraph{Related work.}    
    \emph{Online submodular optimization and coordination.}
    Online and bandit methods address continuous DR-submodular optimization and discrete set-function maximization~\citep{zhang2019online,zhang2023online,tajdini2024nearly}, while gradient-based methods provide constant-factor guarantees for continuous DR-submodular maximization~\citep{bian2017guaranteed,hassani2017gradient}.
    Distributed methods address multi-agent submodular coordination~\citep{xu2023online,liuj2024distributed}, with MA-OSMA using an independent-Bernoulli MLE~\citep{zhang2025near} and MA-SPL employing a policy-based continuous extension structurally analogous to the PME, but focusing on consensus-based online regret minimization with fixed agents~\citep{zhang2025effective}.
    More broadly, online optimization does not synthesize parametric policies nor leverages underlying system dynamics.
    
    \emph{Credit assignment in cooperative MARL.}
    CTDE methods such as MADDPG, MAPPO, and COMA learn parameterized policies for decentralized execution~\citep{lowe2017multi,yu2022surprising,foerster2018coma}.
    Multi-agent credit assignment is addressed through difference rewards, counterfactual baselines, and monotonic value decomposition that attribute team feedback to each agent~\citep{wolpert2002collective,foerster2018coma,rashid2020monotonic}.
    Policy-gradient methods based on difference rewards have been formally analyzed~\citep{castellini2025difference}, but their connection to submodular marginal contributions and partition-feasible PME gradients remains unaddressed.

    \emph{Submodular RL and MARL.}
    Single-agent policy learning incorporates submodularity through greedy marginal-gain optimization for history-dependent rewards~\citep{prajapat2024submodular}, semi-gradient methods for global trajectory rewards~\citep{desanti2024global}, and KL-regularized optimization for adaptive submodular objectives~\citep{kveton2025adaptive}.
    Paper~\citep{chen2026multi} studies MARL with submodular team rewards through sequential greedy policy optimization, requiring global coordination during execution and neglecting open systems.
    Detailed comparisons are given in \autoref{app:related}.

	\section{Problem Formulation: Multi-Agent Online Submodular Coordination}
	\label{sec:problem}
	
	We formulate task allocation in a Multi-Agent System (MAS) as a decentralized partially observable sequential decision problem with a monotone submodular team utility.
	
	\paragraph{States, agents, and actions.}
	Let $t \in [T] \triangleq \{1,\ldots,T\}$ denote the round index and let $\mathcal{N}_t$ be the potentially time-varying set of active agents at round $t$. The global state $s_t \in \mathcal{S}$ includes active agents, tasks, and environment variables. At round $t$, each agent $i \in \mathcal{N}_t$ chooses an action $a_{i,t}$ from a finite action set $\mathcal{A}_{i,t}$ and the joint action is $\mathbf{a}_t = (a_{i,t})_{i \in \mathcal{N}_t}$. The system evolves according to a (time-varying) transition kernel $P_t$ such that $P_t(s_{t+1} \mid s_t,\mathbf{a}_t)$ is the probability of moving to state $s_{t+1}$ from state $s_t$ if actions $\mathbf{a}_t$ are applied.
	The ground set of available agent-action pairs is $\Omega_t \triangleq
	\{(i,a): i \in \mathcal{N}_t,\ a \in \mathcal{A}_{i,t}\}$. Action semantics depend on the specific application: an action may directly assign a task, e.g., sensing location, or select a motion primitive that leads to a task.

    \paragraph{Partition matroid constraint.}
	Each active agent executes at most one action per round. Therefore, the ground set
	decomposes into per-agent subsets $\Omega_{i,t} \triangleq \{(i,a): a \in \mathcal{A}_{i,t}\},
	i \in \mathcal{N}_t$. Feasible joint-action sets form a partition matroid $(\Omega_t,\mathcal{I}_t)$ with the family of independent sets
	$\mathcal{I}_t \triangleq
	\bigl\{
	A \subseteq \Omega_t:
	|A \cap \Omega_{i,t}| \leq 1,\ \forall i \in \mathcal{N}_t
	\bigr\}$.
	Each feasible joint action $A_t \in \mathcal{I}_t$ selects at most one action per agent. 
    
	\paragraph{Observations and decentralized policies.}
	Each agent receives a local observation
	$o_{i,t} = \mathcal{O}(s_t,i)$, which may include onboard sensing and messages
	from neighboring agents. Let $\mathbf{o}_t \triangleq (o_{i,t})_{i \in \mathcal{N}_t}$ denote the
	joint observation vector. We consider decentralized
	factorized categorical policies
	\begin{equation}
		\label{eq:factorized_policy}
		\pi(\mathbf{a}_t \mid \mathbf{o}_t)
		=
		\prod_{i \in \mathcal{N}_t}
		\pi^i(a_{i,t} \mid o_{i,t}),
	\end{equation}
	where $\pi^i(\cdot \mid o_{i,t})$ is a categorical distribution over $\mathcal{A}_{i,t}$.
    This policy class encodes decentralized execution and ensures partition feasibility by construction, since each active agent samples one action from its own (conditional) categorical distribution independently of all other agents.

    \paragraph{Submodular  utilities.}
	MAS performance at each round $t$ is measured by the set function
	$F_t(\cdot\,;s_t): 2^{\Omega_t} \to \mathbb{R}_{\geq 0}$,
    which we assume normalized, monotone, and
	submodular; see \autoref{app:assumptions} for formal definitions of these properties. This model covers common multi-agent objectives such as coverage~\citep{Welikala2022ANP}, facility location~\citep{zhou2019sensor}, and information gain~\citep{xu2025communication}, where overlapping agent contributions feature so-called diminishing returns.
	
	The goal is to learn a decentralized factorized policy of type~\eqref{eq:factorized_policy} that maximizes cumulative submodular utility over a time horizon, while satisfying the partition matroid constraint at every round.
	
	\begin{problem}[Online Multi-Agent Submodular Coordination]
		\label{prob:main}
		Given the time-varying MAS with initial state
        $s_1\sim\rho_1$, where $\rho_1$ is a given initial      distribution, find a factorized policy $\pi$ of the form \eqref{eq:factorized_policy} that solves
		\begin{equation}
			\label{eq:objective}
			\begin{aligned}
				\max_{\pi}\quad
				& J(\pi)
                \triangleq              \mathbb{E}_{s_1\sim\rho_1,\pi,\{P_t\}_{t=1}^T}  \left[\sum_{t=1}^{T}F_t(A_t;s_t)
                \right]\\			\mathrm{s.t.}\quad
				& A_t \triangleq
				\{(i,a_{i,t}): i \in \mathcal{N}_t\}
				\in \mathcal{I}_t,
				\qquad \forall t \in [T].
			\end{aligned}
		\end{equation}
	\end{problem}
	
	\cref{prob:main} combines partial observability, partition-feasible categorical actions, and sequential state dependence, making standard centralized submodular optimization relaxations insufficient.
	
	\section{Partition Multilinear Extension}
	\label{sec:pme}
	
	\cref{prob:main} is a combinatorial optimization problem over the partition matroid family $\mathcal{I}_t$, and is NP-hard even in static submodular maximization settings. Therefore, we introduce a continuous representation of the stage utility that matches the factorized categorical policies used by decentralized MARL. For each agent-action pair $(i,a)\in\Omega_t$, the decision variable $x_{(i,a)}\in[0,1]$ denotes the marginal probability that agent $i$ selects action $a \in \mathcal{A}_{i,t}$, i.e., $x_{(i,a)}=\pi^i(a\mid o_{i,t})$ under the factorized policy $\pi=\prod_i\pi^i$.
	The standard MLE~\citep{calinescu2011maximizing} samples each ground-set element independently. Under the partition matroid, this Bernoulli sampling may select multiple actions for one agent, so the optimized distribution differs from the categorical policy executed by decentralized agents. We therefore use a categorical extension over the partition matroid, called the \emph{Partition Multilinear Extension} (PME).
	All proofs of results in this section are deferred to \autoref{app:proofs_sec3}.
	
	\subsection{Definition and Key Properties}
	\label{sec:pme_definition}
	
	We first define the continuous domain associated with the partition matroid $(\Omega_t, \mathcal{I}_t)$. The associated matroid polytope is $\mathcal{P}(\mathcal{I}_t) = \left\{ \mathbf{x}\in[0,1]^{|\Omega_t|}: \sum_{a\in\mathcal{A}_{i,t}}x_{(i,a)}\le 1,\ \forall i\in\mathcal{N}_t
	\right\}$. Any vector $\mathbf{x} \in \mathcal{P}(\mathcal{I}_t)$ defines a product distribution in which each agent selects at most one feasible action, with the remaining probability assigned to idling.
	
	\begin{definition}[Partition Multilinear Extension]\label{def:pme}
		Given the partition matroid $(\Omega_t, \mathcal{I}_t)$
		and the set function $F_t(\,\cdot\,;\,s_t) : 2^{\Omega_t} \to
		\mathbb{R}_{\geq 0}$, the \emph{partition multilinear extension}
		$\tilde{f}_t(\,\cdot\,;\,s_t) : \mathcal{P}(\mathcal{I}_t) \to \mathbb{R}$ at state $s_t$ is
		\begin{equation}\label{eq:pmv_def}
			\tilde{f}_t(\mathbf{x};\,s_t) \triangleq
			\mathbb{E}_{A \sim \mathcal{D}(\mathbf{x})}[F_t(A;\,s_t)]
			= \sum_{A \in \mathcal{I}_t} F_t(A;\,s_t)
			\prod_{i \in \mathcal{N}_t} p_i(A; \mathbf{x}),
		\end{equation}
		where $\mathcal{D}(\mathbf{x})$ denotes the product of independent categorical distributions over agents, and agent $i$'s factor is
		\begin{equation}\label{eq:agent_prob}
			p_i(A; \mathbf{x}) =
			\begin{cases}
				x_{(i,a)}, & \text{if } A \cap \Omega_{i,t} = \{(i,a)\}, \\[4pt]
				1 - \displaystyle\sum_{a' \in \mathcal{A}_{i,t}} x_{(i,a')},
				& \text{if } A \cap \Omega_{i,t} = \emptyset.
			\end{cases}
		\end{equation}
	\end{definition}
	The PME samples at most one element from each agent subset: agent $i$ selects action $a$ with probability $x_{(i,a)}$ and remains idle with probability $1-\sum_{a\in\mathcal{A}_{i,t}}x_{(i,a)}$.
	Although the PME is defined on the full partition matroid polytope
	$\mathcal{P}(\mathcal{I}_t)$, the factorized categorical policies in
	\eqref{eq:factorized_policy} induce marginals on the smaller face
	\begin{equation}\label{eq:boundary_face}
		\mathcal{F}_t
		\triangleq
		\left\{
		\mathbf{x}\in\mathcal{P}(\mathcal{I}_t):
		\sum_{a\in\mathcal{A}_{i,t}}x_{(i,a)}=1,\;
		\forall i\in\mathcal{N}_t
		\right\}.
	\end{equation}
	On $\mathcal{F}_t$, the idle probability in~\eqref{eq:agent_prob} is zero, so the PME exactly represents the expected utility of the corresponding categorical policy. Since $\tilde f_t$ is monotone, optimizing $\tilde f_t$ over $\mathcal{P}(\mathcal{I}_t)$ has an optimizer on $\mathcal{F}_t$; see \autoref{app:pme_properties}. The PME is also nonnegative, monotone, smooth, and DR-submodular whenever $F_t$ is monotone submodular.
	\vspace{-0.6em}
	\paragraph{Gradient Characterization.}
	For $e\in\Omega_t$ and $A\subseteq\Omega_t\setminus\{e\}$, write
	$F_t(e\mid A;s_t)\triangleq F_t(A\cup\{e\};s_t)-F_t(A;s_t)$.
	The next result shows that each coordinate of the PME gradient equals
	the expected marginal contribution of the corresponding agent-action pair, thereby linking the continuous relaxation to agent-wise credit assignment.
	
	\begin{lemma}[Gradient of PME]\label{lem:gradient}
		The partial derivative of $\tilde{f}_t$ with respect to $x_{(i,a)}$ is
		given by
		\begin{equation}
			\frac{\partial \tilde{f}_t}{\partial x_{(i,a)}}(\mathbf{x})
			= \mathbb{E}_{A^{-i} \sim \mathcal{D}^{-i}(\mathbf{x})}
			\!\Big[ F_t\big((i,a) \mid A^{-i};\,s_t\big) \Big], \nonumber
		\end{equation}
		where $A^{-i}$ denotes a joint action of all agents except agent $i$, sampled from the product distribution $\mathcal{D}^{-i}(\mathbf{x})$ over $\mathcal{N}_t \setminus \{i\}$.
	\end{lemma}
	
	\cref{lem:gradient} shows that PME partial derivatives are expected marginal contributions of agent-action pairs.
	This fact is the basis for the difference-reward estimator introduced in \cref{sec:algorithm}.
	
	\subsection{Objective Equivalence}
	\label{sec:approximation}
	
	With the PME and its properties established, we now relax the discrete
	optimization \eqref{eq:objective} into a continuous maximization problem. We first show that, for any factorized categorical policy, the PME evaluated at its marginals exactly equals the policy's expected stage utility.
	
	\begin{lemma}[Objective Equivalence]\label{lem:equivalence}
		Let $\pi = \prod_{i \in \mathcal{N}_t} \pi^i$ be a
		factorized categorical policy as per~\eqref{eq:factorized_policy} and define the state-dependent marginal vector $\mathbf{x}_t(s_t)$ by
		$x_{(i,a)}(s_t)\triangleq \pi^i(a\mid \mathcal{O}(s_t,i))$. Then, for any fixed state $s_t$, the conditional expected stage utility under $\pi$ equals the PME evaluated at $\mathbf{x}(\pi)$:
		$\mathbb{E}_{A_t \sim \pi}\!\left[
		F_t(A_t;\,s_t)\,\middle|\,s_t
		\right]
		=
		\tilde{f}_t(\mathbf{x}_t(s_t);\,s_t)$. 
		Consequently, by the total expectation,
		$J_t(\pi)
		=
		\mathbb{E}_{s_t\sim\mu_t^\pi}
		\left[
		\tilde f_t(\mathbf{x}_t(s_t);s_t)
		\right]$, where $\mu_t^{\pi}$ denotes the marginal state
        distribution at round $t$ under policy $\pi$,
        starting from $s_1\sim\rho_1$.
		Note that for categorical policies,
		$\sum_{a} x_{(i,a)}(s_t) = 1$, so the idle mass
		$1 - \sum_{a} x_{(i,a)}$ in \eqref{eq:agent_prob} is exactly zero, ensuring exact equivalence at every state $s_t$.
	\end{lemma}	
	
	At a fixed state $s_t$, the PME evaluated at the policy-induced marginals equals the expected utility with the actions executed by the agents, without requiring any rounding.
	Since $F_t$ is monotone, any optimizer of $\tilde f_t$ over the full partition matroid polytope $\mathcal{P}(\mathcal{I}_t)$ can be chosen on $\mathcal{F}_t$; hence the continuous optimum over $\mathcal{P}(\mathcal{I}_t)$ matches the optimum achievable by factorized categorical marginals. Thus, the continuous stage problem can be written as $\max_{\mathbf{x}\in\mathcal{P}(\mathcal{I}_t)}
	\tilde f_t(\mathbf{x};s_t)$.
	In the next section, we optimize this objective through policy gradients. The gradient identity in \cref{lem:gradient} implies that submodular difference rewards provide unbiased gradient information for the PME, which connects the continuous relaxation to decentralized MARL credit assignment.
	
	\section{Policy Gradient as DR-Submodular Optimization}
	\label{sec:algorithm}
	
	The PME reduces the stagewise coordination problem to continuous optimization over policy-induced marginals. We now show how this structure leads to a policy-gradient algorithm for decentralized MARL. The key observation is that submodular marginal contributions provide stochastic gradient information for the PME, while masked categorical policies keep sampled joint actions feasible under the partition matroid. This yields \textit{SubMAPG}, a submodular multi-agent policy gradient method trained with centralized
	utility evaluations and executed using decentralized local policies.
	We use the CTDE paradigm: during training, decentralized actors train masked categorical policies from local observations and a centralized critic estimates state-value baselines; during execution, only the actors are used.
	Proofs are deferred to \autoref{app:proofs_sec4}.
    
	\subsection{Masked Categorical Policies}
	\label{sec:softmax_scope}
	
	For each agent $i\in\mathcal{N}_t$, SubMAPG outputs logits over action set $\mathcal{A}_{i,t}$ and applies a masked softmax:
	\begin{equation}
		\label{eq:masked_softmax_main}
		\pi_\theta^i(a\mid o_{i,t})
		=
		\frac{
			\exp\!\left([\mathbf{W}_h h_{i,t}]_a + M_a\right)}
		{\sum_{a'\in\mathcal{A}_{i,t}}
			\exp\!\left([\mathbf{W}_h h_{i,t}]_{a'} + M_{a'}\right)},
		\qquad a\in\mathcal{A}_{i,t},
	\end{equation}
	where $M_a=0$ for feasible actions and $M_a=-\infty$ for infeasible actions. 
	The induced marginal is $x_{(i,a)}(\theta)\triangleq \pi_\theta^i(a\mid o_{i,t})$.

    \begin{lemma}[Masked-Softmax Feasibility]
    	\label{lem:masked_softmax}
    	For all $\theta\in\mathbb{R}^d$, agents $i\in\mathcal{N}_t$,
    	and observations $o_{i,t}\in\mathcal{O}$,
    	$\sum_{a\in\mathcal{A}_{i,t}}
    	\pi_\theta^i(a\mid o_{i,t})=1$.
    	Hence, it holds $\mathbf{x}(\theta)\in\mathcal{F}_t$
    	and every sampled joint action is feasible under the
    	partition matroid.
    \end{lemma}
    
    Beyond feasibility, we examine the special case of tabular  softmax policies, for which policy-gradient updates increase the PME objective.
    
    \begin{proposition}[Tabular Softmax Induces Feasible PME-Based Ascent]
        \label{prop:tabular_equivalence}
        Let each agent $i$ use a tabular softmax policy
        $\pi_\theta^i(a\mid o)
        =\operatorname{softmax}(\theta^i_{o,\cdot})_a$
        with independent parameters $\theta^i_{o,a}\in\mathbb{R}$.
        Consider the policy-gradient update
        $\theta_{k+1}=\theta_k+\eta\hat g_k^{(\theta)}$,
        where $\hat g_k^{(\theta)}$ is an unbiased estimator of
        $\nabla_\theta J_t(\theta_k)$ conditional on $\theta_k$.
        Let $\mathbf{x}_k=\mathbf{x}(\theta_k)$ and
        $\Delta\mathbf{x}_k=\mathbf{x}_{k+1}-\mathbf{x}_k$.
        For sufficiently small $\eta$, it holds
        \[
        \mathbb{E}\!\left[
        \left\langle
        \nabla_{\mathbf{x}}\tilde f_t(\mathbf{x}_k;s_t),
        \Delta\mathbf{x}_k
        \right\rangle
        \,\middle|\,\theta_k
        \right]
        =
        \eta
        \left\|
        \mathbf{J}_k^\top
        \nabla_{\mathbf{x}}\tilde f_t(\mathbf{x}_k;s_t)
        \right\|_2^2
        +O(\eta^2)
        \geq -O(\eta^2),
        \]
        where
        $\mathbf{J}_k=
        \left.\frac{\partial\mathbf{x}(\theta)}
        {\partial\theta}\right|_{\theta=\theta_k}$.
    \end{proposition}
	
	\subsection{Submodular Policy Gradient via Difference Rewards}
	\label{sec:spg}
	
	A shared team reward $F_t(A_t;s_t)$ gives all agents the same scalar feedback and therefore obscures each agent's contribution to objective improvement.
	By \cref{lem:gradient}, the random variable $F_t((i,a)\mid A_t^{-i};s_t)$ with $A_t^{-i}\sim\mathcal{D}^{-i}(\mathbf{x})$ is an unbiased estimator of the partial derivative $\partial\tilde f_t(\mathbf{x};s_t)/\partial x_{(i,a)}$. 
	Hence, we use submodular marginal contributions
	$F_t((i,a_{i,t})\mid A_t^{-i};s_t)$ for credit assignment.
	
	We next state the formal stagewise policy-gradient identity that connects the PME marginal-gradient estimator to parameter updates of the factorized categorical policy.
	
	\begin{lemma}[Stagewise Submodular Policy Gradient Estimator]
		\label{lem:spg}
		Fix a round $t$ and state $s_t$. Let
		$x_{(i,a)}(\theta)=\pi_\theta^i(a\mid o_{i,t})$ and define the conditional stage objective $J_t(\theta;s_t) \triangleq
		\mathbb{E}_{A_t\sim\pi_\theta} \!\left[F_t(A_t;s_t)\mid s_t\right] = \tilde f_t(\mathbf{x}(\theta);s_t)$.
		Then
		$\nabla_\theta J_t(\theta;s_t)
		=
		\mathbb{E}_{A_t\sim\pi_\theta}
		\left[
		\sum_{i\in\mathcal{N}_t}
		\nabla_\theta \log \pi_\theta^i(a_{i,t}\mid o_{i,t})
		\,
		\bigl(F_t((i,a_{i,t})\mid A_t^{-i};s_t)-b_{i,t}
		\bigr)
		\right]$,
		where $b_{i,t}$ is any term independent of $a_{i,t}$ conditional on $s_t$, $o_{i,t}$, and the other agents' sampled actions.
	\end{lemma}
	
	\cref{lem:spg} follows from the score-function identity, policy factorization, and the action-independence of $F_t(A_t^{-i};s_t)$ given $s_t$ and $A_t^{-i}$.
	It justifies the sampled-action difference reward as a principled stagewise policy-gradient signal,
	while suboptimality analysis is deferred to \cref{sec:theory}.
	
	For the long-horizon objective $J$, an action at time $t$ may affect future utilities.
	In the implementation, SubMAPG uses accumulated difference returns as a Monte Carlo surrogate training signal for the long-horizon objective in~\eqref{eq:objective}~\citep{sutton2018reinforcement}:
	\begin{equation}
		\label{eq:Psi}		\Psi_{i,t}=\sum_{\ell=t}^{T}F_\ell((i,a_{i,\ell})\mid A_\ell^{-i};s_\ell)-b_{i,t},
	\end{equation}
	yielding the sample-based gradient estimator
	\begin{equation}
		\label{eq:spg_surrogate}
		\widehat{\nabla_\theta J}
		=
		\sum_{t=1}^{T}
		\sum_{i\in\mathcal{N}_t}
		\nabla_\theta \log \pi_\theta^i(a_{i,t}\mid o_{i,t})\,
		\Psi_{i,t}.
	\end{equation}
	
	\Cref{alg:gspg} summarizes SubMAPG.
	Training uses centralized evaluations of $F_t(A_t;s_t)$ and $F_t(A_t^{-i};s_t)$ to compute submodular difference rewards, while execution is decentralized through local masked categorical policies.
	The optimizer can be instantiated by tabular stochastic gradient ascent or neural actor-critic updates; 
	see implementation details in \autoref{app:architecture}.
	
	\begin{algorithm}
		\small
		\caption{Submodular Multi-Agent Policy Gradient (SubMAPG)}
		\label{alg:gspg}
		\begin{algorithmic}[1]
			\STATE {\bfseries Input:} horizon $T$, step size $\eta$, policy-gradient optimizer \texttt{PG-Optimizer}, policy parameters $\theta_0$
			\FOR{$k = 0,\ldots,K-1$}
			\STATE Collect trajectory
			$\tau=\{(s_t,A_t,F_t(A_t;s_t))\}_{t=1}^T$
			by executing $\pi_{\theta_k}$
			\FOR{$t=1,\ldots,T$}
			\FOR{each agent $i\in\mathcal{N}_t$}
			\STATE Compute marginal contributions
			$F_t((i,a_{i,t})\mid A_t^{-i};s_t)$
			\ENDFOR
			\ENDFOR
			\STATE Compute difference returns $\Psi_{i,t}$ using~\eqref{eq:Psi}
			\STATE Estimate the policy-gradient surrogate
			$\widehat{\nabla_\theta J}$ using~\eqref{eq:spg_surrogate}   
			\STATE Update
			$\theta_{k+1}\leftarrow
			\texttt{PG-Optimizer}(\theta_k,\widehat{\nabla_\theta J};\eta)$
			\ENDFOR
			\STATE {\bfseries Output:} trained policy parameters $\theta_K$
		\end{algorithmic}
	\end{algorithm}
	
	\section{Main Results: Performance Guarantees}
	\label{sec:theory}
		
	SubMAPG in \cref{alg:gspg} is difficult to formally analyze in general.
    This is because gradient-based updates depend on the specific architectural properties of the policy, making it difficult to establish convergence guarantees that hold uniformly across models and parametrization.    
	To obtain analytical insight,
	we introduce \emph{SubMAPG-$x$}, a variant of SubMAPG that directly optimizes action marginals $\mathbf{x}$ rather than policy parameters $\theta$, and establish stagewise approximation and dynamic-regret guarantees on it.
	This provides a surrogate analysis of SubMAPG in the PME marginal space without replacing parameter-space convergence guarantees for general policies.
    All proofs of this section are deferred to \autoref{app:proofs_sec5}.
    
    \subsection{SubMAPG-$x$: Projected PME Marginal-Space Optimization}
    SubMAPG-$x$ maintains a feasible PME marginal vector $\mathbf{x}_t\in\mathcal{F}_t$ whose coordinates define the categorical action distributions over active agents.
    At each round, it samples a joint action from the induced factorized categorical distribution and constructs a stochastic PME gradient using the marginal contributions of candidate actions.
    Then, it executes the sampled action and updates the marginal vector by projected gradient ascent,
    with a second projection accommodating changes in the active agent and action sets. The procedure is summarized in \Cref{alg:submapg-x}.
    We note that, although tabular softmax policies preserve feasibility and induce first-order PME marginal-ascent directions on $\mathcal{F}_t$ as per \cref{prop:tabular_equivalence}, their induced marginal updates need not coincide with the Euclidean projected-gradient updates implemented by SubMAPG-$x$. 
    
	\begin{algorithm}
		\small
		\caption{SubMAPG-$x$: Projected PME Marginal-Space Optimization}
		\label{alg:submapg-x}
		\begin{algorithmic}[1]
			\STATE {\bfseries Input:} horizon $T$, step size $\eta>0$,
			initial marginals $\mathbf{x}_1\in\mathcal{F}_1$
			\FOR{$t=1,\ldots,T$}
			\STATE Observe state $s_t$, active agents $\mathcal{N}_t$,
			and feasible action sets $\{\mathcal{A}_{i,t}\}_{i\in\mathcal{N}_t}$
			\STATE Sample $A_t\sim\mathcal{D}(\mathbf{x}_t)$
			\STATE Compute $\mathbf{g}_t$ coordinatewise as
			$(\mathbf{g}_t)_{(i,a)}=F_t((i,a)\mid A_t^{-i};s_t)$
			for all $(i,a)\in\Omega_t$
			\STATE Extend $\mathbf{g}_t$ by zero on inactive coordinates
			\STATE Execute $A_t$
			\IF{$t<T$}
			\STATE $\bar{\mathbf{x}}_{t+1}\gets\Pi_{\mathcal{F}_t}(\mathbf{x}_t+\eta \mathbf{g}_t)$
			\STATE Obtain the next feasible face $\mathcal{F}_{t+1}$
			\STATE $\mathbf{x}_{t+1}\gets\Pi_{\mathcal{F}_{t+1}}(\bar{\mathbf{x}}_{t+1})$
			\ENDIF
			\ENDFOR
			\STATE {\bfseries Output:} marginal sequence $\{\mathbf{x}_t\}_{t=1}^{T}$
		\end{algorithmic}
	\end{algorithm}
    
	\subsection{Stagewise Approximation Guarantee}
	\label{sec:stagewise}
	
	We first analyze a fixed round $t$ with fixed state-dependent utility $F_t(\cdot;s_t)$ and feasible family $\mathcal{I}_t$. Our goal is to bound the gap between the expected utility of a categorical policy induced by the projected marginal iterate and the optimal discrete action set for this fixed stage.
	
	\begin{definition}[Stagewise $(\alpha,\epsilon)$-Approximation]
		\label{def:stagewise}
		Given $\alpha > 0$ and $\epsilon \geq 0$, a policy $\pi$ achieves a
		\emph{stagewise $(\alpha,\epsilon)$-approximation} if, for every round $t$,
		$ \mathbb{E}_{A_t \sim \pi(\cdot \mid \mathbf{o}_t)}
		[F_t(A_t;\,s_t)]
		\geq \alpha\, \mathrm{OPT}_t(s_t) - \epsilon$,
		where
		$\mathrm{OPT}_t(s_t) \triangleq \max_{A \in \mathcal{I}_t} F_t(A;\,s_t)$
		is the optimal discrete utility at state $s_t$.
	\end{definition}
	
	The proof uses the fact that the categorical-policy face $\mathcal{F}_t$ is a convex face of the partition matroid polytope, that the PME satisfies the restricted DR first-order inequality on this face, and that monotonicity ensures an optimal PME solution exists on $\mathcal{F}_t$; see \Cref{lem:dr_boundary}.
	
	\begin{theorem}[Stagewise Approximation]
		\label{thm:stagewise}		
        For a given round $t$,
        consider the marginal update of SubMAPG-$x$ on $\mathcal{F}_t$,
        $\mathbf{x}_{k+1}=\Pi_{\mathcal{F}_t}(\mathbf{x}_k+\eta\mathbf{g}_k)$,        
        where $\mathbf{x}_0\in\mathcal{F}_t$ and
        $(\mathbf{g}_k)_{(i,a)}
        =F_t((i,a)\mid A_k^{-i};s_t)$, with
        $A_k^{-i}\sim\mathcal{D}^{-i}(\mathbf{x}_k)$.
		Let $\eta = D_t/\sqrt{K(G^2+\sigma^2)}$, where $D_t\triangleq        \max_{\mathbf{x},\mathbf{y}\in\mathcal{F}_t}
        \|\mathbf{x}-\mathbf{y}\|_2$
        is the diameter of the stage-$t$ feasible face.
		Under \cref{ass:submodular},
		\begin{equation}
			\label{eq:stagewise_bound}
			\frac{1}{K}\sum_{k=0}^{K-1}
			\mathbb{E}\!\left[\tilde{f}_t(\mathbf{x}_k;\,s_t)\right]
			\geq
			\frac{1}{2}\mathrm{OPT}_t(s_t)
			-
			\frac{D\sqrt{G^2+\sigma^2}}{2\sqrt{K}}.
		\end{equation}
		Consequently, if $k^*$ is sampled uniformly from $\{0,\ldots,K-1\}$ and $\pi_{k^*}$ is the factorized categorical policy
		with marginal vector $\mathbf{x}_{k^*}$, then
		$ \mathbb{E}\left[
		\mathbb{E}_{A_t\sim\pi_{k^*}}\![F_t(A_t;\,s_t)]
		\right]
		\geq
		\frac{1}{2}\mathrm{OPT}_t(s_t)
		-
        \frac{D\sqrt{G^2+\sigma^2}}{2\sqrt{K}}$, where the outer expectation is over the random $k^*$ and stochastic gradient iterates.
	\end{theorem}
	
	\Cref{thm:stagewise} characterizes the approximation power of SubMAPG-$x$ when it takes multiple steps $K$ on the same stage objective $\tilde f_t$.
    It thus serves as a static PME marginal-space optimization benchmark that justifies the constant $\alpha=1/2$ in the following regret analysis, where we restore $K=1$.
	
	\subsection{Dynamic Regret Analysis}
	\label{sec:regret}
	We now consider the open-system setting where the utility $F_t$, active agent set $\mathcal{N}_t$, ground set $\Omega_t$, and feasible face $\mathcal{F}_t$ may change over time. Following online learning theory~\citep{zinkevich2003online,hazan2016introduction}, we use a constant step size to balance stochastic gradient noise and the ability to track a moving comparator.
	We measure performance against a dynamic $\alpha$-approximation benchmark.
	
	\begin{definition}[Dynamic $\alpha$-Regret]
		\label{def:regret}
		Given $\alpha \in (0,1]$, the expected dynamic $\alpha$-regret of a policy $\pi_\theta$ over horizon $T$ is
		$\mathrm{Regret}_T^{\alpha}(\pi_\theta)
		\triangleq
		\sum_{t=1}^{T}
		\left[
		\alpha\,\mathbb{E}_{s_t \sim \mu_t^{\pi_\theta}}
		\!\left[\mathrm{OPT}_t(s_t)\right]
		-
		\mathbb{E}_{\pi_\theta}\!\left[F_t(A_t;\,s_t)\right]
		\right]$,
		where $\mathrm{OPT}_t(s_t) \triangleq \max_{A\in\mathcal{I}_t}F_t(A;\,s_t)$ and $\mu_t^{\pi_\theta}$ is the state distribution at round $t$ induced by $\pi_\theta$.
	\end{definition}
	This regret is policy-consistent: both the benchmark and the achieved utility are evaluated under the same state distribution $\mu_t^{\pi_\theta}$. This avoids comparing against a clairvoyant trajectory distribution that is unavailable in an online Dec-POMDP.
	
	To quantify non-stationarity across time-varying 
	agent sets, we first embed all marginal vectors into a common 
	Euclidean space. Let $N_{\max} = \max_t |\mathcal{N}_t|$ and
	$N^a_{\max} = \max_{t,i}|\mathcal{A}_{i,t}|$. We define an 
	isometric zero-padding embedding
	$\iota_t : \mathcal{F}_t \rightarrow 
	\mathbb{R}^{N_{\max}N^a_{\max}}$
	by assigning each agent $i \in \mathcal{N}_t$ a fixed index slot 
	in $\{1,\ldots,N_{\max}\}$ and padding coordinates of inactive 
	agents with zeros. Since $\iota_t$ preserves $\ell_2$ norms and 
	inner products on the support of $\mathcal{F}_t$, all subsequent 
	distances are well-defined in the common ambient space.
	The \emph{expected path length} of the optimal continuous 
	solutions is defined as
	\begin{equation}\label{eq:path_length}
		\mathcal{P}_T \;\triangleq\; \mathbb{E}\!\left[\sum_{t=1}^{T-1}
		\left\|\iota_t(\mathbf{x}_t^*) 
		- \iota_{t+1}(\mathbf{x}_{t+1}^*)\right\|_2\right],
	\end{equation}
	where
	$\mathbf{x}_t^* \in
	\arg\max_{\mathbf{x}\in\mathcal{F}_t}
	\tilde f_t(\mathbf{x};s_t)$
	is an optimal continuous solution at the realized state $s_t$, and the expectation is over trajectories generated from $s_1\sim\rho_1$ under $\pi_\theta$. A small $\mathcal{P}_T$ means that the optimal PME marginals drift slowly over time, even though the open system changes dimension.
	
	Define the polytope diameter $D \triangleq
	\max_{t,u\in[T]}
	\max_{\mathbf{x}\in\iota_t(\mathcal{F}_t),\,
		\mathbf{y}\in\iota_u(\mathcal{F}_u)}
	\|\mathbf{x}-\mathbf{y}\|_2$, largest gradient magnitude
	$G \triangleq
	\max_t\max_{\mathbf{x}\in\mathcal{F}_t}
	\sup_{s_t\in\mathcal{S}}
	\|\nabla \tilde f_t(\mathbf{x};s_t)\|_2$,
	and gradient estimator variance bound $\sigma^2
	\triangleq
	\max_{t\in[T]}
	\mathbb{E}\!\left[
	\|\mathbf{g}_t-\nabla \tilde f_t(\mathbf{x}_t;s_t)\|_2^2
	\right]$.
	Because each active agent's marginal vector lies in a probability simplex of diameter $\sqrt{2}$, the embedded feasible sets satisfy
	$D\leq\sqrt{2N_{\max}}$.
	
	\begin{theorem}[Dynamic Regret Bound]
		\label{thm:regret}
		Suppose $F_t(\cdot;s_t)$ satisfies \Cref{ass:submodular}
		for all $t\in[T]$. Then $\tilde f_t(\cdot;s_t)$ is nonnegative, monotone, and
		DR-submodular on the partition matroid polytope
		$\mathcal{P}(\mathcal{I}_t)$. Since
		$\mathcal{F}_t\subseteq \mathcal{P}(\mathcal{I}_t)$, the following
		first-order gap inequality holds for all $\mathbf{x},\mathbf{y}\in\mathcal{F}_t$:
        \begin{equation}
            \frac{1}{2}\tilde f_t(\mathbf{y};s_t)
            -\tilde f_t(\mathbf{x};s_t)
            \leq
            \frac{1}{2}
            \left\langle
            \nabla\tilde f_t(\mathbf{x};s_t),
            \mathbf{y}-\mathbf{x}
            \right\rangle.
        \end{equation}
		Let $\{\mathbf{x}_t\}_{t=1}^T$ denote the feasible PME marginal sequence evaluated at the realized states and generated by SubMAPG-$x$ with constant stepsize $\eta>0$.
        Then, it holds
		\begin{equation}
			\label{eq:regret_master}
			\mathbb{E}\!\left[\mathrm{Regret}_T^{1/2}(\pi)\right]
			\leq
			\frac{D^2}{4\eta}
			+
			\frac{D\mathcal{P}_T}{2\eta}
			+
			\frac{\eta T(G^2+\sigma^2)}{4}.
		\end{equation}
		Moreover, setting $\eta=
		\sqrt{
			\frac{D(D+2\mathcal{P}_T)}
			{T(G^2+\sigma^2)}
		}$ gives
		\begin{equation}
			\label{eq:regret_exact}
			\mathbb{E}\!\left[\mathrm{Regret}_T^{1/2}(\pi)\right]
			\leq
			\frac{1}{2}
			\sqrt{D(D+2\mathcal{P}_T)\,T(G^2+\sigma^2)}.
		\end{equation}
	\end{theorem}
	
	The first two terms in~\eqref{eq:regret_master} are initialization and tracking costs, while the last term is the cumulative stochastic-gradient cost.
	Constants $D,G,$ and $\sigma^2$ can be explicitly bounded based on system parameters;
	see \autoref{app:proof_bounds}. 
    The next corollary establishes sublinear regret when $\mathcal{P}_T$ grows slowly.
	
	\begin{corollary}[Sublinear Regret] \label{cor:sublinear}
		If $D$, $G$, and $\sigma$ are uniformly bounded and the path length satisfies $\mathcal{P}_T=o(T)$, then
		$\limsup_{T\to\infty}
		\frac{1}{T}
		\mathbb{E}\!\left[\mathrm{Regret}_T^{1/2}(\pi)\right]
		\le 0$.
	\end{corollary}

    Finally, we use the regret bound in \cref{thm:regret} to derive a suboptimality gap for \cref{prob:main}.
    
    \begin{corollary}[Episodic Performance Bound]
        \label{cor:episodic}
        Under the assumptions of \Cref{thm:regret}, the cumulative objective of the policy induced by the SubMAPG-$x$ marginal sequence satisfies
        \begin{equation}
            \label{eq:episodic_bound}
            \mathbb{E}[J(\pi)]
            \geq \frac{1}{2}J^\star
            -\frac{1}{2}\mathbb{E}[\Delta_T^{\mathrm{occ}}(\pi)]
            -\frac{1}{2}\sqrt{D(D+2\mathcal{P}_T)T(G^2+\sigma^2)}.
        \end{equation}
        where
        \[\Delta_T^{\mathrm{occ}}(\pi)
        \triangleq
        \left[\sum_{t=1}^{T}\left(
        \mathbb{E}_{s_t\sim\mu_t^{\pi^\star}}[\mathrm{OPT}_t(s_t)]
        -\mathbb{E}_{s_t\sim\mu_t^\pi}[\mathrm{OPT}_t(s_t)]
        \right)\right]_+,\]
        $\pi^\star$ is an optimal policy, and $[x]_+=\max\{x,0\}$.
    \end{corollary}
    
    These theoretical results show that SubMAPG-$x$ preserves the approximation
    guarantees of submodular coordination in direct PME marginal-space optimization, and that this guarantee extends to a long-horizon performance gap relative to an optimal policy.
	\section{Experiments}                  
	\label{sec:experiments}
	
	\begin{figure*}[t]
		\centering
		\vspace{-0.5em}
		
		\begin{subfigure}[t]{0.24\linewidth}
			\centering
			\includegraphics[width=\linewidth]{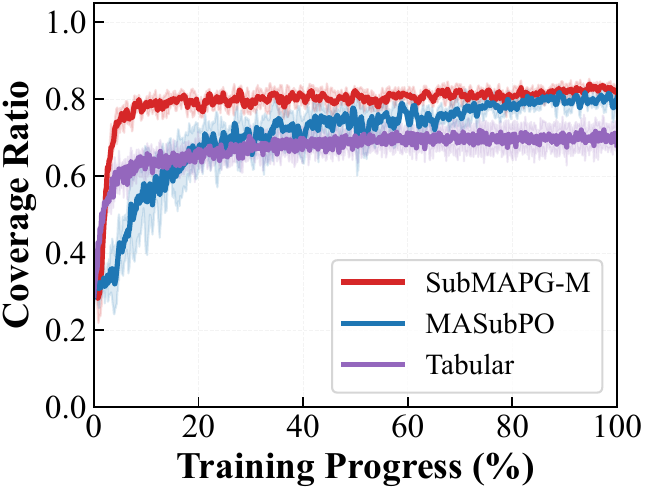}
			\caption{Coverage train (GP)}
			\label{fig:main_cov_train_gp}
		\end{subfigure}
		\hfill
		\begin{subfigure}[t]{0.24\linewidth}
			\centering
			\includegraphics[width=\linewidth]{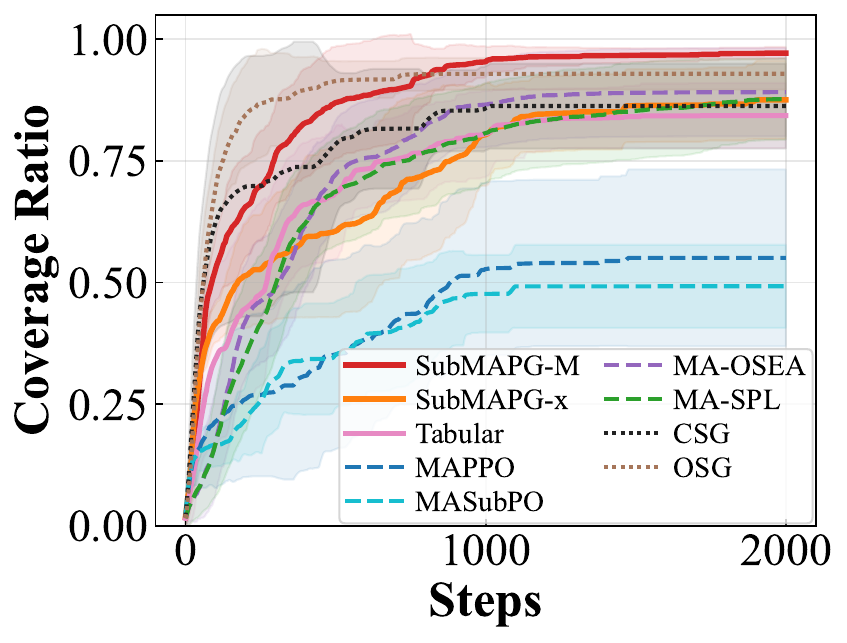}
			\caption{Open cov. ratio (GP)}
			\label{fig:main_cov_ratio_gp}
		\end{subfigure}
		\hfill
		\begin{subfigure}[t]{0.24\linewidth}
			\centering
			\includegraphics[width=\linewidth]{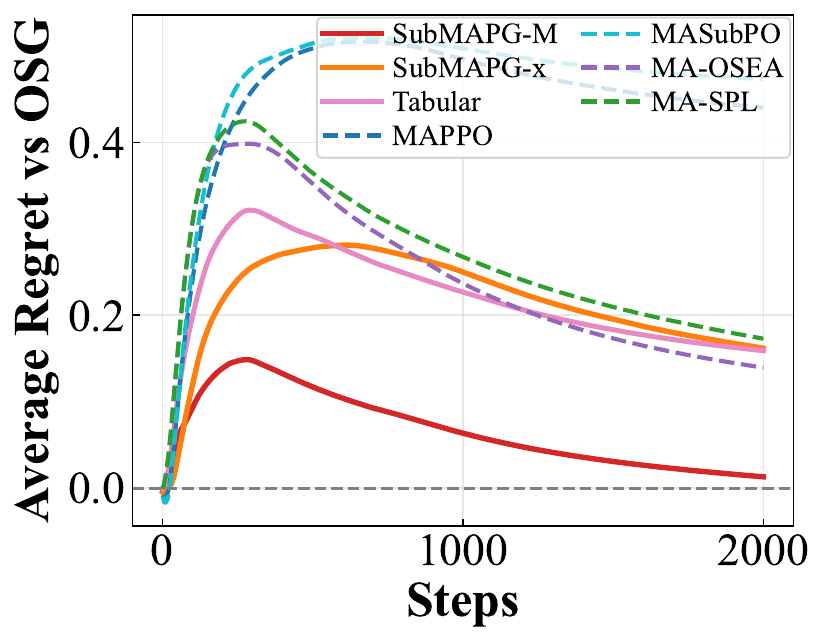}
			\caption{Open cov. gap (GP)}
			\label{fig:main_cov_regret_gp}
		\end{subfigure}
		\hfill
		\begin{subfigure}[t]{0.24\linewidth}
			\centering
			\includegraphics[width=\linewidth]{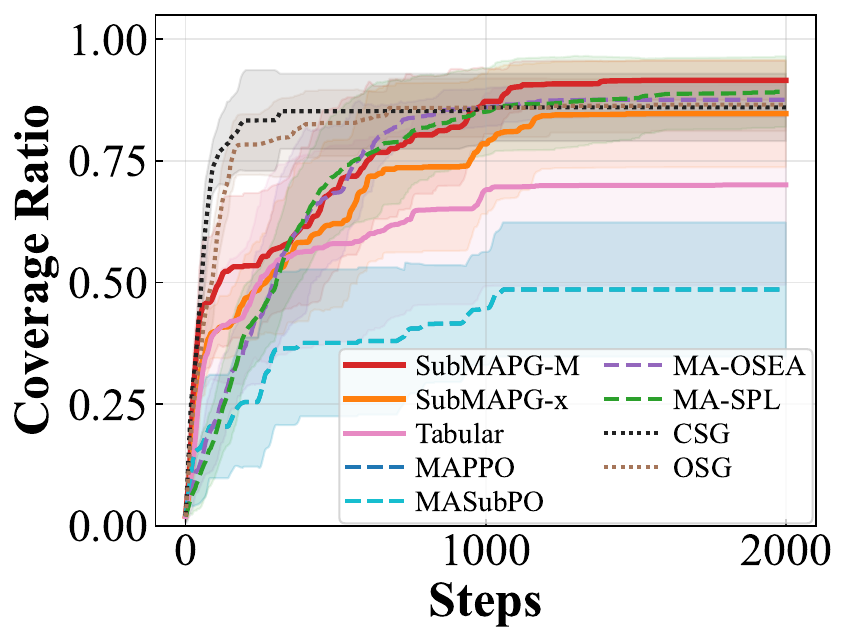}
			\caption{Open cov. (Bimodal)}
			\label{fig:main_cov_ratio_bimodal}
		\end{subfigure}
		
		\vspace{0.15em}
		
		\begin{subfigure}[t]{0.24\linewidth}
			\centering
			\includegraphics[width=\linewidth]{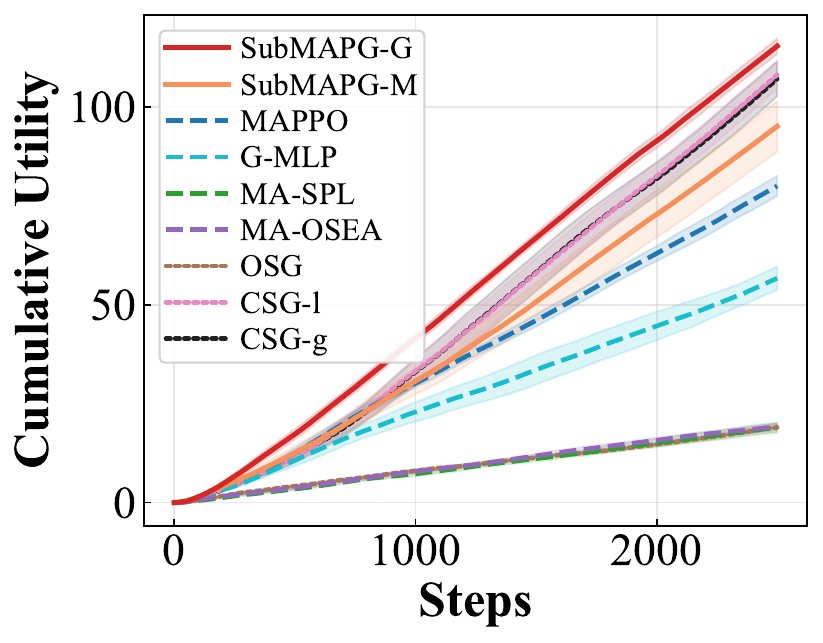}
			\caption{Tracking utility}
			\label{fig:main_track_utility}
		\end{subfigure}
		\hfill
		\begin{subfigure}[t]{0.24\linewidth}
			\centering
			\includegraphics[width=\linewidth]{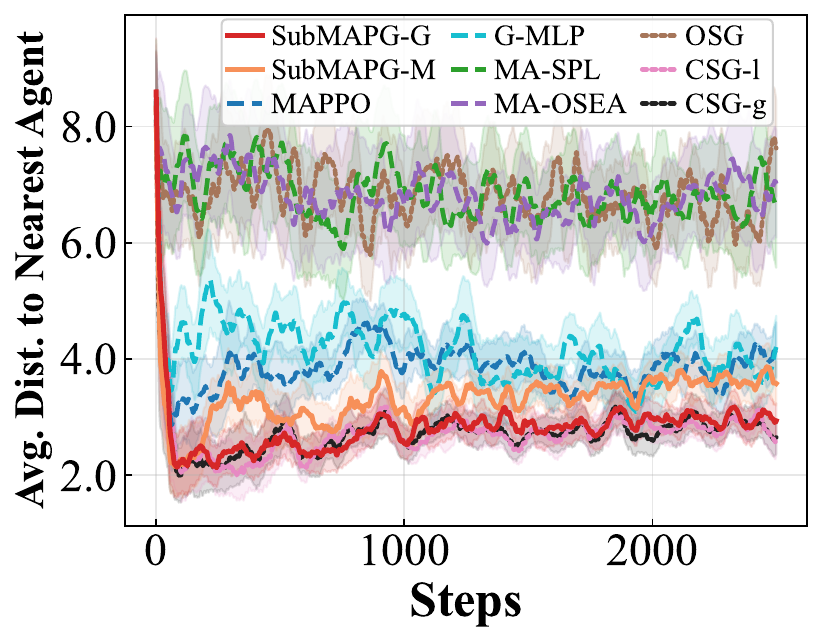}
			\caption{Tracking distance}
			\label{fig:main_track_distance}
		\end{subfigure}
		\hfill
		\begin{subfigure}[t]{0.24\linewidth}
			\centering
			\includegraphics[width=\linewidth]{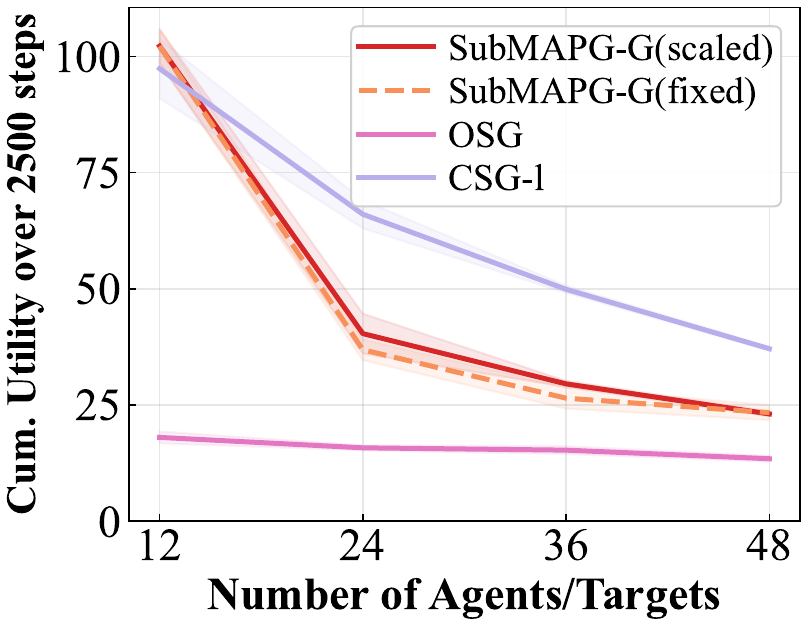}
			\caption{Zero-shot scaling}
			\label{fig:main_track_scaling}
		\end{subfigure}
		\hfill
		\begin{subfigure}[t]{0.24\linewidth}
			\centering
			\includegraphics[width=\linewidth]{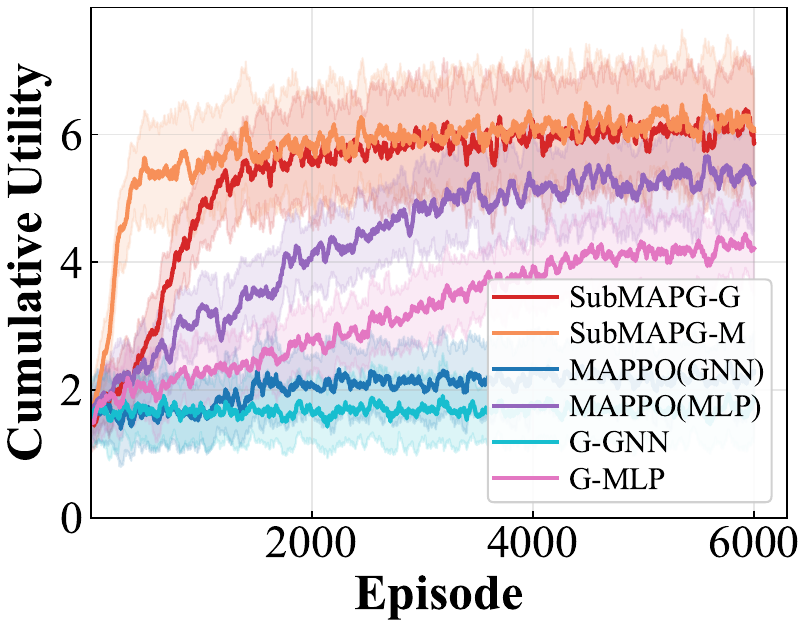}
			\caption{Reward ablation}
			\label{fig:main_reward_ablation}
		\end{subfigure}
		\caption{Main empirical results. Top row: information coverage on the Gaussian Process and Bimodal density field. Bottom row: dynamic target tracking under open-system evaluation, zero-shot scaling, and reward ablation.}
		\label{fig:main_results}
		\vspace{-1.0em}
	\end{figure*}
	
	We evaluate SubMAPG on information coverage and target tracking, including the PME marginal-space variant SubMAPG-$x$
    and neural instantiations SubMAPG-M and SubMAPG-G, based on MLP and GNN policies respectively.    
    These experiments illustrate representative applications, while the framework can also extend to other allocation settings, such as epidemic control.
	Architectures and additional results are in \autoref{app:experiments}.
	
	\paragraph{Multi-agent information coverage.}
	We evaluate SubMAPG-M on a $30\times30$ grid coverage task with up to five agents. Each agent chooses from five motion actions and covers cells within radius $r_{\mathrm{cov}}=1$; the team utility is the total covered information value
	$F_t^{\mathrm{cov}}(A_t)=\sum_{v\in\cup_{i\in\mathcal{N}_t}D_{i,t}}\rho(v)$.
	Policies are trained in closed environments and evaluated in open environments over $T=2000$ steps with agent arrivals and departures. We compare with a tabular softmax baseline (``Tabular'')~\citep{zhang2022on};
	Multi-Agent Proximal Policy Optimization (``MAPPO'')~\citep{yu2022surprising} with shared team reward $r_t=F_t^{\mathrm{cov}}(A_t)$;
	Multi-Agent Submodular Policy Optimization (``MASubPO''), an extension of Submodular Policy Optimization (SubPO) using the shared reward $r_t=F_t^{\mathrm{cov}}(A_t)-F_{t-1}^{\mathrm{cov}}(A_{t-1})$~\citep{prajapat2024submodular};
	Multi-Agent Surrogate Policy Learning (``MA-SPL'')~\citep{zhang2025effective};
	Multi-Agent Online Surrogate Entropic Ascent (``MA-OSEA'')~\citep{zhang2025near};
	centralized sequential greedy (``CSG'')~\citep{fisher1978analysis}; 
	online submodular greedy (``OSG'')~\citep{xu2023online}.  
	\autoref{fig:main_results}(a)--(d) show that SubMAPG-M converges faster than MASubPO, while SubMAPG-M and SubMAPG-x maintain high coverage under openness and achieve smaller empirical utility gaps relative to OSG.
	
	\vspace{-1.2em}
	\paragraph{Dynamic multi-target tracking.}
	We evaluate the GNN-based SubMAPG-G in a $100\mathrm{m}\times100\mathrm{m}$ tracking environment with mobile agents and targets. Agents use unicycle dynamics, select one of 12 steering actions, observe targets within $r_{\mathrm{sen}}=10\mathrm{m}$, and communicate within $r_{\mathrm{com}}=25\mathrm{m}$. The tracking utility is the normalized sum of the best tracking scores over active targets, $F_t^{\mathrm{trk}}(A_t)=\frac{1}{|\mathcal{N}_t||\mathcal{M}_t|}\sum_{j\in\mathcal{M}_t}\max_{(i,a)\in A_t}w_{j,t}(i,a)$. In open evaluation, agents and targets arrive and depart over $T=2500$ steps. 
	We compare with centralized sequential greedy with global sensing (``CSG-g'') or local sensing (``CSG-l''), a shared-reward ablation (``G-MLP''), OSG, MAPPO, MA-SPL, MA-OSEA, and SubMAPG-M. \autoref{fig:main_results}(e)--(h) show that SubMAPG-G achieves cumulative utility comparable to centralized greedy, keeps target-agent distances low using only local observations, and transfers zero-shot from systems with up to 12 agents and targets to systems with up to 48; the reward ablation confirms that submodular difference rewards outperform shared global rewards under the same MLP encoder, isolating the benefit of marginal-contribution credit assignment.
	
	\section{Conclusion}
	\label{sec:conclusion}
	
	We introduced SubMAPG for distributed submodular task allocation.
	Our defined PME matches expected stage utility of factorized categorical policies, and we prove that difference rewards provide unbiased marginal-gradient information.
	We perform surrogate analysis to achieve analytical insight into our method,
	establishing sublinear dynamic-regret guarantees for projected PME marginal updates under slowly varying environment and tasks.
	Experiments on cooperative coverage and tracking in open systems demonstrate the effectiveness of SubMAPG with neural policies.
	Establishing provable suboptimality gaps for neural policies is a compelling avenue for future work.

    \begin{ack}
    This work was supported in part by the National Key Research
    and Development Program of China (Grant No. 2025YFA1016504);
    the National Natural Science Foundation of China
    (Grant Nos. 62573198, 62233005, U2441245, and U25B6002);
    the National Key Laboratory of Space Target Awareness
    (Grant No. STA2025ZCB0208);
    and the Shanghai Institute for Mathematics and
    Interdisciplinary Sciences (SIMIS)
    (Grant No. SIMIS-ID-2025-SP).
    \end{ack}
	\bibliography{refs}
	\bibliographystyle{plainnat}
	
	\newpage
	\appendix
	\numberwithin{equation}{section}
	\startcontents[appendix]
	\printcontents[appendix]{l}{0}{\section*{Appendix}}
	
	\newpage

	\section{List of Symbols}
	\label{app:symbols}
	
	\renewcommand{\arraystretch}{1.12}
	\setlength{\tabcolsep}{6pt}
	
	\begin{longtable}{@{}>{$}l<{$} @{\quad} c @{\quad} p{0.78\linewidth}@{}}
		\toprule
		\endfirsthead
		\toprule
		\endhead
		\bottomrule
		\endfoot
		
		\multicolumn{3}{c}{\textbf{General Indices and Sets}}\\
		\midrule
		t                      & $\triangleq$ & Time/round index, $t\in[T]\triangleq\{1,\dots,T\}$.\\
		T                      & $\triangleq$ & Horizon length / total number of rounds.\\
		K                      & $\triangleq$ & Number of policy update iterations.\\
		i                      & $\triangleq$ & Agent index, $i\in\mathcal{N}_t$.\\
		j                      & $\triangleq$ & Task/target index, $j\in\mathcal{M}_t$ when tasks or targets are considered.\\
		\mathcal{N}_t          & $\triangleq$ & Active agent set at round $t$.\\
		\mathcal{M}_t          & $\triangleq$ & Active task/target set at round $t$ in the experiments.\\
		N_{\max}               & $\triangleq$ & Maximum number of active agents, $N_{\max}=\max_t|\mathcal{N}_t|$.\\
		M_{\max}               & $\triangleq$ & Maximum number of active tasks/targets, $M_{\max}=\max_t|\mathcal{M}_t|$.\\
		N^a_{\max}             & $\triangleq$ & Maximum action-set size, $N^a_{\max}=\max_{t,i}|\mathcal{A}_{i,t}|$.\\
		N_0                    & $\triangleq$ & Number of active agents at the initial round in tracking experiments.\\
		M_0                    & $\triangleq$ & Number of active targets at the initial round in tracking experiments.\\
		\addlinespace
		
		\multicolumn{3}{c}{\textbf{MAS State, Observation, and Action Model}}\\
		\midrule
		s_t                    & $\triangleq$ & Global state at round $t$, including active agents, tasks, and environment variables.\\
		\mathcal{S}            & $\triangleq$ & Global state space.\\
		P_t(\cdot\mid s_t,\mathbf{a}_t) & $\triangleq$ & Time-varying controlled transition kernel.\\
		\mathcal{A}_{i,t}      & $\triangleq$ & Discrete feasible action set of agent $i$ at round $t$.\\
		a_{i,t}                & $\triangleq$ & Action sampled or executed by agent $i$ at round $t$.\\
		\mathbf{a}_t           & $\triangleq$ & Joint action tuple, $\mathbf{a}_t=(a_{i,t})_{i\in\mathcal{N}_t}$.\\
		o_{i,t}                & $\triangleq$ & Local observation of agent $i$ at round $t$.\\
		\mathbf{o}_t           & $\triangleq$ & Joint observation tuple, $\mathbf{o}_t=(o_{i,t})_{i\in\mathcal{N}_t}$.\\
		\mathcal{O}            & $\triangleq$ & Observation map or observation space, depending on context.\\
		\addlinespace
		
		\multicolumn{3}{c}{\textbf{Ground Set and Partition Matroid}}\\
		\midrule
		\Omega_t          & $\triangleq$ & Ground set of agent-action pairs at round $t$: $\Omega_t=\{(i,a): i\in\mathcal{N}_t,\ a\in\mathcal{A}_{i,t}\}$.\\
		\Omega_{i,t}      & $\triangleq$ & Agent-$i$ subset in the partition: $\Omega_{i,t}=\{(i,a): a\in\mathcal{A}_{i,t}\}$.\\
		A_t               & $\triangleq$ & Realized joint action set, $A_t=\{(i,a_{i,t}):i\in\mathcal{N}_t\}\in\mathcal{I}_t$.\\
		A_t^{-i}          & $\triangleq$ & Joint action set excluding agent $i$: $A_t^{-i}=A_t\setminus\{(i,a_{i,t})\}$.\\
		A^{-i}            & $\triangleq$ & Generic partial joint action of all agents except agent $i$.\\
		\mathcal{I}_t     & $\triangleq$ & Feasible family of the partition matroid: $\mathcal{I}_t=\{A\subseteq\Omega_t: |A\cap\Omega_{i,t}|\le1,\ \forall i\in\mathcal{N}_t\}$.\\
		\mathcal{I}_t^{-i} & $\triangleq$ & Feasible family for all agents except $i$, defined over $\mathcal{N}_t\setminus\{i\}$.\\
		\mathcal{P}(\mathcal{I}_t) & $\triangleq$ & Partition matroid polytope: $\{\mathbf{x}\in[0,1]^{|\Omega_t|}:\sum_{a\in\mathcal{A}_{i,t}}x_{(i,a)}\le1,\ \forall i\in\mathcal{N}_t\}$.\\
		\mathcal{F}_t & $\triangleq$ & Categorical-policy face of the partition matroid polytope: $\{\mathbf{x}\in\mathcal{P}(\mathcal{I}_t):\sum_{a\in\mathcal{A}_{i,t}}x_{(i,a)}=1,\ \forall i\in\mathcal{N}_t\}$.\\
		x_{(i,a)}         & $\triangleq$ & Marginal probability that agent $i$ selects action $a$.\\
		\mathbf{x}        & $\triangleq$ & Vector of action marginals over $\Omega_t$.\\
		\mathbf{x}_t      & $\triangleq$ & Marginal vector at round $t$ in the PME marginal-space dynamics.\\
		\mathbf{x}(\theta) & $\triangleq$ & Marginal vector induced by $\pi_\theta$, with $x_{(i,a)}(\theta)=\pi_\theta^i(a\mid o_{i,t})$.\\
		\mathbf{x}_t^*    & $\triangleq$ & Optimal continuous solution at round $t$: $\mathbf{x}_t^*\in\arg\max_{\mathbf{x}\in\mathcal{F}_t}\tilde f_t(\mathbf{x};s_t)$.\\
		\addlinespace
		
		\multicolumn{3}{c}{\textbf{Submodular Utilities}}\\
		\midrule
		F_t(\,\cdot\,;\,s_t) & $\triangleq$ & Stage utility set function, $F_t(\,\cdot\,;\,s_t):2^{\Omega_t}\to\mathbb{R}_{\ge0}$, assumed normalized, monotone, and submodular.\\
		F_t(e\mid A;\,s_t) & $\triangleq$ & Marginal gain of element $e$: $F_t(e\mid A;\,s_t)=F_t(A\cup\{e\};s_t)-F_t(A;s_t)$.\\
		F_t((i,a)\mid A^{-i};s_t) & $\triangleq$ & Marginal contribution of agent-action pair $(i,a)$ given the other agents' action set $A^{-i}$.\\
		B                      & $\triangleq$ & Uniform bound on marginal gains: $0\le F_t(e\mid A;s_t)\le B$.\\
		F_t^{\mathrm{cov}}     & $\triangleq$ & Coverage utility in the information coverage experiment.\\
		F_t^{\mathrm{trk}}     & $\triangleq$ & Tracking utility in the multi-target tracking experiment.\\
		\addlinespace
		
		\multicolumn{3}{c}{\textbf{Continuous Relaxations}}\\
		\midrule
		f_t^{\mathrm{mle}}(\mathbf{x}) & $\triangleq$ & Standard Bernoulli-based Multilinear Extension used in classical submodular maximization.\\
		\tilde f_t(\mathbf{x};s_t) & $\triangleq$ & Partition Multilinear Extension (PME): $\tilde f_t(\mathbf{x};s_t)=\mathbb{E}_{A\sim\mathcal{D}(\mathbf{x})}[F_t(A;s_t)]$.\\
		\mathcal{D}(\mathbf{x})       & $\triangleq$ & Product distribution of independent per-agent categorical factors induced by $\mathbf{x}$.\\
		\mathcal{D}^{-i}(\mathbf{x})  & $\triangleq$ & Product distribution over agents $\mathcal{N}_t\setminus\{i\}$ induced by $\mathbf{x}$.\\
		p_i(A;\mathbf{x}) & $\triangleq$ & Agent-$i$ factor in the PME: $x_{(i,a)}$ if $A\cap\Omega_{i,t}=\{(i,a)\}$, and $1-\sum_{a'}x_{(i,a')}$ if $A\cap\Omega_{i,t}=\emptyset$.\\
		\boldsymbol{\chi}_e & $\triangleq$ & Standard basis vector associated with element $e$ in the definition of DR-submodularity.\\
		\mathbf{e}_{(i,a)} & $\triangleq$ & Standard basis vector associated with coordinate $(i,a)$.\\
		\addlinespace
		
		\multicolumn{3}{c}{\textbf{Policies and Policy Gradient}}\\
		\midrule
		\pi                    & $\triangleq$ & Joint factorized categorical policy, $\pi(\mathbf{a}_t\mid\mathbf{o}_t)=\prod_{i\in\mathcal{N}_t}\pi^i(a_{i,t}\mid o_{i,t})$.\\
		\pi^i                  & $\triangleq$ & Local categorical policy of agent $i$ over $\mathcal{A}_{i,t}$.\\
		\pi_\theta             & $\triangleq$ & Parameterized factorized policy with parameters $\theta$.\\
		\pi_\theta^i           & $\triangleq$ & Parameterized local policy of agent $i$.\\
		\pi^{-i}               & $\triangleq$ & Joint policy of all agents except $i$: $\pi^{-i}=\prod_{j\in\mathcal{N}_t\setminus\{i\}}\pi^j$.\\
		\theta                 & $\triangleq$ & Policy parameters.\\
		\phi                   & $\triangleq$ & Value-function parameters in the actor-critic implementation.\\
        \rho_1                  & $\triangleq$ & Initial-state distribution.\\
		J(\pi)                 & $\triangleq$ & Cumulative objective, $J(\pi)=\mathbb{E}_\pi[\sum_{t=1}^T F_t(A_t;s_t)]$.\\
		J(\theta)              & $\triangleq$ & Cumulative objective under $\pi_\theta$, $J(\theta)=\mathbb{E}_{\pi_\theta}[\sum_{t=1}^T F_t(A_t;s_t)]$.\\
		J_t(\pi)               & $\triangleq$ & Expected stage utility, $J_t(\pi)=\mathbb{E}_{\pi}[F_t(A_t;s_t)]$.\\
		J_t(\theta)            & $\triangleq$ & Expected stage utility under $\pi_\theta$, $J_t(\theta)=\mathbb{E}_{A_t\sim\pi_\theta}[F_t(A_t;s_t)]$.\\
		\mu_t^\pi              & $\triangleq$ & Marginal state distribution at round $t$ induced by policy $\pi$.\\
		\mu_t^{\pi_\theta}     & $\triangleq$ & Marginal state distribution at round $t$ induced by policy $\pi_\theta$.\\		
		r_t                    & $\triangleq$ & Shared global or temporal team reward used by baselines or ablations.\\
		\Psi_{i,t}             & $\triangleq$ & Difference return, $\Psi_{i,t}=\sum_{k=t}^{T}F_k((i,a_{i,k})\mid A_k^{-i};s_k)-b_{i,t}$.\\
		b_{i,t}                & $\triangleq$ & Agent- and time-dependent baseline that is independent of $a_{i,t}$ under the conditioning used in \cref{lem:spg}.\\
		\mathbf{g}_t           & $\triangleq$ & Stochastic PME gradient estimator based on submodular marginal contributions at round $t$.\\
		\mathbf{g}_k           & $\triangleq$ & Stochastic gradient estimator used at policy update iteration $k$.\\
		\widehat{\nabla_\theta J} & $\triangleq$ & Monte Carlo estimate of the policy-gradient surrogate used in Algorithm~\ref{alg:gspg}.\\
		\eta                   & $\triangleq$ & Step size.\\
		\eta^*                 & $\triangleq$ & Optimized constant step size in the dynamic regret bound.\\
		\addlinespace
		
		\multicolumn{3}{c}{\textbf{Approximation and Dynamic Regret}}\\
		\midrule
		\mathrm{OPT}_t(s_t) & $\triangleq$ & Optimal discrete utility at state $s_t$: $\mathrm{OPT}_t(s_t)=\max_{A\in\mathcal{I}_t}F_t(A;s_t)$.\\
		\alpha              & $\triangleq$ & Approximation factor, e.g., $\alpha=1/2$ in the main guarantees.\\
		\epsilon            & $\triangleq$ & Additive approximation error in the definition of stagewise $(\alpha,\epsilon)$-approximation.\\
		\mathrm{Regret}_T^\alpha(\pi_\theta) & $\triangleq$ & Dynamic $\alpha$-regret: $\sum_{t=1}^T[\alpha\,\mathbb{E}_{s_t\sim\mu_t^{\pi_\theta}}[\mathrm{OPT}_t(s_t)]-\mathbb{E}_{\pi_\theta}[F_t(A_t;s_t)]]$.\\
		\mathcal{P}_T       & $\triangleq$ & Expected path length of optimal continuous solutions: $\mathcal{P}_T=\mathbb{E}[\sum_{t=1}^{T-1}\|\iota_t(\mathbf{x}_t^*)-\iota_{t+1}(\mathbf{x}_{t+1}^*)\|_2]$.\\
		\iota_t             & $\triangleq$ & Isometric zero-padding embedding, $\iota_t:\mathcal{F}_t\rightarrow\mathbb{R}^{N_{\max}N^a_{\max}}$.\\
		D                   & $\triangleq$ & Uniform diameter of embedded feasible faces, $D=\max_{t,u\in[T]}\max_{\mathbf{x}\in\iota_t(\mathcal{F}_t),\,\mathbf{y}\in\iota_u(\mathcal{F}_u)}\|\mathbf{x}-\mathbf{y}\|_2$.\\
		G                   & $\triangleq$ & Uniform gradient norm bound, $G=\max_t\max_{\mathbf{x}\in\mathcal{F}_t}\sup_{s_t\in\mathcal{S}}\|\nabla\tilde f_t(\mathbf{x};s_t)\|_2$.\\
		\sigma^2            & $\triangleq$ & Gradient estimator variance bound, $\sigma^2=\max_{t\in[T]}\mathbb{E}[\|\mathbf{g}_t-\nabla\tilde f_t(\mathbf{x}_t;s_t)\|_2^2]$.\\
		\Pi_{\mathcal{F}_t} & $\triangleq$ & Euclidean projection onto $\mathcal{F}_t$.\\
		T_{\mathcal{F}_t}(\mathbf{x}) & $\triangleq$ & Tangent space of $\mathcal{F}_t$ at $\mathbf{x}$.\\
		\mathrm{proj}_{T_{\mathcal{F}_t}(\mathbf{x})} & $\triangleq$ & Projection onto the tangent space $T_{\mathcal{F}_t}(\mathbf{x})$.\\
		\addlinespace
		
		\multicolumn{3}{c}{\textbf{Policy Architecture Symbols}}\\
		\midrule
		\Phi_\theta            & $\triangleq$ & GNN encoder mapping graph-structured observations to node-level latent representations.\\
		\mathcal{G}_{i,t}      & $\triangleq$ & Local observation graph used by agent $i$ at round $t$.\\
		\mathcal{G}_t          & $\triangleq$ & Global state graph used by the critic in the GNN implementation.\\
		\mathcal{E}_t          & $\triangleq$ & Edge set containing communication and sensing links.\\
		h_{i,t}                & $\triangleq$ & Latent embedding of agent $i$ at round $t$.\\
		\mathbf{W}_h           & $\triangleq$ & Learnable projection from latent embeddings to action logits.\\
		M_a                    & $\triangleq$ & Mask value in the masked softmax: $M_a=0$ for feasible actions and $M_a=-\infty$ for infeasible actions.\\
		V_\phi                 & $\triangleq$ & Centralized value function used by the critic.\\
		k                      & $\triangleq$ & Number of nearest neighboring agents or targets used by the MLP encoder when fixed-dimensional inputs are constructed.\\
		\addlinespace
		
		\multicolumn{3}{c}{\textbf{Experiment Symbols}}\\
		\midrule
		D_{i,t}                & $\triangleq$ & Set of grid cells covered by agent $i$ at round $t$ in the coverage experiment.\\
		\rho(v)                & $\triangleq$ & Information density at grid cell $v$.\\
		v                      & $\triangleq$ & Grid cell index in the coverage utility.\\
		r_{\mathrm{cov}}       & $\triangleq$ & Coverage radius in the information coverage experiment.\\
		r_{\mathrm{com}}       & $\triangleq$ & Communication radius.\\
		r_{\mathrm{sen}}       & $\triangleq$ & Sensing radius.\\
		\mathbf{p}_{i,t}       & $\triangleq$ & Position of agent $i$ at round $t$.\\
		\mathbf{q}_{j,t}       & $\triangleq$ & Position of target $j$ at round $t$.\\
		p^x_{i,t},p^y_{i,t}    & $\triangleq$ & Cartesian coordinates of agent $i$ at round $t$.\\
		v_a                    & $\triangleq$ & Agent speed.\\
		v_m                    & $\triangleq$ & Target speed.\\
		\Delta t               & $\triangleq$ & Time-step interval in the tracking dynamics.\\
		\psi_{i,t}             & $\triangleq$ & Heading angle of agent $i$ at round $t$.\\
		\omega_{i,t}           & $\triangleq$ & Angular velocity or steering action of agent $i$ at round $t$.\\
	\end{longtable}
	
	\section{Preliminaries}
	\label{app:assumptions}
	
	\begin{definition}[Submodular Function \citep{nemhauser1978analysis}]
		\label{def:app_submodular}
		Let $\mathcal{E}$ be a finite ground set.
		A set function $F : 2^\mathcal{E} \to \mathbb{R}$ is
		\emph{submodular} if for every $A \subseteq B \subseteq \mathcal{E}$
		and every $e \in \mathcal{E} \setminus B$,
		\begin{equation}
			\label{eq:app_submodular_def}
			F(A \cup \{e\}) - F(A)
			\ge
			F(B \cup \{e\}) - F(B).
		\end{equation}
		Equivalently, $F$ is submodular if and only if for every
		$A, B \subseteq \mathcal{E}$,
		\begin{equation}
			\label{eq:app_submodular_equiv}
			F(A) + F(B) \ge F(A \cup B) + F(A \cap B).
		\end{equation}
	\end{definition}
	
	The defining property \eqref{eq:app_submodular_def} is the \emph{diminishing-returns} condition: the marginal gain of adding element $e$ to a set does not increase as the set grows. This captures the redundancy that arises in multi-agent coordination
	when agents sense overlapping regions, cover intersecting areas, or collect correlated information.
	
	\begin{definition}[Normalized Monotone Submodular Function \citep{qu2019distributed}]
		\label{def:app_normalized_monotone}
		A submodular function $F : 2^\mathcal{E} \to \mathbb{R}_{\ge 0}$ is:
		\begin{enumerate}[label=(\roman*)]
			\item \emph{normalized} if $F(\emptyset) = 0$;
			\item \emph{monotone} if $F(A) \le F(B)$ whenever $A \subseteq B$.
		\end{enumerate}
	\end{definition}
	
	\begin{definition}[Marginal Gain]
		\label{def:marginal}
		The marginal gain of element $e \in \Omega_t$
		with respect to set $A \subseteq \Omega_t \setminus \{e\}$ under state
		$s_t$ is $F_t(e \mid A;\,s_t) \;\triangleq\;
		F_t(A \cup \{e\};\,s_t) - F_t(A;\,s_t)$.
	\end{definition}
	
	The stage utility $F_t$ in \cref{prob:main} is assumed to be
	normalized, monotone, and submodular throughout the paper.
	
	\begin{definition}[DR Property and DR-Submodular Functions~\citep{bian2017guaranteed}]
		\label{def:dr-submodular}
		A function $f: \mathcal{X} \to \mathbb{R}$ defined over a convex domain $\mathcal{X} \subseteq [0, 1]^{|\Omega_t|}$ satisfies the diminishing returns (DR) property if for all $\mathbf{a} \leq \mathbf{b} \in \mathcal{X}$, $\forall e \in \Omega_t$, $\forall k \in \mathbb{R}_+$ such that $(k\boldsymbol{\chi}_e + \mathbf{a}) \in \mathcal{X}$ and $(k\boldsymbol{\chi}_e + \mathbf{b}) \in \mathcal{X}$, it holds:
		\begin{equation}
			f(k\boldsymbol{\chi}_e + \mathbf{a}) - f(\mathbf{a}) \geq f(k\boldsymbol{\chi}_e + \mathbf{b}) - f(\mathbf{b}),
		\end{equation}
		where $\boldsymbol{\chi}_e \in \{0, 1\}^{|\Omega_t|}$ is the standard basis vector corresponding to element $e \in \Omega_t$. Such a function $f(\cdot)$ is called a DR-submodular function.
	\end{definition}
	
	\begin{assumption}
		\label{ass:submodular}
		For each round $t$, the function $F_t:2^{\Omega_t}\to\mathbb R_{\ge 0}$ satisfies:
		\begin{enumerate}[label=(\arabic*), leftmargin=*]
			\item \textbf{Normalization:} $F_t(\emptyset)=0$.
			\item \textbf{Monotonicity:} $F_t(A;\,s_t)\le F_t(B;\,s_t)$ for all $A\subseteq B\subseteq \Omega_t$.
			\item \textbf{Submodularity:} For all $A\subseteq B\subseteq \Omega_t$ and $e\in \Omega_t\setminus B$,
			\begin{equation}\label{eq:submod_def}
				F_t(e\mid A;\,s_t) \ge F_t(e\mid B;\,s_t),
			\end{equation}
			where $F_t(e \mid A;\,s_t) \triangleq F_t(A \cup \{e\};\,s_t) - F_t(A;\,s_t)$ is the marginal gain.
			\item \textbf{Bounded marginals:} There exists $B < \infty$ such that $0 \leq F_t(e \mid A;\,s_t) \leq B$ for all $A \in \mathcal{I}_t$, $e \in \Omega_t$, and all $s_t \in \mathcal{S}$.
		\end{enumerate}
	\end{assumption}
	
	\section{Related Work}
	\label{app:related}	
	
	\paragraph{Submodular Optimization and Online Learning.}
	Multilinear-extension-based relaxations have been applied to control problems to obtain scalable approximate solutions with performance guarantees~\citep{kazma2025multilinear}.
    Related approaches have also been developed for observability and sensing problems~\citep{rakhshan2025observability}.
    In multi-agent online coordination, MA-OSMA employs an independent-Bernoulli multilinear extension~\citep{zhang2025near}.
    Since this extension does not enforce the one-action-per-agent constraint, it can generate joint actions that are infeasible under the partition matroid.
    MA-SPL takes a different approach by constructing a policy-based continuous extension from per-agent action probabilities and establishing approximation guarantees for weakly submodular objectives through consensus-based distributed gradient ascent~\citep{zhang2025effective}.
    Although structurally analogous to the PME, MA-SPL focuses on online regret minimization under a fixed agent population.
	
	\paragraph{Submodular Reinforcement Learning and Credit Assignment in MARL.}
	Credit assignment in cooperative MARL commonly uses difference rewards~\citep{wolpert2002collective}, counterfactual baselines~\citep{foerster2018coma}, and monotonic value decomposition~\citep{rashid2020monotonic}.
    These approaches provide different mechanisms for assigning or decomposing team feedback, but do not explicitly exploit submodular structure in team utilities.
    For submodular team rewards, Chen et al. study sequential greedy policy optimization in tabular MDPs, establishing a 1/2-approximation guarantee for known dynamics~\citep{chen2026multi}.
    For unknown dynamics, they further establish sublinear regret using UCB-based methods~\citep{chen2026multi}.
    Their analysis assumes a fixed agent population and fully observable tabular states, with sequential agent ordering rather than simultaneous decentralized execution.
    Abadi et al. discuss how openness and non-stationarity can exacerbate credit misattribution in MARL~\citep{abadi2025challenges}.
    Difference-reward policy gradients have also received formal analysis~\citep{castellini2025difference}, but their connection to continuous relaxations under partition matroid constraints remains unexplored.
	
	\paragraph{Multi-Agent Coordination in Open Environments.}
	Multi-agent problems such as dynamic target tracking and coverage are often formulated as combinatorial task allocation problems with diminishing returns~\citep{zhang2018dynamic,oliva2023sum}.
    Related work also considers active sensing and distributed coordination~\citep{liu2024distributed}.
    Conventional formulations typically assume fixed agent and task sets, relying on centralized greedy or relaxation-based solutions.
    More recent studies address time-varying objectives and online decentralized decisions. For example, constrained coverage and sensing have been investigated by Cervino et al.~\citep{cervino2025constrained}, while Xu et al. and Zhang et al. study online submodular coordination for multi-agent tracking with regret guarantees~\citep{xu2023online,zhang2025near}.
    However, these approaches generally assume a fixed number of agents throughout execution.
    Open multi-agent systems (OMAS), where agents and tasks dynamically arrive and depart, introduce additional challenges~\citep{deplano2026optimization}.
    Conventional MARL algorithms, including MAPPO~\citep{yu2022surprising} and MADDPG~\citep{lowe2017multi}, also assume fixed agent sets and stationary state-action spaces, limiting their direct applicability.
    Recent OMAS research has explored distributed consensus and robustness to population changes~\citep{oliva2023sum,liu2024distributed}.
    Graph neural networks (GNNs) have also been investigated for permutation-invariant representations of variable-sized populations~\citep{jiang2020gcrl,xu2022powerful,lai2025roboballet}.
    Nevertheless, existing GNN-based MARL methods provide limited formal guarantees for combinatorial team objectives.

\begin{table}[t]
    \centering
    \caption{Systematic comparison of submodular optimization and
    multi-agent coordination methods for task allocation.}
    \label{tab:applicability}
    \resizebox{\textwidth}{!}{%
        \renewcommand{\arraystretch}{1.15}
        \setlength{\tabcolsep}{4pt}
        \begin{tabular}{@{}l c c c c l l@{}}
            \toprule
            \textbf{Method} &
            \textbf{\makecell{Required\\Communication}} &
            \textbf{\makecell{Partition\\Feasibility}} &
            \textbf{\makecell{Decentralized\\Execution}} &
            \textbf{\makecell{Open Multi-Agent\\Systems}} &
            \textbf{Guarantee} &
            \textbf{\makecell{Key Limitation}} \\
            \midrule

            Centralized greedy~\citep{fisher1978analysis} &
            Global & \cmark & \xmark & \xmark &
            $1/2$ under general matroids &
            Myopic; requires full state \\

            Continuous greedy + rounding~\citep{calinescu2011maximizing} &
            Global & \cmark & \xmark & \xmark &
            $(1-1/e)$ &
            Centralized relaxation and rounding \\

            MA-OSMA~\citep{zhang2025near} &
            Consensus & \cmark$^{\dag}$ & \cmark & \xmark &
            $(1-1/e)$ &
            High comm. overhead; fixed agents \\

            MA-SPL~\citep{zhang2025effective} &
            Consensus & \cmark$^{\dag}$ & \cmark & \xmark &
            $(1-c/e)$ &
            Weakly submodular; fixed agents \\

            GPO~\citep{chen2026multi} &
            Seq. order$^{\ddag}$ & \cmark & \xmark & \xmark &
            $1/2$ &
            Tabular MDP; fixed agents; sequential \\

            \rowcolor{blue!10}
            \textbf{Ours (SubMAPG)} &
            \textbf{Local} & \textbf{\cmark} & \textbf{\cmark} & \textbf{\cmark} &
            \makecell[cl]{\textbf{Stagewise approximation;}\\
              \textbf{dynamic regret}} &
            \textbf{Surrogate analysis on SubMAPG-$x$ only} \\

            \bottomrule
        \end{tabular}%
    }
    \vspace{2pt}
    {\footnotesize
        $^{\dag}$The continuous relaxation may not match partition-feasible
        categorical execution; feasibility is restored by rounding or
        post-processing.
        \quad
        $^{\ddag}$Agents are optimized sequentially in a fixed order rather
        than by simultaneous decentralized execution.
    }
\end{table}

	\subsection{Detailed Comparison with Existing Methods for Submodular Optimization and Learning}
	\label{app:comparison}
	
	We provide a systematic comparison between the proposed PME and existing approaches for submodular coordination in multi-agent settings. \cref{tab:applicability} summarizes key properties across different methods.

	
	Simple centralized greedy achieves a $1/2$-approximation for monotone submodular maximization under general matroid constraints, while the $(1-1/e)$ guarantee is obtained by continuous greedy together with rounding
	schemes~\citep{fisher1978analysis,calinescu2011maximizing}.
	However, they require global coordination and instantaneous optimization at each decision step, which limits their applicability to sequential settings with learned policies.
	
	MLE-based continuous relaxations preserve the $(1-1/e)$ guarantee via (often centralized) rounding schemes~\citep{calinescu2011maximizing,zhang2025near}. 
	Importantly, while rounding \emph{does} restore feasibility, it is a \emph{post-hoc} correction: the optimized surrogate is defined under an independent Bernoulli model, whereas execution in OMAS uses per-agent categorical decisions constrained by the partition structure. 
	This objective-execution discrepancy is particularly consequential in policy-gradient and sequential settings, where the learning signal is driven by the surrogate rather than the actually executed distribution. 
	Moreover, distributed implementations typically rely on consensus steps to coordinate a shared continuous solution, incurring nontrivial communication overhead.
	
	Standard MARL methods such as MAPPO~\citep{yu2022surprising} and QMIX~\citep{rashid2020monotonic} support decentralized execution and can be extended to handle sequential decisions. However, they provide no formal guarantees for submodular objectives and typically assume fixed agent sets, limiting their direct applicability to open multi-agent systems.
	
	Our PME-based approach bridges this gap by providing a continuous relaxation that natively respects partition matroid constraints. While classical centralized greedy enjoys the stronger $(1-1/e)$ approximation ratio in static fully observable submodular maximization, it requires global state information and sequential one-shot optimization at every round. SubMAPG instead targets a different regime: first-order policy learning with factorized categorical policies, simultaneous decentralized execution, local observations, and open agent/task sets. The resulting $1/2$ factor
	is the standard guarantee for first-order methods on monotone
	DR-submodular functions~\citep{bian2017guaranteed,hassani2017gradient}.
	Thus, the weaker worst-case constant is the price of avoiding centralized rounding or sequential greedy selection, while gaining compatibility with policy-gradient learning and open-system execution. Empirically, SubMAPG can match or exceed myopic greedy baselines in cumulative utility because the learned policy optimizes long-horizon behavior under dynamics, whereas greedy baselines optimize instantaneous stage utilities.
	
	\section{Proofs for \Cref{sec:pme}}   
	\label{app:proofs_sec3}
	
	This appendix provides complete proofs for the theoretical results presented in \cref{sec:pme}.
	
	\subsection{Proof of \cref{lem:gradient}}
	\label{app:proof_gradient}
	
	\textbf{Lemma~\ref{lem:gradient}} (Gradient of PME). 
	\emph{The partial derivative of $\tilde{f}_t$ with respect to $x_{(i,a)}$ is
		given by
		\begin{equation}
			\frac{\partial \tilde{f}_t}{\partial x_{(i,a)}}(\mathbf{x})
			= \mathbb{E}_{A^{-i} \sim \mathcal{D}^{-i}(\mathbf{x})}
			\!\Big[ F_t\big((i,a) \mid A^{-i};\,s_t\big) \Big], \nonumber
		\end{equation}
		where $A^{-i}$ denotes a joint action of all agents except agent $i$, sampled from the product distribution $\mathcal{D}^{-i}(\mathbf{x})$ over $\mathcal{N}_t \setminus \{i\}$.}
	
	\begin{proof}
		We derive the gradient by differentiating the PME definition with respect to $x_{(i,a)}$ and carefully tracking how this variable appears in the probability terms.
		
		Based on \cref{def:pme}, the PME is defined as:
		\begin{equation}
			\tilde{f}_t(\mathbf{x};\,s_t) = \sum_{A \in \mathcal{I}_t} 
			F_t(A;\,s_t) \prod_{k \in \mathcal{N}_t} p_k(A; \mathbf{x}),
			\label{eq:app_pme_def}
		\end{equation}
		where for each agent $k \in \mathcal{N}_t$, the inclusion probability is:
		\begin{equation}
			p_k(A; \mathbf{x}) = 
			\begin{cases}
				x_{(k,a')}, & \text{if } A \cap \Omega_{k,t} = \{(k, a')\} \text{ for some } a' \in \mathcal{A}_{k,t}, \\[4pt]
				1 - \sum_{a'' \in \mathcal{A}_{k,t}} x_{(k,a'')}, & \text{if } A \cap \Omega_{k,t} = \emptyset.
			\end{cases} \nonumber
		\end{equation}
		
		We focus on a specific agent $i$ and action $a$. The term $p_i(A; \mathbf{x})$ depends on $x_{(i,a)}$ in three mutually exclusive cases:
		\begin{enumerate}[label=(\alph*)]
			\item If $A \cap \Omega_{i,t} = \{(i, a)\}$, then $p_i(A; \mathbf{x}) = x_{(i,a)}$. The partial derivative is:
			$ \frac{\partial p_i(A; \mathbf{x})}{\partial x_{(i,a)}} = 1. $
			
			\item If $A \cap \Omega_{i,t} = \{(i, a')\}$ for some $a' \neq a$, then $p_i(A; \mathbf{x}) = x_{(i,a')}$. This term is independent of $x_{(i,a)}$, so:
			$ \frac{\partial p_i(A; \mathbf{x})}{\partial x_{(i,a)}} = 0. $
			
			\item If $A \cap \Omega_{i,t} = \emptyset$, then $p_i(A; \mathbf{x}) = 1 - \sum_{a'' \in \mathcal{A}_{i,t}} x_{(i,a'')}$. Since $x_{(i,a)}$ is included in the sum:
			$ \frac{\partial p_i(A; \mathbf{x})}{\partial x_{(i,a)}} = -1. $
		\end{enumerate}
		
		Crucially, for any other agent $k \neq i$, the probability $p_k(A; \mathbf{x})$ does not depend on agent $i$'s marginals. Applying the product rule to Eq.~\eqref{eq:app_pme_def}, only terms involving agent $i$ contribute non-zero derivatives:
		\begin{align}
			\frac{\partial \tilde{f}_t}{\partial x_{(i,a)}}(\mathbf{x}) 
			&= \sum_{A \in \mathcal{I}_t} F_t(A;\,s_t) \left( \frac{\partial p_i(A; \mathbf{x})}{\partial x_{(i,a)}} \prod_{k \neq i} p_k(A; \mathbf{x}) \right).  \nonumber
		\end{align}
		Splitting the sum based on the cases above (Case (b) vanishes):
		\begin{align} \label{eq:app_grad_split}
			\frac{\partial \tilde{f}_t}{\partial x_{(i,a)}}(\mathbf{x}) 
			&= \underbrace{\sum_{A: (i,a) \in A} F_t(A;\,s_t) \prod_{k \neq i} p_k(A; \mathbf{x})}_{\text{Case (a): } \partial p_i / \partial x_{(i,a)} = +1} 
			- \underbrace{\sum_{A: A \cap \Omega_{i,t} = \emptyset} F_t(A;\,s_t) \prod_{k \neq i} p_k(A; \mathbf{x})}_{\text{Case (c): } \partial p_i / \partial x_{(i,a)} = -1}. 
		\end{align}
		
		Let $A^{-i}$ denote a feasible joint action for all agents except $i$, drawn from the restricted family $\mathcal{I}_t^{-i}$. We observe a one-to-one correspondence:
		\begin{enumerate}
			\item[1)] Any $A$ containing $(i,a)$ can be written as $A^{-i} \cup \{(i,a)\}$.
			\item[2)] Any $A$ where $i$ is idle is effectively just $A^{-i}$.
		\end{enumerate}
		Furthermore, the product term $\prod_{k \neq i} p_k(A; \mathbf{x})$ is exactly the probability of observing $A^{-i}$ under the product distribution $\mathcal{D}^{-i}(\mathbf{x})$, denoted as $\Pr(A^{-i}; \mathbf{x})$.
		
		Substituting these into Eq.~\eqref{eq:app_grad_split}:
		\begin{align}
			\frac{\partial \tilde{f}_t}{\partial x_{(i,a)}}(\mathbf{x}) 
			&= \sum_{A^{-i} \in \mathcal{I}_t^{-i}} F_t(A^{-i} \cup \{(i,a)\};\,s_t) \cdot \Pr(A^{-i}; \mathbf{x}) 
			- \sum_{A^{-i} \in \mathcal{I}_t^{-i}} F_t(A^{-i};\,s_t) \cdot \Pr(A^{-i}; \mathbf{x}) \nonumber \\
			&= \sum_{A^{-i} \in \mathcal{I}_t^{-i}} \Pr(A^{-i}; \mathbf{x}) \cdot \Big[ F_t(A^{-i} \cup \{(i,a)\};\,s_t) - F_t(A^{-i};\,s_t) \Big] \nonumber \\
			&= \mathbb{E}_{A^{-i} \sim \mathcal{D}^{-i}(\mathbf{x})} \Big[ F_t\big((i,a) \mid A^{-i};\,s_t\big) \Big].  \nonumber
		\end{align}
		This completes the proof.
	\end{proof}

	\subsection{Properties of the Partition Multilinear Extension} \label{app:pme_properties}
	
	We state and prove the structural properties of the PME that enable gradient-based optimization.
	
	\begin{theorem}[Properties of PME]\label{thm:properties}
		If $F_t$ satisfies \cref{ass:submodular}, then $\tilde{f}_t: \mathcal{P}(\mathcal{I}_t) \to \mathbb{R}$ satisfies:
		\begin{enumerate}[label=(\arabic*), leftmargin=*]
			\item \textbf{Normalization:} $\tilde{f}_t(\mathbf{0}) = F_t(\emptyset) = 0$.
			\item \textbf{Non-negativity:} $\tilde{f}_t(\mathbf{x};\,s_t) \geq 0$ for all $\mathbf{x} \in \mathcal{P}(\mathcal{I}_t)$.
			\item \textbf{Smoothness:} $\tilde{f}_t$ is a multilinear polynomial and infinitely differentiable.
			\item \textbf{Monotonicity:} $\nabla \tilde{f}_t(\mathbf{x};\,s_t) \geq \mathbf{0}$ for all $\mathbf{x} \in \mathcal{P}(\mathcal{I}_t)$.
			\item \textbf{DR-submodularity:} For any
			$\mathbf{x},\mathbf{y}\in\mathcal{P}(\mathcal{I}_t)$ with $\mathbf{x}\le \mathbf{y}$ componentwise,
			$\nabla \tilde f_t(\mathbf{x};s_t)
			\ge \nabla \tilde f_t(\mathbf{y};s_t)$
			componentwise.
		\end{enumerate}
	\end{theorem}
	
	\begin{proof}
		We prove each property in turn. Properties (1)-(3) follow directly from \cref{def:pme}, while properties (4)-(5) additionally rely on the gradient characterization established in \cref{lem:gradient}.
		
		\paragraph{(1) Normalization.}
		When $\mathbf{x} = \mathbf{0}$, we have $x_{(k,a)} = 0$ for all $(k, a) \in \Omega_t$. By \cref{def:pme}, the probability that agent $k$ contributes to a set $A$ is
		\begin{equation}
			p_k(A; \mathbf{0}) = 
			\begin{cases}
				0, & \text{if } A \cap \Omega_{k,t} = \{(k, a)\} \text{ for some } a, \\[4pt]
				1, & \text{if } A \cap \Omega_{k,t} = \emptyset.
			\end{cases} \nonumber
		\end{equation} 
		Thus $p_k(A; \mathbf{0}) = 1$ if and only if agent $k$ is idle in $A$. For the product $\prod_{k \in \mathcal{N}_t} p_k(A; \mathbf{0})$ to be nonzero, every agent must be idle, which occurs only when $A = \emptyset$. Therefore,
		\begin{equation}
			\tilde{f}_t(\mathbf{0}) = \sum_{A \in \mathcal{I}_t} F_t(A;\,s_t) \prod_{k \in \mathcal{N}_t} p_k(A; \mathbf{0}) = F_t(\emptyset) \cdot 1 = 0, \nonumber
		\end{equation}
		where the last equality uses the normalization property $F_t(\emptyset) = 0$.
		
		\paragraph{(2) Non-negativity.}
		$F_t(A;\,s_t) \geq 0$ for all $A \in \mathcal{I}_t$. Moreover, for any $\mathbf{x} \in \mathcal{P}(\mathcal{I}_t)$, each probability term satisfies $p_k(A; \mathbf{x}) \geq 0$ by definition. Since $\tilde{f}_t(\mathbf{x};\,s_t)$ is a sum of products of nonnegative terms,
		\begin{equation}
			\tilde{f}_t(\mathbf{x};\,s_t) = \sum_{A \in \mathcal{I}_t} \underbrace{F_t(A;\,s_t)}_{\geq 0} \cdot \underbrace{\prod_{k \in \mathcal{N}_t} p_k(A; \mathbf{x})}_{\geq 0} \geq 0. \nonumber
		\end{equation}
		
		\paragraph{(3) Smoothness.}
		By \cref{def:pme}, we have $\tilde{f}_t(\mathbf{x};\,s_t) = \sum_{A \in \mathcal{I}_t} F_t(A;\,s_t) \prod_{k \in \mathcal{N}_t} p_k(A; \mathbf{x})$. For each agent $k$, the term $p_k(A; \mathbf{x})$ is an affine function of the variables $\{x_{(k,a)}\}_{a \in \mathcal{A}_{k,t}}$:
		\begin{enumerate}[label=(\alph*)]
			\item If $A \cap \Omega_{k,t} = \{(k, a)\}$, then $p_k(A; \mathbf{x}) = x_{(k,a)}$ (linear in $x_{(k,a)}$).
			\item If $A \cap \Omega_{k,t} = \emptyset$, then $p_k(A; \mathbf{x}) = 1 - \sum_{a \in \mathcal{A}_{k,t}} x_{(k,a)}$ (affine).
		\end{enumerate}
		Since each $p_k$ depends only on the variables of agent $k$ and different agents have disjoint variable sets, the product $\prod_{k \in \mathcal{N}_t} p_k(A; \mathbf{x})$ is a multilinear polynomial: it has degree at most one in each agent's variables, with total degree at most $|\mathcal{N}_t|$. As a finite sum of such products weighted by constants $F_t(A;\,s_t)$, $\tilde{f}_t(\mathbf{x};\,s_t)$ is a polynomial in $\mathbf{x}$, which is $C^\infty$ (infinitely differentiable) on $\mathbb{R}^{|\Omega_t|}$.
		
		\paragraph{(4) Monotonicity of Gradient.}
		By \cref{lem:gradient}, for any $(i, a) \in \Omega_t$,
		\begin{equation}
			\frac{\partial \tilde{f}_t}{\partial x_{(i,a)}}(\mathbf{x}) = \mathbb{E}_{A^{-i} \sim \mathcal{D}^{-i}(\mathbf{x})} \Big[ F_t\big((i,a) \mid A^{-i}\big) \Big]. \nonumber
		\end{equation}
		By the monotonicity of $F_t$, the marginal gain is nonnegative:
		\begin{equation}
			F_t\big((i,a) \mid A^{-i}\big) = F_t(A^{-i} \cup \{(i,a)\}) - F_t(A^{-i};\,s_t) \geq 0, \quad \forall A^{-i} \in \mathcal{I}_t^{-i}. \nonumber
		\end{equation}
		Since the expectation of a nonnegative random variable is nonnegative, we have $\frac{\partial \tilde{f}_t}{\partial x_{(i,a)}}(\mathbf{x}) \geq 0$ for all $(i, a) \in \Omega_t$, and hence $\nabla \tilde{f}_t(\mathbf{x};\,s_t) \geq \mathbf{0}$ componentwise.
		
		\paragraph{(5) DR-Submodularity.}
		We prove the DR property through the antitonicity of the gradient on
		$\mathcal{P}(\mathcal{I}_t)$. Since $\tilde f_t$ is a multilinear polynomial,
		it is twice continuously differentiable on an open neighborhood of
		$\mathcal{P}(\mathcal{I}_t)$. It therefore suffices to show that all second
		partial derivatives are non-positive on $\mathcal{P}(\mathcal{I}_t)$ and then
		integrate the Hessian along line segments inside the convex polytope.
		
		We verify the second-order condition by considering two cases.
		
		\begin{enumerate}[label=(\alph*)]
			\item \textbf{Same agent ($u = i$).}
			By property (3), $\tilde f_t$ is affine in the variables of each
			individual agent subset. Indeed, in every monomial of the PME, at most one
			factor depends on the variables
			$\{x_{(i,a)}:a\in\mathcal{A}_{i,t}\}$, and this factor is affine.
			Therefore, for any fixed agent $i$, all second-order partial derivatives
			within the same subset vanish:
			\begin{equation}
				\frac{\partial^2 \tilde f_t}
				{\partial x_{(i,v)}\partial x_{(i,a)}}(\mathbf{x};s_t)
				=
				0
				\le 0,
				\qquad
				\forall a,v\in\mathcal{A}_{i,t}.
				\nonumber
			\end{equation}
			
			\item \textbf{Different agents ($u \neq i$).}
			From \cref{lem:gradient}, for any $(i,a)\in\Omega_t$,
			\begin{equation}
				\frac{\partial \tilde f_t}{\partial x_{(i,a)}}(\mathbf{x};s_t)
				=
				\mathbb{E}_{A^{-i}\sim\mathcal{D}^{-i}(\mathbf{x})}
				\left[
				F_t((i,a)\mid A^{-i};s_t)
				\right].
				\nonumber
			\end{equation}
			Let $A^{-\{i,u\}}$ denote a joint action of all agents except $i$ and
			$u$, sampled from $\mathcal{D}^{-\{i,u\}}(\mathbf{x})$. Conditioning on
			the contribution of agent $u$, we can write
			\begin{align}
				\frac{\partial \tilde f_t}{\partial x_{(i,a)}}(\mathbf{x};s_t)
				&=
				\mathbb{E}_{A^{-\{i,u\}}}
				\Bigg[
				\sum_{b\in\mathcal{A}_{u,t}}
				x_{(u,b)}
				F_t((i,a)\mid A^{-\{i,u\}}\cup\{(u,b)\};s_t)
				\nonumber\\
				&\qquad\qquad
				+
				\left(1-\sum_{b\in\mathcal{A}_{u,t}}x_{(u,b)}\right)
				F_t((i,a)\mid A^{-\{i,u\}};s_t)
				\Bigg].
				\nonumber
			\end{align}
			Since $x_{(u,v)}$ appears linearly, differentiating with respect to
			$x_{(u,v)}$ gives
			\begin{equation}
				\label{eq:app_hessian}
				\frac{\partial^2 \tilde f_t}
				{\partial x_{(u,v)}\partial x_{(i,a)}}(\mathbf{x};s_t)
				=
				\mathbb{E}_{A^{-\{i,u\}}}
				\left[
				F_t((i,a)\mid A^{-\{i,u\}}\cup\{(u,v)\};s_t)
				-
				F_t((i,a)\mid A^{-\{i,u\}};s_t)
				\right].
			\end{equation}
			By submodularity of $F_t$, marginal gains diminish as the conditioning
			set grows. Since
			$A^{-\{i,u\}}\subseteq A^{-\{i,u\}}\cup\{(u,v)\}$, we have
			\[
			F_t((i,a)\mid A^{-\{i,u\}}\cup\{(u,v)\};s_t)
			\le
			F_t((i,a)\mid A^{-\{i,u\}};s_t).
			\]
			Thus the integrand in~\eqref{eq:app_hessian} is non-positive. Because
			$\mathbf{x}\in\mathcal{P}(\mathcal{I}_t)$ ensures that
			$\mathcal{D}^{-\{i,u\}}(\mathbf{x})$ is a valid product distribution with
			nonnegative weights, the expectation is also non-positive:
			\[
			\frac{\partial^2 \tilde f_t}
			{\partial x_{(u,v)}\partial x_{(i,a)}}(\mathbf{x};s_t)
			\le 0.
			\]
		\end{enumerate}
		
		Combining the two cases, every second partial derivative of $\tilde f_t$ is
		non-positive on $\mathcal{P}(\mathcal{I}_t)$. Now take any
		$\mathbf{x},\mathbf{y}\in\mathcal{P}(\mathcal{I}_t)$ with
		$\mathbf{x}\le\mathbf{y}$ componentwise. Since
		$\mathcal{P}(\mathcal{I}_t)$ is convex, the segment
		$\mathbf{z}(\tau)=\mathbf{x}+\tau(\mathbf{y}-\mathbf{x})$ lies in
		$\mathcal{P}(\mathcal{I}_t)$ for all $\tau\in[0,1]$. For any coordinate
		$(i,a)$,
		\begin{align}
			\frac{\partial \tilde f_t}{\partial x_{(i,a)}}(\mathbf{y};s_t)
			-
			\frac{\partial \tilde f_t}{\partial x_{(i,a)}}(\mathbf{x};s_t)
			&=
			\int_0^1
			\sum_{(u,v)\in\Omega_t}
			\frac{\partial^2 \tilde f_t}
			{\partial x_{(u,v)}\partial x_{(i,a)}}
			(\mathbf{z}(\tau);s_t)
			\,
			\bigl(y_{(u,v)}-x_{(u,v)}\bigr)
			\,d\tau
			\nonumber\\
			&\le 0,
		\end{align}
		because each Hessian entry is non-positive and
		$y_{(u,v)}-x_{(u,v)}\ge0$. Hence
		$\nabla \tilde f_t(\mathbf{x};s_t)
		\ge
		\nabla \tilde f_t(\mathbf{y};s_t)$
		componentwise. This is exactly the DR-submodularity property on $\mathcal{P}(\mathcal{I}_t)$.
	\end{proof}
	
	\subsection{Proof of \cref{lem:equivalence}}
	\label{app:proof_equivalence}
	
	\textbf{Lemma~\ref{lem:equivalence}} (Objective Equivalence).
	\emph{
		Let $\pi = \prod_{i \in \mathcal{N}_t} \pi^i$ be a
		factorized categorical policy as per~\eqref{eq:factorized_policy} and define the state-dependent marginal vector $\mathbf{x}_t(s_t)$ by
		$x_{(i,a)}(s_t)\triangleq \pi^i(a\mid \mathcal{O}(s_t,i))$. Then, for any fixed state $s_t$, the conditional expected stage utility under $\pi$ equals the PME evaluated at $\mathbf{x}(\pi)$:
		$\mathbb{E}_{A_t \sim \pi}\!\left[
		F_t(A_t;\,s_t)\,\middle|\,s_t
		\right]
		=
		\tilde{f}_t(\mathbf{x}_t(s_t);\,s_t)$.  
		Consequently, by the total expectation,
		$J_t(\pi)
		=
		\mathbb{E}_{s_t\sim\mu_t^\pi}
		\left[
		\tilde f_t(\mathbf{x}_t(s_t);s_t)
		\right]$, where $\mu_t^{\pi}$ is induced by the initial
        distribution $\rho_1$, policy $\pi$, and the
        state transition kernels.
		Note that for categorical policies,
		$\sum_{a} x_{(i,a)}(s_t) = 1$, so the idle mass
		$1 - \sum_{a} x_{(i,a)}$ in \eqref{eq:agent_prob} is exactly zero, ensuring exact equivalence at every state $s_t$.
	}
	
	\begin{proof}
		Fix a round $t$ and a state $s_t$. Conditional on $s_t$, each agent
		$i\in\mathcal{N}_t$ observes $o_{i,t}=\mathcal{O}(s_t,i)$ and samples one
		action independently according to the categorical distribution
		$\pi^i(\cdot\mid o_{i,t})$. Therefore the realized joint action set is
		$A_t=\{(i,a_{i,t}): i\in\mathcal{N}_t\},
		a_{i,t}\sim \pi^i(\cdot\mid \mathcal{O}(s_t,i))$.
		Since each active agent contributes exactly one agent-action pair, we have
		$A_t\in\mathcal{I}_t$ with probability one.
		
		Consider any feasible set $A\in\mathcal{I}_t$. Under the state-dependent
		marginal vector $\mathbf{x}_t(s_t)$, the per-agent factor in the PME is
		\[
		p_i(A;\mathbf{x}_t(s_t))
		=
		\begin{cases}
			x_{(i,a)}(s_t)
			=
			\pi^i(a\mid \mathcal{O}(s_t,i)),
			& \text{if } A\cap\Omega_{i,t}=\{(i,a)\},\\[4pt]
			1-\displaystyle\sum_{a'\in\mathcal{A}_{i,t}}x_{(i,a')}(s_t),
			& \text{if } A\cap\Omega_{i,t}=\emptyset.
		\end{cases}
		\]
		Because $\pi^i(\cdot\mid \mathcal{O}(s_t,i))$ is a categorical distribution,
		$\sum_{a'\in\mathcal{A}_{i,t}}x_{(i,a')}(s_t)=1$,
		so the idle probability is zero for every active agent. Thus any feasible
		set that omits an active agent has zero probability under the categorical
		execution model, while any feasible set selecting exactly one action per
		active agent has probability equal to the product of the corresponding local
		action probabilities.
		
		Using the conditional independence of agents' action samples given $s_t$, the
		probability of realizing $A$ is
		$\Pr(A_t=A\mid s_t)= \prod_{i\in\mathcal{N}_t} p_i(A;\mathbf{x}_t(s_t))$.
		Therefore,
		\begin{align}
			\mathbb{E}_{A_t\sim\pi}
			\!\left[
			F_t(A_t;s_t)\mid s_t
			\right]
			&=
			\sum_{A\in\mathcal{I}_t}
			F_t(A;s_t)\Pr(A_t=A\mid s_t)
			\nonumber\\
			&=
			\sum_{A\in\mathcal{I}_t}
			F_t(A;s_t)
			\prod_{i\in\mathcal{N}_t}p_i(A;\mathbf{x}_t(s_t))
			\nonumber\\
			&=
			\tilde f_t(\mathbf{x}_t(s_t);s_t),
		\end{align}
		where the last equality is exactly the definition of the PME in
		\cref{def:pme}.
		Finally, applying the law of total expectation over
		$s_t\sim\mu_t^\pi$ gives
		\[
		J_t(\pi)
		=
		\mathbb{E}_{\pi}\!\left[F_t(A_t;s_t)\right]
		=
		\mathbb{E}_{s_t\sim\mu_t^\pi}
		\left[
		\mathbb{E}_{A_t\sim\pi}
		\!\left[
		F_t(A_t;s_t)\mid s_t
		\right]
		\right]
		=
		\mathbb{E}_{s_t\sim\mu_t^\pi}
		\left[
		\tilde f_t(\mathbf{x}_t(s_t);s_t)
		\right].
		\]
	\end{proof}
	
	\section{Proofs for \Cref{sec:algorithm}}
	\label{app:proofs_sec4}
	
	This appendix provides complete proofs for all theoretical results in \cref{sec:algorithm}.
	
	\subsection{Masked-Softmax Feasibility and Policy Scope}
	\label{app:parameterization}
	
	\subsubsection{Proof of \cref{lem:masked_softmax}}
	
	\textbf{Lemma~\ref{lem:masked_softmax}} (Masked-Softmax Feasibility).
	For all $\theta\in\mathbb{R}^d$, agents $i\in\mathcal{N}_t$,
	and observations $o_{i,t}\in\mathcal{O}$,
	$\sum_{a\in\mathcal{A}_{i,t}}
	\pi_\theta^i(a\mid o_{i,t})=1$.
	Hence, since $\mathcal{F}_t$ is defined by
	$\sum_{a}x_{(i,a)}=1$, we have $\mathbf{x}(\theta)\in\mathcal{F}_t$,
	and every sampled joint action is feasible under the
	partition matroid.

	\begin{proof}
		By \cref{eq:masked_softmax_main}, the policy of agent $i$ is
		\begin{equation}
			\pi_\theta^i(a \mid o_{i,t})
			= \frac{\exp\!\bigl([\mathbf{W}_h h_{i,t}]_a + M_a\bigr)}
			{\sum_{a' \in \mathcal{A}_{i,t}}
				\exp\!\bigl([\mathbf{W}_h h_{i,t}]_{a'} + M_{a'}\bigr)},
			\quad \forall\, a \in \mathcal{A}_{i,t},
		\end{equation}
		where $M_a = 0$ for all feasible $a \in \mathcal{A}_{i,t}$ and
		$M_a = -\infty$ for infeasible actions (which are excluded from the
		sum by convention). Since the denominator is exactly
		$\sum_{a' \in \mathcal{A}_{i,t}}
		\exp([\mathbf{W}_h h_{i,t}]_{a'} + M_{a'})$,
		the softmax normalization gives
		\begin{equation}
			\sum_{a \in \mathcal{A}_{i,t}} \pi_\theta^i(a \mid o_{i,t})
			= \frac{\sum_{a \in \mathcal{A}_{i,t}}
				\exp([\mathbf{W}_h h_{i,t}]_a + M_a)}
			{\sum_{a' \in \mathcal{A}_{i,t}}
				\exp([\mathbf{W}_h h_{i,t}]_{a'} + M_{a'})}
			= 1,
		\end{equation}
		for all $\theta \in \mathbb{R}^d$, all $i \in \mathcal{N}_t$, and all
		$o_{i,t} \in \mathcal{O}$. Consequently,
		$x_{(i,a)}(\theta) \triangleq \pi_\theta^i(a \mid o_{i,t})$
		satisfies $\sum_{a \in \mathcal{A}_{i,t}} x_{(i,a)}(\theta) = 1$
		for all $i$, which by \cref{eq:boundary_face} means
		$\mathbf{x}(\theta) \in \mathcal{F}_t$.
	\end{proof}
	
	\subsubsection{Proof of \cref{prop:tabular_equivalence}}
	\label{app:proof_tabular_equivalence}
	
	\textbf{Proposition~\ref{prop:tabular_equivalence}}(Tabular Softmax Induces Feasible PME-Based Ascent).		
		Let each agent $i$ use a tabular softmax policy
		$\pi_\theta^i(a\mid o)
		=
		\operatorname{softmax}(\theta^i_{o,\cdot})_a$
		with independent parameters $\theta^i_{o,a}\in\mathbb{R}$. Consider the policy-gradient update
		$\theta_{k+1}
		=
		\theta_k+\eta \hat g_k^{(\theta)}$,
		where $\hat g_k^{(\theta)}$ is an unbiased estimator of
		$\nabla_\theta J_t(\theta_k)$ conditional on $\theta_k$. Let
		$\mathbf{x}_k=\mathbf{x}(\theta_k)$ and
		$\Delta\mathbf{x}_k\triangleq \mathbf{x}_{k+1}-\mathbf{x}_k$. Then:
		\begin{enumerate}[label=(\roman*)]
			\item $\mathbf{x}_k\in\mathcal{F}_t$ for all $k$.
			\item For sufficiently small $\eta$,
			\[
			\mathbb{E}\!\left[
			\left\langle
			\nabla_{\mathbf{x}}\tilde f_t(\mathbf{x}_k;s_t),
			\Delta\mathbf{x}_k
			\right\rangle
			\,\middle|\,\theta_k
			\right]
			=
			\eta
			\left\|
			\mathbf{J}_k^\top
			\nabla_{\mathbf{x}}\tilde f_t(\mathbf{x}_k;s_t)
			\right\|_2^2
			+O(\eta^2)
			\geq -O(\eta^2),
			\]
			where $\mathbf{J}_k=\partial\mathbf{x}(\theta)/\partial\theta|_{\theta=\theta_k}$.
		\end{enumerate}
		Consequently, tabular softmax gradient ascent preserves feasibility on
		$\mathcal{F}_t$ and produces, to first order, a PME-based marginal update.
		The guarantees in Theorems~\ref{thm:stagewise} and~\ref{thm:regret} are stated for the idealized Euclidean projected stochastic-gradient dynamics in $x$-space, rather than for general parameter-space dynamics.

	\begin{proof}
		Part (i) follows from \cref{lem:masked_softmax}: for every
		$\theta$, each local softmax distribution satisfies
		$\sum_{a\in\mathcal{A}_{i,t}}\pi_\theta^i(a\mid o_{i,t})=1$, and hence the induced marginal vector lies in $\mathcal{F}_t$.
		
		For part (ii), by \cref{lem:pg} and the chain rule,
		$\nabla_\theta J_t(\theta)
		=
		\mathbf{J}(\theta)^\top
		\nabla_{\mathbf{x}}\tilde f_t(\mathbf{x}(\theta);s_t)$,
		where
		$\mathbf{J}(\theta)
		=
		{\partial \mathbf{x}(\theta)}/{\partial \theta}$
		is the Jacobian of the tabular softmax mapping.
		For a fixed agent $i$ and observation $o$, the softmax Jacobian is
		$J^i_o
		=
		\mathrm{diag}(\pi^i)
		-
		\pi^i(\pi^i)^\top$.
		This matrix is symmetric positive semidefinite and has range contained in the tangent space of the probability simplex,
		\[
		T_{\Delta(\mathcal{A}_{i,t})}
		=
		\left\{
		\mathbf{z}\in\mathbb{R}^{|\mathcal{A}_{i,t}|}:
		\sum_{a\in\mathcal{A}_{i,t}} z_a=0
		\right\}.
		\]
		It is the covariance matrix of a categorical distribution. In particular, its range property implies that first-order changes of the marginals remain tangent to the product simplex.
		
		Using a first-order Taylor expansion of the smooth mapping $\theta \rightarrow\mathbf{x}(\theta)$ around $\theta_k$, we have
		$\Delta\mathbf{x}_k=\mathbf{J}_k(\theta_{k+1}-\theta_k)+O(\|\theta_{k+1}-\theta_k\|_2^2)=
		\eta\mathbf{J}_k\hat g_k^{(\theta)}
		+O(\eta^2)$,
		where the $O(\eta^2)$ term is uniform on compact parameter sets, or whenever the relevant second derivatives of the softmax map are bounded along the update.
		
		Taking conditional expectation given $\theta_k$ and using
		$\mathbb{E}\!\left[\hat g_k^{(\theta)}\mid\theta_k\right] =
		\nabla_\theta J_t(\theta_k)
		=\mathbf{J}_k^\top
		\nabla_{\mathbf{x}}\tilde f_t(\mathbf{x}_k;s_t)$,
		we obtain
		$\mathbb{E}\!\left[
		\Delta\mathbf{x}_k
		\,\middle|\,\theta_k
		\right]=\eta
		\mathbf{J}_k\mathbf{J}_k^\top
		\nabla_{\mathbf{x}}\tilde f_t(\mathbf{x}_k;s_t)
		+
		O(\eta^2)$.
		Therefore,
		\begin{align}
			&\mathbb{E}\!\left[
			\left\langle
			\nabla_{\mathbf{x}}\tilde f_t(\mathbf{x}_k;s_t),
			\Delta\mathbf{x}_k
			\right\rangle
			\,\middle|\,\theta_k
			\right]
			\nonumber\\
			&\quad =
			\eta
			\nabla_{\mathbf{x}}\tilde f_t(\mathbf{x}_k;s_t)^\top
			\mathbf{J}_k\mathbf{J}_k^\top
			\nabla_{\mathbf{x}}\tilde f_t(\mathbf{x}_k;s_t)
			+
			O(\eta^2)
			\nonumber\\
			&\quad =
			\eta
			\left\|
			\mathbf{J}_k^\top
			\nabla_{\mathbf{x}}\tilde f_t(\mathbf{x}_k;s_t)
			\right\|_2^2
			+
			O(\eta^2)
			\nonumber\\
			&\quad \geq -O(\eta^2),
		\end{align}
		where the second equality uses
		$\mathbf{v}^\top\mathbf{J}_k\mathbf{J}_k^\top\mathbf{v}
		=
		\|\mathbf{J}_k^\top\mathbf{v}\|_2^2$
		
		for any vector $\mathbf{v}$. This proves the claimed first-order PME-based ascent property.
	\end{proof}
	
	\subsection{Unbiased PME Gradient Estimation}
	\label{app:unbiased}

	\subsubsection{Proof of \cref{lem:pg}}
	
	\begin{lemma}[Policy Gradient via Partition Extension]\label{lem:pg}
		Let $x_{(i,a)}(\theta) = \pi_\theta^i(a \mid o_{i,t})$ denote the marginal probabilities induced by the policy $\pi_\theta^i$. The gradient of the expected stage utility $J_t(\theta) \triangleq \mathbb{E}_{A_t \sim \pi_\theta}[F_t(A_t)]$ is
		$\nabla_\theta J_t(\theta) = \nabla_\theta \tilde{f}_t(\mathbf{x}(\theta))$. 
	\end{lemma}
	
	\begin{proof}
		The result follows from applying the chain rule to the composition of the PME with the policy-induced marginals. By \cref{lem:equivalence}, the expected stage utility under the factorized policy $\pi_\theta$ equals the PME evaluated at the induced marginals:
		\begin{equation}
			J_t(\theta) = \mathbb{E}_{A_t \sim \pi_\theta}[F_t(A_t)] = \tilde{f}_t(\mathbf{x}(\theta)), \nonumber
		\end{equation}
		where $\mathbf{x}(\theta) \in \mathcal{P}(\mathcal{I}_t)$ is the vector with components $x_{(i,a)}(\theta) = \pi_\theta^i(a \mid o_{i,t})$ for all $(i, a) \in \Omega_t$.
		
		Since $J_t(\theta) = \tilde{f}_t(\mathbf{x}(\theta))$ is a composition of the differentiable function $\tilde{f}_t: \mathcal{P}(\mathcal{I}_t) \to \mathbb{R}$ (cf.\ \cref{thm:properties}(3)) with the mapping $\theta \rightarrow \mathbf{x}(\theta)$, the multivariate chain rule yields
		\begin{equation}\label{eq:app_chain_rule}
			\nabla_\theta J_t(\theta) = \sum_{(i,a) \in \Omega_t} \frac{\partial \tilde{f}_t}{\partial x_{(i,a)}}(\mathbf{x}(\theta)) \cdot \nabla_\theta x_{(i,a)}(\theta).
		\end{equation}
		Since the ground set decomposes as $\Omega_t = \bigcup_{i \in \mathcal{N}_t} \Omega_{i,t}$ with $\Omega_{i,t} = \{(i, a) : a \in \mathcal{A}_{i,t}\}$, and $x_{(i,a)}(\theta) = \pi_\theta^i(a \mid o_{i,t})$, we can rewrite \eqref{eq:app_chain_rule} as
		\begin{equation}
			\nabla_\theta J_t(\theta) = \sum_{i \in \mathcal{N}_t} \sum_{a \in \mathcal{A}_{i,t}} \frac{\partial \tilde{f}_t}{\partial x_{(i,a)}}(\mathbf{x}(\theta)) \cdot \nabla_\theta \pi_\theta^i(a \mid o_{i,t}). \nonumber
		\end{equation}
		This is precisely $\nabla_\theta \tilde{f}_t(\mathbf{x}(\theta))$, completing the proof.
	\end{proof}
	
	\subsubsection{Proof of \cref{lem:spg}}
	
	\textbf{Lemma~\ref{lem:spg}} (Stagewise Submodular Policy Gradient Estimator).
	\emph{Fix a round $t$ and state $s_t$. Let
		$x_{(i,a)}(\theta)=\pi_\theta^i(a\mid o_{i,t})$ and define the conditional stage objective $J_t(\theta;s_t)
		\triangleq
		\mathbb{E}_{A_t\sim\pi_\theta}
		\!\left[F_t(A_t;s_t)\mid s_t\right]
		=
		\tilde f_t(\mathbf{x}(\theta);s_t)$.
		Then
		$\nabla_\theta J_t(\theta;s_t)
		=
		\mathbb{E}_{A_t\sim\pi_\theta}
		\left[
		\sum_{i\in\mathcal{N}_t}
		\nabla_\theta \log \pi_\theta^i(a_{i,t}\mid o_{i,t})
		\,
		\bigl(F_t((i,a_{i,t})\mid A_t^{-i};s_t)-b_{i,t}
		\bigr)
		\right]$,
		where $b_{i,t}$ is any term independent of $a_{i,t}$ conditional on $s_t$, $o_{i,t}$, and the other agents' sampled actions.}
	
	\begin{proof}
		Fix a round $t$ and condition on the state $s_t$. Then the local
		observations $\{o_{i,t}\}_{i\in\mathcal{N}_t}$ are fixed through the
		observation map, and the only randomness in the conditional stage objective
		comes from the factorized categorical policy
		\[
		\pi_\theta(\mathbf{a}_t\mid\mathbf{o}_t)
		=
		\prod_{i\in\mathcal{N}_t}
		\pi_\theta^i(a_{i,t}\mid o_{i,t}).
		\]
		For notational simplicity, we suppress the conditioning on $s_t$ in
		intermediate expectations.
		
		By the score-function identity,
		\begin{align}
			\nabla_\theta J_t(\theta;s_t)
			&=
			\nabla_\theta
			\sum_{\mathbf{a}_t}
			\pi_\theta(\mathbf{a}_t\mid\mathbf{o}_t)
			F_t(A_t;s_t) \nonumber\\
			&=
			\sum_{\mathbf{a}_t}
			\pi_\theta(\mathbf{a}_t\mid\mathbf{o}_t)
			\nabla_\theta
			\log \pi_\theta(\mathbf{a}_t\mid\mathbf{o}_t)
			F_t(A_t;s_t) \nonumber\\
			&=
			\mathbb{E}_{A_t\sim\pi_\theta}
			\left[
			\nabla_\theta
			\log \pi_\theta(\mathbf{a}_t\mid\mathbf{o}_t)
			F_t(A_t;s_t)
			\right].
			\label{eq:app_stage_score}
		\end{align}
		Using the factorization of the joint policy,
		\[
		\nabla_\theta
		\log \pi_\theta(\mathbf{a}_t\mid\mathbf{o}_t)
		=
		\sum_{i\in\mathcal{N}_t}
		\nabla_\theta
		\log \pi_\theta^i(a_{i,t}\mid o_{i,t}),
		\]
		Thus, ~\eqref{eq:app_stage_score} becomes
		\begin{equation}
			\nabla_\theta J_t(\theta;s_t)
			=
			\mathbb{E}_{A_t\sim\pi_\theta}
			\left[
			\sum_{i\in\mathcal{N}_t}
			\nabla_\theta
			\log \pi_\theta^i(a_{i,t}\mid o_{i,t})
			F_t(A_t;s_t)
			\right].
			\label{eq:app_factorized_stage_score}
		\end{equation}
		
		For each agent $i$, decompose the team utility as
		\begin{equation}
			F_t(A_t;s_t)
			=
			F_t((i,a_{i,t})\mid A_t^{-i};s_t)
			+
			F_t(A_t^{-i};s_t).
			\label{eq:app_stage_decomposition}
		\end{equation}
		We next prove that the second term in~\eqref{eq:app_stage_decomposition}
		does not contribute to the expectation in~\eqref{eq:app_factorized_stage_score}.
		Condition on $s_t$, $o_{i,t}$, and the sampled actions of all agents except
		$i$, equivalently on $A_t^{-i}$. Under this conditioning,
		$F_t(A_t^{-i};s_t)$ is fixed and independent of $a_{i,t}$. Hence
		\begin{align}
			&\mathbb{E}_{a_{i,t}\sim\pi_\theta^i}
			\left[
			\nabla_\theta \log \pi_\theta^i(a_{i,t}\mid o_{i,t})
			F_t(A_t^{-i};s_t)
			\mid s_t,o_{i,t},A_t^{-i}
			\right] \nonumber\\
			&\qquad =
			F_t(A_t^{-i};s_t)
			\sum_{a\in\mathcal{A}_{i,t}}
			\pi_\theta^i(a\mid o_{i,t})
			\nabla_\theta \log \pi_\theta^i(a\mid o_{i,t}) \nonumber\\
			&\qquad =
			F_t(A_t^{-i};s_t)
			\sum_{a\in\mathcal{A}_{i,t}}
			\nabla_\theta \pi_\theta^i(a\mid o_{i,t}) \nonumber\\
			&\qquad =
			F_t(A_t^{-i};s_t)
			\nabla_\theta
			\sum_{a\in\mathcal{A}_{i,t}}
			\pi_\theta^i(a\mid o_{i,t})
			=
			F_t(A_t^{-i};s_t)\nabla_\theta 1
			=
			0.
			\label{eq:app_counterfactual_zero}
		\end{align}
		Taking expectation over $A_t^{-i}$ gives
		\begin{equation}
			\mathbb{E}_{A_t\sim\pi_\theta}
			\left[
			\nabla_\theta \log \pi_\theta^i(a_{i,t}\mid o_{i,t})
			F_t(A_t^{-i};s_t)
			\right]
			=
			0.
			\label{eq:app_counterfactual_zero_uncond}
		\end{equation}
		
		The same argument applies to any baseline $b_{i,t}$ that is independent of
		$a_{i,t}$ conditional on $s_t$, $o_{i,t}$, and $A_t^{-i}$. Indeed,
		\begin{align}
			&\mathbb{E}_{a_{i,t}\sim\pi_\theta^i}
			\left[
			\nabla_\theta \log \pi_\theta^i(a_{i,t}\mid o_{i,t})
			b_{i,t}
			\mid s_t,o_{i,t},A_t^{-i}
			\right] \nonumber\\
			&\qquad =
			b_{i,t}
			\sum_{a\in\mathcal{A}_{i,t}}
			\nabla_\theta \pi_\theta^i(a\mid o_{i,t})
			=
			b_{i,t}\nabla_\theta 1
			=
			0.
			\label{eq:app_baseline_zero}
		\end{align}
		
		Substituting the decomposition~\eqref{eq:app_stage_decomposition}
		into~\eqref{eq:app_factorized_stage_score}, and then applying
		\eqref{eq:app_counterfactual_zero_uncond} and~\eqref{eq:app_baseline_zero},
		we obtain
		\[
		\nabla_\theta J_t(\theta;s_t)
		=
		\mathbb{E}_{A_t\sim\pi_\theta}
		\left[
		\sum_{i\in\mathcal{N}_t}
		\nabla_\theta \log \pi_\theta^i(a_{i,t}\mid o_{i,t})
		\,
		\bigl(
		F_t((i,a_{i,t})\mid A_t^{-i};s_t)-b_{i,t}
		\bigr)
		\right].
		\]
		This proves the stagewise policy-gradient identity.
	\end{proof}
	
	\section{Proofs for \cref{sec:theory}}
	\label{app:proofs_sec5}
	
	This appendix provides complete proofs for all theoretical results in \cref{sec:theory}, beginning with the two preparatory lemmas (Lemmas~\ref{lem:masked_softmax} and~\ref{lem:gradient_properties})
	that underpin the main approximation and regret theorems.
	
	\subsection{Auxiliary Results for the Main Guarantees}
	\label{app:theory_auxiliary}
	
	\begin{lemma}[Idealized Projected Stochastic Gradient Properties]
		\label{lem:gradient_properties}
		Under \cref{ass:submodular}, define the full-information
		difference-reward estimator $\mathbf{g}_t \in \mathbb{R}^{|\Omega_t|}$ by $(\mathbf{g}_t)_{(i,a)}
		\triangleq F_t((i,a)\mid A_t^{-i};\,s_t),
		\qquad \forall (i,a)\in\Omega_t$,
		where $A_t^{-i}\sim \mathcal{D}^{-i}(\mathbf{x}_t)$. Then it  satisfies, for all $t \in [T]$, 
		$\mathbf{x} \in \mathcal{F}_t$, and $s_t \in \mathcal{S}$:
		\begin{enumerate}[label=(\roman*)]
			\item \textbf{Bounded gradient:} 
			$\|\nabla_{\mathbf{x}} \tilde{f}_t(\mathbf{x};\,s_t)\|_2 
			\leq B\sqrt{N_{\max}N^a_{\max}} \triangleq G$;
			\item \textbf{Unbiased estimation:}
			$\mathbb{E}[\mathbf{g}_t \mid \mathbf{x}_t,s_t] 
			= \nabla_{\mathbf{x}} \tilde{f}_t(\mathbf{x}_t;\,s_t)$;
			\item \textbf{Bounded variance:}
			$\mathbb{E}[\|\mathbf{g}_t 
			- \nabla_{\mathbf{x}} \tilde{f}_t(\mathbf{x}_t;\,s_t)\|_2^2 
			\mid \mathbf{x}_t,s_t] 
			\leq N_{\max}N^a_{\max}B^2 \triangleq \sigma^2$.
		\end{enumerate}
	\end{lemma}
	
	\begin{proof}
		We prove each property in turn.
		
		\paragraph{(i) Bounded gradient.}
		By \cref{lem:gradient}, the partial derivative of the PME
		satisfies
		\begin{equation}
			\frac{\partial \tilde{f}_t}{\partial x_{(i,a)}}(\mathbf{x};\,s_t)
			= \mathbb{E}_{A^{-i} \sim \mathcal{D}^{-i}(\mathbf{x})}
			\!\left[F_t\!\left((i,a)\mid A^{-i};\,s_t\right)\right].
		\end{equation}
		By \cref{ass:submodular}(4), the marginal gain satisfies
		$0 \leq F_t(e \mid A;\,s_t) \leq B$ for all $A \in \mathcal{I}_t$,
		$e \in \Omega_t$, and all $s_t \in \mathcal{S}$.
		Hence each partial derivative satisfies
		$0 \leq \frac{\partial \tilde{f}_t}{\partial x_{(i,a)}} \leq B$
		uniformly. Summing over all $(i,a) \in \Omega_t$:
		\begin{equation}
			\|\nabla_{\mathbf{x}}\tilde{f}_t(\mathbf{x};\,s_t)\|_2^2
			= \sum_{i \in \mathcal{N}_t}\sum_{a \in \mathcal{A}_{i,t}}
			\!\left(\frac{\partial \tilde{f}_t}{\partial x_{(i,a)}}\right)^{\!2}
			\leq |\mathcal{N}_t|\cdot|\mathcal{A}_{i,t}|\cdot B^2
			\leq N_{\max}N^a_{\max}B^2,
		\end{equation}
		so $\|\nabla_{\mathbf{x}}\tilde{f}_t(\mathbf{x};\,s_t)\|_2
		\leq B\sqrt{N_{\max}N^a_{\max}} \triangleq G$.
		
		\paragraph{(ii) Unbiased estimation.}
		The estimator is defined as
		$(\mathbf{g}_t)_{(i,a)}=F_t((i,a)\mid A_t^{-i};s_t)$,
		where $A_t^{-i}\sim\mathcal D^{-i}(\mathbf{x}_t)$.
		By \cref{lem:gradient}, for any action $a \in
		\mathcal{A}_{i,t}$,
		\begin{equation}
			\mathbb{E}\!\left[
			(\mathbf{g}_t)_{(i,a)}
			\mid \mathbf{x}_t,s_t
			\right]=
			\mathbb{E}_{A_t^{-i}\sim\mathcal{D}^{-i}(\mathbf{x}_t)}
			\left[
			F_t((i,a)\mid A_t^{-i};\,s_t)
			\right]=
			\frac{\partial \tilde f_t}{\partial x_{(i,a)}}(\mathbf{x}_t;\,s_t).
		\end{equation}
		Since this holds for every coordinate $(i,a)$, we obtain
		$\mathbb{E}[\mathbf{g}_t \mid \mathbf{x}_t,s_t]
		= \nabla_{\mathbf{x}}\tilde{f}_t(\mathbf{x}_t;\,s_t)$.
		
		\paragraph{(iii) Bounded variance.}
		We bound the mean squared error using the identity
		$\|\mathbf{u}\|_2^2 = \sum_{(i,a)} u_{(i,a)}^2$,
		which requires no independence assumption across coordinates:
		\begin{align}
			\mathbb{E}\!\left[\|\mathbf{g}_t
			- \nabla_{\mathbf{x}}\tilde{f}_t(\mathbf{x}_t;\,s_t)\|_2^2
			\mid \mathbf{x}_t,s_t\right]
			&= \sum_{(i,a)\in\Omega_t}
			\mathbb{E}\!\left[\!\left((\mathbf{g}_t)_{(i,a)}
			- \frac{\partial\tilde{f}_t}{\partial x_{(i,a)}}\right)^{\!2}
			\!\mid \mathbf{x}_t,s_t\right] \nonumber \\
			&\leq \sum_{(i,a)\in\Omega_t}
			\mathbb{E}\!\left[(\mathbf{g}_t)_{(i,a)}^2
			\mid \mathbf{x}_t,s_t\right] \nonumber \\
			&\leq \sum_{(i,a)\in\Omega_t} B^2
			\;\leq\; N_{\max}N^a_{\max}B^2
			\;\triangleq\; \sigma^2,
		\end{align}
		where the first inequality uses $\mathrm{Var}[X] \leq \mathbb{E}[X^2]$
		for nonnegative $X$, and the second uses
		$(\mathbf{g}_t)_{(i,a)} \in [0,B]$.
	\end{proof}
	
	\begin{lemma}[Restricted DR First-Order Inequality on the Categorical Face]
		\label{lem:dr_boundary}
		Let $\mathcal{F}_t$ be the categorical-policy face of the partition matroid polytope. Then:
		\begin{enumerate}[label=(\roman*)]
			\item $\mathcal{F}_t$ is a convex polytope and a face of
			$\mathcal{P}(\mathcal{I}_t)$.
			
			\item For any $\mathbf{x},\mathbf{y}\in\mathcal{F}_t$,
			\begin{equation}
				\label{eq:ambient_dr_gap}
				\frac{1}{2}\tilde f_t(\mathbf{y};s_t)
				-
				\tilde f_t(\mathbf{x};s_t)
				\le
				\frac{1}{2}
				\left\langle
				\nabla \tilde f_t(\mathbf{x};s_t),
				\mathbf{y}-\mathbf{x}
				\right\rangle .
			\end{equation}
			
			\item The global maximum of $\tilde f_t$ over
			$\mathcal{P}(\mathcal{I}_t)$ is attained on $\mathcal{F}_t$. Consequently,
			\begin{equation}
				\label{eq:boundary_opt}
				\max_{\mathbf{x}\in\mathcal{F}_t}
				\tilde f_t(\mathbf{x};s_t)
				=
				\max_{\mathbf{x}\in\mathcal{P}(\mathcal{I}_t)}
				\tilde f_t(\mathbf{x};s_t)
				\ge
				\mathrm{OPT}_t(s_t).
			\end{equation}
		\end{enumerate}
	\end{lemma}
	
	\begin{proof}
		\textbf{(i)}
		The polytope $\mathcal{P}(\mathcal{I}_t)$ is defined by the box constraints
		$0\le x_{(i,a)}\le 1$ and the partition constraints
		$\sum_{a\in\mathcal{A}_{i,t}}x_{(i,a)}\le 1$ for all
		$i\in\mathcal{N}_t$. For each agent $i$, the affine hyperplane
		\[
		H_i
		\triangleq
		\left\{
		\mathbf{x}:
		\sum_{a\in\mathcal{A}_{i,t}}x_{(i,a)}=1
		\right\}
		\]
		is a supporting hyperplane of $\mathcal{P}(\mathcal{I}_t)$, because it is the
		tight form of the valid inequality
		$\sum_{a\in\mathcal{A}_{i,t}}x_{(i,a)}\le 1$. Therefore,
		\[
		\mathcal{F}_t
		=
		\mathcal{P}(\mathcal{I}_t)\cap
		\bigcap_{i\in\mathcal{N}_t}H_i
		\]
		is a face of $\mathcal{P}(\mathcal{I}_t)$. Since it is an intersection of a
		polytope with affine hyperplanes, it is also a convex polytope.
		
		\textbf{(ii)}
		We prove the restricted first-order inequality directly on the categorical
		face. This avoids relying on a lattice argument over the full hypercube,
		which may not preserve the partition constraints.
		
		Fix any $\mathbf{x},\mathbf{y}\in\mathcal{F}_t$. Let
		$A\sim\mathcal{D}(\mathbf{x})$ and $Y\sim\mathcal{D}(\mathbf{y})$ be
		independent categorical joint-action sets. Since
		$\mathbf{x},\mathbf{y}\in\mathcal{F}_t$, both $A$ and $Y$ select exactly one
		action from each active agent subset almost surely. Define
		$A=\{e_i:i\in\mathcal{N}_t\},
		Y=\{u_i:i\in\mathcal{N}_t\}$,
		where $e_i,u_i\in\Omega_{i,t}$ are the selected agent-action pairs of agent
		$i$, and $A^{-i}=A\setminus\{e_i\}$.
		
		We first establish the following deterministic inequality for any two such
		feasible sets $A$ and $Y$:
		\begin{equation}
			\label{eq:deterministic_face_gap}
			F_t(Y;s_t)-2F_t(A;s_t)
			\le
			\sum_{i\in\mathcal{N}_t}
			\left[
			F_t(u_i\mid A^{-i};s_t)
			-
			F_t(e_i\mid A^{-i};s_t)
			\right].
		\end{equation}
		
		To prove~\eqref{eq:deterministic_face_gap}, we use two consequences of
		monotonicity and submodularity. First, by monotonicity,
		$F_t(Y;s_t)-F_t(A;s_t)
		\le
		F_t(A\cup Y;s_t)-F_t(A;s_t)$, and add to $A$ only those elements of $Y$ that are not already in $A$ in an
		arbitrary order. By submodularity, the total gain is at most the sum of their
		marginal gains with respect to $A$. For any $i$ with $u_i\notin A$, we have
		$u_i\notin A^{-i}$ and $A^{-i}\subseteq A$, so diminishing returns gives
		$F_t(u_i\mid A;s_t)
		\le
		F_t(u_i\mid A^{-i};s_t)$.
		For any $i$ with $u_i\in A$, the corresponding gain in
		$F_t(A\cup Y;s_t)-F_t(A;s_t)$ is zero and is also bounded by
		$F_t(u_i\mid A^{-i};s_t)\ge 0$. Hence
		\begin{equation}
			\label{eq:add_Y_bound}
			F_t(Y;s_t)-F_t(A;s_t)
			\le
			\sum_{i\in\mathcal{N}_t}
			F_t(u_i\mid A^{-i};s_t).
		\end{equation}
		
		Second, for any ordering of the elements $\{e_i\}_{i\in\mathcal{N}_t}$,
		submodularity gives
		\begin{equation}
			\label{eq:removal_marginal_bound}
			\sum_{i\in\mathcal{N}_t}
			F_t(e_i\mid A^{-i};s_t)
			\le
			F_t(A;s_t).
		\end{equation}
		Let $A=\{e_1,\ldots,e_m\}$ under an arbitrary ordering and define
		$S_{i-1}=\{e_1,\ldots,e_{i-1}\}$. Since $S_{i-1}\subseteq A^{-i}$ and
		$e_i\notin A^{-i}$, diminishing returns implies
		$F_t(e_i\mid A^{-i};s_t)
		\le
		F_t(e_i\mid S_{i-1};s_t)$.
		Summing over $i$ gives
		\[
		\sum_{i=1}^m F_t(e_i\mid A^{-i};s_t)
		\le
		\sum_{i=1}^m F_t(e_i\mid S_{i-1};s_t)
		=
		F_t(A;s_t)-F_t(\emptyset;s_t)
		=
		F_t(A;s_t),
		\]
		where the last equality uses normalization.
		
		Combining~\eqref{eq:add_Y_bound} and~\eqref{eq:removal_marginal_bound}, we obtain
		\begin{align}
			F_t(Y;s_t)-2F_t(A;s_t)
			&=
			\bigl(F_t(Y;s_t)-F_t(A;s_t)\bigr)-F_t(A;s_t)
			\nonumber\\
			&\le
			\sum_{i\in\mathcal{N}_t}F_t(u_i\mid A^{-i};s_t)
			-
			F_t(A;s_t)
			\nonumber\\
			&\le
			\sum_{i\in\mathcal{N}_t}
			\left[
			F_t(u_i\mid A^{-i};s_t)
			-
			F_t(e_i\mid A^{-i};s_t)
			\right],
		\end{align}
		which proves~\eqref{eq:deterministic_face_gap}.
		
		Then take expectations over the independent samples
		$A\sim\mathcal{D}(\mathbf{x})$ and $Y\sim\mathcal{D}(\mathbf{y})$.
		The left-hand side becomes
		$\mathbb{E}[F_t(Y;s_t)]-2\mathbb{E}[F_t(A;s_t)]
		=
		\tilde f_t(\mathbf{y};s_t)-2\tilde f_t(\mathbf{x};s_t)$.
		For the first term on the right-hand side of
		\eqref{eq:deterministic_face_gap}, using the independence of $Y$ and $A$ and
		\cref{lem:gradient}, we have
		\begin{align}
			\mathbb{E}
			\left[
			\sum_{i\in\mathcal{N}_t}
			F_t(u_i\mid A^{-i};s_t)
			\right]
			&=
			\sum_{i\in\mathcal{N}_t}
			\sum_{a\in\mathcal{A}_{i,t}}
			y_{(i,a)}
			\,
			\mathbb{E}_{A^{-i}\sim\mathcal{D}^{-i}(\mathbf{x})}
			\left[
			F_t((i,a)\mid A^{-i};s_t)
			\right]
			\nonumber\\
			&=
			\sum_{i\in\mathcal{N}_t}
			\sum_{a\in\mathcal{A}_{i,t}}
			y_{(i,a)}
			\frac{\partial \tilde f_t}{\partial x_{(i,a)}}(\mathbf{x};s_t).
			\label{eq:y_grad_term}
		\end{align}
		Similarly, because $e_i$ is sampled according to $x_{(i,\cdot)}$ and is
		independent of $A^{-i}$,
		\begin{align}
			\mathbb{E}
			\left[
			\sum_{i\in\mathcal{N}_t}
			F_t(e_i\mid A^{-i};s_t)
			\right]
			&=
			\sum_{i\in\mathcal{N}_t}
			\sum_{a\in\mathcal{A}_{i,t}}
			x_{(i,a)}
			\,
			\mathbb{E}_{A^{-i}\sim\mathcal{D}^{-i}(\mathbf{x})}
			\left[
			F_t((i,a)\mid A^{-i};s_t)
			\right]
			\nonumber\\
			&=
			\sum_{i\in\mathcal{N}_t}
			\sum_{a\in\mathcal{A}_{i,t}}
			x_{(i,a)}
			\frac{\partial \tilde f_t}{\partial x_{(i,a)}}(\mathbf{x};s_t).
			\label{eq:x_grad_term}
		\end{align}
		Taking expectations in~\eqref{eq:deterministic_face_gap} and using
		\eqref{eq:y_grad_term}--\eqref{eq:x_grad_term}, we obtain
		\[
		\tilde f_t(\mathbf{y};s_t)-2\tilde f_t(\mathbf{x};s_t)
		\le
		\left\langle
		\nabla \tilde f_t(\mathbf{x};s_t),
		\mathbf{y}-\mathbf{x}
		\right\rangle .
		\]
		Dividing both sides by $2$ proves~\eqref{eq:ambient_dr_gap}.
		
		\textbf{(iii)}
		Let $\mathbf{x}^* \in \arg\max_{\mathbf{x}\in\mathcal{P}(\mathcal{I}_t)}
		\tilde f_t(\mathbf{x};s_t)$
		be a global maximizer over the full partition matroid polytope. If
		$\mathbf{x}^*\in\mathcal{F}_t$, there is nothing to prove. Otherwise, there
		exists an agent $i$ with positive idle mass
		\[
		\rho_i
		\triangleq
		1-\sum_{a\in\mathcal{A}_{i,t}}x^*_{(i,a)}
		>0.
		\]
		Choose any action $a_i\in\mathcal{A}_{i,t}$ and define $\mathbf{x}'=\mathbf{x}^*+\rho_i\mathbf{e}_{(i,a_i)}$. Then $\mathbf{x}'\in\mathcal{P}(\mathcal{I}_t)$, because the $i$-th partition
		constraint becomes tight and no other subset is changed. By monotonicity of
		$\tilde f_t$ on $\mathcal{P}(\mathcal{I}_t)$,
		$\tilde f_t(\mathbf{x}';s_t)
		\ge
		\tilde f_t(\mathbf{x}^*;s_t)$.
		Repeating this operation for every agent subset with positive idle mass yields
		a point $\bar{\mathbf{x}}\in\mathcal{F}_t$ such that
		$\tilde f_t(\bar{\mathbf{x}};s_t)
		\ge
		\tilde f_t(\mathbf{x}^*;s_t)$.
		Since $\mathbf{x}^*$ is already a global maximizer over
		$\mathcal{P}(\mathcal{I}_t)$ and $\bar{\mathbf{x}}\in\mathcal{P}(\mathcal{I}_t)$,
		we must have equality:
		$\tilde f_t(\bar{\mathbf{x}};s_t)
		=
		\tilde f_t(\mathbf{x}^*;s_t)$.
		Therefore, a global maximizer over $\mathcal{P}(\mathcal{I}_t)$ is attained on
		$\mathcal{F}_t$, and hence
		$\max_{\mathbf{x}\in\mathcal{F}_t}\tilde f_t(\mathbf{x};s_t)
		=
		\max_{\mathbf{x}\in\mathcal{P}(\mathcal{I}_t)}
		\tilde f_t(\mathbf{x};s_t)$.
		
		It remains to relate this continuous optimum to the discrete optimum. Let $A_t^*\in\arg\max_{A\in\mathcal{I}_t}F_t(A;s_t)$. If $A_t^*$ does not select
		an action for some active agent, monotonicity of $F_t$ allows us to add an
		arbitrary action from that agent's action set without decreasing the value.
		Repeating this for all omitted active agents gives a feasible set
		$\bar A_t^*\in\mathcal{I}_t$ that selects exactly one action per active agent
		and satisfies
		$F_t(\bar A_t^*;s_t)\ge F_t(A_t^*;s_t)=\mathrm{OPT}_t(s_t)$.
		The indicator vector $\mathbf{1}_{\bar A_t^*}$ lies in $\mathcal{F}_t$, and
		by the definition of the PME at deterministic vertices,
		$\tilde f_t(\mathbf{1}_{\bar A_t^*};s_t)
		=
		F_t(\bar A_t^*;s_t)$.
		Therefore,
		\[
		\max_{\mathbf{x}\in\mathcal{F}_t}
		\tilde f_t(\mathbf{x};s_t)
		\ge
		\tilde f_t(\mathbf{1}_{\bar A_t^*};s_t)
		=
		F_t(\bar A_t^*;s_t)
		\ge
		\mathrm{OPT}_t(s_t).
		\]
		This proves~\eqref{eq:boundary_opt}.
	\end{proof}
	
	\begin{lemma}[Pointwise-to-Expected DR-submodular Gap]
		\label{lem:pointwise_to_expected}
		Suppose that, for every $s_t\in\mathcal{S}$, the restricted DR-submodular gap inequality in \cref{lem:dr_boundary}(ii) holds on $\mathcal{F}_t$. 
		Let $\mathbf{x}(\cdot),\mathbf{y}(\cdot):\mathcal{S}\to\mathcal{F}_t$ be measurable selectors, and let $\mu$ be any probability measure over $\mathcal{S}$. Then
		\begin{equation}\label{eq:pointwise_dr}
			\frac{1}{2}\tilde{f}_t(\mathbf{y}(s_t);\,s_t)
			- \tilde{f}_t(\mathbf{x}(s_t);\,s_t)
			\;\leq\; 
			\frac{1}{2}
			\left\langle
			\nabla_{\mathbf{x}}\tilde{f}_t(\mathbf{x}(s_t);\,s_t),\,
			\mathbf{y}(s_t) - \mathbf{x}(s_t)
			\right\rangle,
			\qquad \forall\, s_t \in \mathcal{S}.
		\end{equation}
		Consequently,
		\begin{equation}\label{eq:integrated_dr}
			\begin{aligned}   
				\frac{1}{2} &\mathbb{E}_{s_t \sim \mu}
				\!\left[\tilde{f}_t(\mathbf{y}(s_t);\,s_t)\right]
				- \mathbb{E}_{s_t \sim \mu}
				\!\left[\tilde{f}_t(\mathbf{x}(s_t);\,s_t)\right]\\
				&\;\leq\;
				\frac{1}{2}
				\mathbb{E}_{s_t \sim \mu}\!\left[
				\left\langle
				\nabla_{\mathbf{x}}\tilde{f}_t(\mathbf{x}(s_t);\,s_t),\,
				\mathbf{y}(s_t) - \mathbf{x}(s_t)
				\right\rangle
				\right].
			\end{aligned}
		\end{equation}
		In particular, the result applies to $\mu=\mu_t^\pi$, to a state-dependent marginal iterate $\mathbf{x}_t(s_t)$, and to any measurable optimal selector
		$\mathbf{x}_t^*(s_t)\in\arg\max_{\mathbf{z}\in\mathcal{F}_t}
		\tilde f_t(\mathbf{z};s_t)$.
	\end{lemma}
	
	\begin{proof}
		For each fixed $s_t\in\mathcal{S}$, \cref{lem:dr_boundary}(ii) applies to the two feasible points $\mathbf{x}(s_t),\mathbf{y}(s_t)\in\mathcal{F}_t$, yielding \eqref{eq:pointwise_dr}. Taking expectations with respect to any probability measure $\mu$ over $\mathcal{S}$ preserves the inequality and gives \eqref{eq:integrated_dr}.
		The measure $\mu$ may be policy-induced, and the feasible selectors
		$\mathbf{x}(\cdot)$ and $\mathbf{y}(\cdot)$ may depend on the realized state; the inequality is applied pointwise before integration. If an exact measurable maximizer is not selected, the same argument applies to any measurable $\epsilon$-optimal selector, and the result follows by letting $\epsilon \rightarrow 0$.
	\end{proof}

	\subsection{Proof of \cref{thm:stagewise}}
	\label{app:proof_stagewise}
	
	\textbf{Theorem~\ref{thm:stagewise}} (Stagewise Approximation).
	\emph{
		Consider the projected stochastic gradient ascent iterates on $\mathcal{F}_t$: $\mathbf{x}_{k+1}=\Pi_{\mathcal{F}_t}(\mathbf{x}_k+\eta\mathbf{g}_k)$, where $\mathbf{g}_k$ satisfies \cref{lem:gradient_properties}. Let $\eta = D_t/\sqrt{K(G^2+\sigma^2)}$, where $D_t\triangleq       \max_{\mathbf{x},\mathbf{y}\in\mathcal{F}_t}
        \|\mathbf{x}-\mathbf{y}\|_2$
        is the diameter of the stage-$t$ feasible face. Under \cref{ass:submodular} and \cref{lem:dr_boundary},
		\begin{equation}
			\frac{1}{K}\sum_{k=0}^{K-1}
			\mathbb{E}\!\left[\tilde{f}_t(\mathbf{x}_k;\,s_t)\right]
			\geq
			\frac{1}{2}\mathrm{OPT}_t(s_t)
			-			\frac{D_t\sqrt{G^2+\sigma^2}}{2\sqrt{K}}. \tag{9}
		\end{equation}
		Consequently, if $k^*$ is sampled uniformly from $\{0,\ldots,K-1\}$ and $\pi_{k^*}$ is the factorized categorical policy
		with marginal vector $\mathbf{x}_{k^*}$, then
		$\mathbb{E}\!\left[
		\mathbb{E}_{A_t\sim\pi_{k^*}}\![F_t(A_t;\,s_t)]
		\right]
		\geq
		\frac{1}{2}\mathrm{OPT}_t(s_t)
		-
        \frac{D_t\sqrt{G^2+\sigma^2}}
        {2\sqrt{K}}$.
		The outer expectation is over the random choice of $k^*$ and the stochastic gradient iterates.}
	
	\begin{proof}
		By DR-submodularity of $\tilde{f}_t$ on $\mathcal{F}_t$
		(\cref{lem:dr_boundary}(ii)), for any
		$\mathbf{x}, \mathbf{x}^* \in \mathcal{F}_t$:
		\begin{equation}\label{eq:dr_gap}
			\frac{1}{2}\tilde{f}_t(\mathbf{x}^*;\,s_t)
			- \tilde{f}_t(\mathbf{x};\,s_t)
			\;\leq\; 
			\frac{1}{2}
			\left\langle \nabla \tilde{f}_t(\mathbf{x};\,s_t),\,
			\mathbf{x}^* - \mathbf{x} \right\rangle .
		\end{equation}
		Let $\mathbf{x}^* \in \arg\max_{\mathbf{x}\in\mathcal{F}_t}
		\tilde{f}_t(\mathbf{x};\,s_t)$.
		By non-expansiveness of projection, the standard online gradient descent bound gives~\citep{zinkevich2003online}:
		\begin{equation}\label{eq:ogd_bound}
			\sum_{k=0}^{K-1} \langle \mathbf{g}_k,\,
			\mathbf{x}^* - \mathbf{x}_k \rangle
			\;\leq\; \frac{\|\mathbf{x}_0 - \mathbf{x}^*\|_2^2}{2\eta}
			+ \frac{\eta}{2}\sum_{k=0}^{K-1} \|\mathbf{g}_k\|_2^2. \nonumber
		\end{equation}
		Taking expectations, using unbiasedness
		$\mathbb{E}[\mathbf{g}_k \mid \mathbf{x}_k]
		= \nabla\tilde{f}_t(\mathbf{x}_k;\,s_t)$
		(\cref{lem:gradient_properties}(ii)), the diameter bound
		$\|\mathbf{x}_0 - \mathbf{x}^*\|_2^2 \leq D^2$, and the second-moment
		bound $\mathbb{E}[\|\mathbf{g}_k\|_2^2] \leq G^2 + \sigma^2$
		(\cref{lem:gradient_properties}(i,iii)):
		\begin{equation}\label{eq:sum_grad_bound}
			\sum_{k=0}^{K-1} \mathbb{E}\!\left[\langle
			\nabla \tilde{f}_t(\mathbf{x}_k;\,s_t),\,
			\mathbf{x}^* - \mathbf{x}_k \rangle\right]
			\;\leq\; \frac{D^2}{2\eta}
			+ \frac{\eta K(G^2 + \sigma^2)}{2}.
		\end{equation}
		Summing~\eqref{eq:dr_gap} over $k = 0,\ldots,K-1$,
		taking expectations, and applying~\eqref{eq:sum_grad_bound}:
		\begin{equation}
			\frac{K}{2}\tilde{f}_t(\mathbf{x}^*;\,s_t)
			- \sum_{k=0}^{K-1}\mathbb{E}[\tilde{f}_t(\mathbf{x}_k;\,s_t)]
			\;\leq\; 
			\frac{1}{2}
			\left(
			\frac{D^2}{2\eta}
			+ \frac{\eta K(G^2+\sigma^2)}{2}
			\right). \nonumber
		\end{equation}
		
		Dividing by $K$:
		\begin{equation}
			\frac{1}{2}\tilde{f}_t(\mathbf{x}^*;\,s_t)
			- \frac{1}{K}\sum_{k=0}^{K-1}
			\mathbb{E}[\tilde{f}_t(\mathbf{x}_k;\,s_t)]
			\;\leq\; \frac{D^2}{4\eta K}
			+ \frac{\eta(G^2+\sigma^2)}{4}. \nonumber
		\end{equation}
		Setting $\eta = D/\sqrt{K(G^2+\sigma^2)}$ gives
		\[
		\frac{D^2}{4\eta K}
		+
		\frac{\eta(G^2+\sigma^2)}{4}
		=
		\frac{D\sqrt{G^2+\sigma^2}}{2\sqrt{K}}.
		\]
		Since $\tilde{f}_t(\mathbf{x}^*;\,s_t) \geq \mathrm{OPT}_t(s_t)$
		by \cref{lem:dr_boundary}(iii), we obtain
		\eqref{eq:stagewise_bound}.
		Since each $\mathbf{x}_k \in \mathcal{F}_t$ is the marginal vector
		of a factorized categorical policy $\pi_k$,
		\cref{lem:equivalence} gives
		$\mathbb{E}_{A_t\sim\pi_k}[F_t(A_t;\,s_t)]
		= \tilde{f}_t(\mathbf{x}_k;\,s_t)$
		for every $k$.
		Let $k^*$ be drawn uniformly from $\{0,\ldots,K-1\}$ independently of the iterates.
		Then
		\begin{align}
			\mathbb{E}\!\left[
			\mathbb{E}_{A_t\sim\pi_{k^*}}\![F_t(A_t;\,s_t)]
			\right]
			= \frac{1}{K}\sum_{k=0}^{K-1}
			\mathbb{E}\!\left[\tilde{f}_t(\mathbf{x}_k;\,s_t)\right] 
			\geq
			\frac{1}{2}\mathrm{OPT}_t(s_t)
			-
			\frac{D\sqrt{G^2+\sigma^2}}{2\sqrt{K}}, \nonumber
		\end{align}
		completing the proof.
	\end{proof}
	
	\subsection{Proof of \cref{thm:regret}}
	\label{app:proof_regret}
	
	\textbf{Theorem~\ref{thm:regret}} (Dynamic Regret Bound).
	\emph{Suppose $F_t(\cdot;s_t)$ satisfies \cref{ass:submodular}
		for all $t\in[T]$. Then $\tilde f_t(\cdot;s_t)$ is nonnegative, monotone, and
		DR-submodular on the partition matroid polytope
		$\mathcal{P}(\mathcal{I}_t)$. Since
		$\mathcal{F}_t\subseteq \mathcal{P}(\mathcal{I}_t)$, the
		first-order gap inequality in \cref{lem:dr_boundary}(ii)
		applies to all $\mathbf{x},\mathbf{y}\in\mathcal{F}_t$.
		Let $\{\mathbf{x}_t\}_{t=1}^T$ denote the feasible PME marginal sequence evaluated at the realized states and generated by the two-step projected stochastic-gradient dynamics
		$\bar{\mathbf{x}}_{t+1}
		=
		\Pi_{\mathcal{F}_t}
		\!\left(\mathbf{x}_t+\eta\mathbf{g}_t\right)$, and
		$\mathbf{x}_{t+1}
		=
		\Pi_{\mathcal{F}_{t+1}}
		\!\left(\bar{\mathbf{x}}_{t+1}\right)$.
		Here $\mathbf{g}_t$ is a stochastic PME gradient estimator satisfying \cref{lem:gradient_properties}, and $\eta>0$ is constant. Then
		\begin{equation}
			\mathbb{E}\!\left[\mathrm{Regret}_T^{1/2}(\pi)\right]
			\leq
			\frac{D^2}{4\eta}
			+
			\frac{D\mathcal{P}_T}{2\eta}
			+
			\frac{\eta T(G^2+\sigma^2)}{4}. \tag{11}
		\end{equation}
		Moreover, setting $\eta^*
		=
		\sqrt{
			\frac{D(D+2\mathcal{P}_T)}
			{T(G^2+\sigma^2)}
		}$ gives
		\begin{equation}
			\mathbb{E}\!\left[\mathrm{Regret}_T^{1/2}(\pi)\right]
			\leq
			\frac{1}{2}
			\sqrt{D(D+2\mathcal{P}_T)\,T(G^2+\sigma^2)}. \tag{12}
	\end{equation}}
	
	\begin{proof}
		We combine the restricted DR-submodular first-order gap inequality with
		projected stochastic-gradient analysis over time-varying feasible faces.
		The first projection is onto the current feasible face $\mathcal{F}_t$, which
		allows the standard projection inequality to be applied with the current
		comparator $\mathbf{x}_t^*\in\mathcal{F}_t$. The second projection onto
		$\mathcal{F}_{t+1}$ maintains feasibility after the feasible face changes.
		
		We embed all time-varying marginal vectors into the common Euclidean space
		$\mathbb{R}^{d_{\max}}$ with $d_{\max}=N_{\max}N^a_{\max}$ using the
		isometric zero-padding map
		$\iota_t:\mathcal{F}_t\rightarrow\mathbb{R}^{d_{\max}}$.
		The gradient is extended by zeros on padded coordinates, so inner products
		and norms on the active coordinates are preserved.
		Define $\mathbf{x}_t$, $\bar{\mathbf{x}}_{t+1}$, and $\mathbf{x}_t^*$ for their
		embedded images and suppress $\iota_t$ from the notation. All projections,
		distances, and inner products below are understood in this common embedded
		Euclidean space. The embedded diameter is
		$D =
		\max_{t,u\in[T]}
		\max_{\mathbf{x}\in \iota_t(\mathcal{F}_t),\,
			\mathbf{y}\in \iota_u(\mathcal{F}_u)}
		\|\mathbf{x}-\mathbf{y}\|_2$.
		For each realized state $s_t$, let
		$\mathbf{x}_t^*
		\in
		\arg\max_{\mathbf{x}\in\mathcal{F}_t}
		\tilde f_t(\mathbf{x};s_t)$.
		By \cref{lem:dr_boundary}(ii), for
		$\mathbf{x}_t,\mathbf{x}_t^*\in\mathcal{F}_t$,
		\begin{equation}
			\frac{1}{2}\tilde f_t(\mathbf{x}_t^*;s_t)
			-
			\tilde f_t(\mathbf{x}_t;s_t)
			\leq
			\frac{1}{2}
			\left\langle
			\nabla \tilde f_t(\mathbf{x}_t;s_t),
			\mathbf{x}_t^*-\mathbf{x}_t
			\right\rangle. \nonumber
		\end{equation}
		Moreover, \cref{lem:dr_boundary}(iii) gives
		$\tilde f_t(\mathbf{x}_t^*;s_t)\geq \mathrm{OPT}_t(s_t)$.
		Therefore,
		\begin{equation}
			\frac{1}{2}\mathrm{OPT}_t(s_t)
			-
			\tilde f_t(\mathbf{x}_t;s_t)
			\leq
			\frac{1}{2}
			\left\langle
			\nabla \tilde f_t(\mathbf{x}_t;s_t),
			\mathbf{x}_t^*-\mathbf{x}_t
			\right\rangle. \nonumber
		\end{equation}
		Taking expectations over the trajectory randomness and using
		\cref{lem:equivalence}, we obtain
		\begin{equation}
			\label{eq:proof_expected_gap_twostep}
			\frac{1}{2}
			\mathbb{E}_{s_t\sim\mu_t^\pi}
			\!\left[\mathrm{OPT}_t(s_t)\right]
			-
			\mathbb{E}_{\pi}\!\left[F_t(A_t;s_t)\right]
			\leq
			\frac{1}{2}
			\mathbb{E}
			\left[
			\left\langle
			\nabla \tilde f_t(\mathbf{x}_t;s_t),
			\mathbf{x}_t^*-\mathbf{x}_t
			\right\rangle
			\right].
		\end{equation}
		
		Next, consider the first projection step
		$\bar{\mathbf{x}}_{t+1}
		=
		\Pi_{\mathcal{F}_t}
		\left(\mathbf{x}_t+\eta\mathbf{g}_t\right)$.
		Since $\mathbf{x}_t^*\in\mathcal{F}_t$, the Euclidean projection inequality
		gives $\left\|
		\bar{\mathbf{x}}_{t+1}-\mathbf{x}_t^*
		\right\|_2^2
		\leq
		\left\|
		\mathbf{x}_t+\eta\mathbf{g}_t-\mathbf{x}_t^*
		\right\|_2^2$. Expanding the right-hand side and rearranging yields
		\begin{equation}
			\label{eq:proof_ogd_twostep}
			\left\langle
			\mathbf{g}_t,\mathbf{x}_t^*-\mathbf{x}_t
			\right\rangle
			\leq
			\frac{
				\|\mathbf{x}_t-\mathbf{x}_t^*\|_2^2
				-
				\|\bar{\mathbf{x}}_{t+1}-\mathbf{x}_t^*\|_2^2
			}{2\eta}
			+
			\frac{\eta}{2}\|\mathbf{g}_t\|_2^2 .
		\end{equation}
		Conditioning on the history and the realized state $s_t$, the comparator
		$\mathbf{x}_t^*$ is fixed and \cref{lem:gradient_properties} gives
		$\mathbb{E}[\mathbf{g}_t\mid \mathbf{x}_t,s_t]
		=
		\nabla \tilde f_t(\mathbf{x}_t;s_t),
		\qquad
		\mathbb{E}[\|\mathbf{g}_t\|_2^2]\leq G^2+\sigma^2$.
		Taking expectations in~\eqref{eq:proof_ogd_twostep} gives
		\begin{equation}
			\label{eq:proof_expected_ogd_twostep}
			\mathbb{E}
			\left[
			\left\langle
			\nabla \tilde f_t(\mathbf{x}_t;s_t),
			\mathbf{x}_t^*-\mathbf{x}_t
			\right\rangle
			\right]
			\leq
			\frac{
				\mathbb{E}\|\mathbf{x}_t-\mathbf{x}_t^*\|_2^2
				-
				\mathbb{E}\|\bar{\mathbf{x}}_{t+1}-\mathbf{x}_t^*\|_2^2
			}{2\eta}
			+
			\frac{\eta}{2}(G^2+\sigma^2).
		\end{equation}
		
		Combining~\eqref{eq:proof_expected_gap_twostep} and
		\eqref{eq:proof_expected_ogd_twostep}, and summing over
		$t=1,\ldots,T$, gives
		\begin{equation}
			\label{eq:proof_sum_twostep}
			\mathbb{E}\!\left[\mathrm{Regret}_T^{1/2}(\pi)\right]
			\leq
			\frac{1}{4\eta}
			\sum_{t=1}^{T}
			\left(
			\mathbb{E}\|\mathbf{x}_t-\mathbf{x}_t^*\|_2^2
			-
			\mathbb{E}\|\bar{\mathbf{x}}_{t+1}-\mathbf{x}_t^*\|_2^2
			\right)
			+
			\frac{\eta T}{4}(G^2+\sigma^2).
		\end{equation}
		
		It remains to control the telescoping term with the time-varying comparator.
		Define
		$\Phi_t(\mathbf{x})
		\triangleq
		\|\mathbf{x}-\mathbf{x}_t^*\|_2^2$.
		Then
		\begin{align}
			\sum_{t=1}^{T}
			\left(
			\Phi_t(\mathbf{x}_t)
			-
			\Phi_t(\bar{\mathbf{x}}_{t+1})
			\right) =
			\Phi_1(\mathbf{x}_1)
			-
			\Phi_T(\bar{\mathbf{x}}_{T+1})
			+
			\sum_{t=1}^{T-1}
			\left(
			\Phi_{t+1}(\mathbf{x}_{t+1})
			-
			\Phi_t(\bar{\mathbf{x}}_{t+1})
			\right).
			\label{eq:proof_telescope_twostep}
		\end{align}
		Here $\bar{\mathbf{x}}_{T+1}$ is the auxiliary point produced by the final
		current-face projection. Since
		$\Phi_1(\mathbf{x}_1)\leq D^2$ and
		$\Phi_T(\bar{\mathbf{x}}_{T+1})\geq 0$, it remains to bound the drift term.
		
		For each $t<T$, the second projection step is
		$\mathbf{x}_{t+1}
		=
		\Pi_{\mathcal{F}_{t+1}}(\bar{\mathbf{x}}_{t+1})$.
		Since $\mathbf{x}_{t+1}^*\in\mathcal{F}_{t+1}$, non-expansiveness of
		Euclidean projection gives
		$\|\mathbf{x}_{t+1}-\mathbf{x}_{t+1}^*\|_2
		\leq
		\|\bar{\mathbf{x}}_{t+1}-\mathbf{x}_{t+1}^*\|_2$.
		Therefore,
		\begin{align}
			\Phi_{t+1}(\mathbf{x}_{t+1})
			-
			\Phi_t(\bar{\mathbf{x}}_{t+1})
			&=
			\|\mathbf{x}_{t+1}-\mathbf{x}_{t+1}^*\|_2^2
			-
			\|\bar{\mathbf{x}}_{t+1}-\mathbf{x}_t^*\|_2^2
			\nonumber\\
			&\leq
			\|\bar{\mathbf{x}}_{t+1}-\mathbf{x}_{t+1}^*\|_2^2
			-
			\|\bar{\mathbf{x}}_{t+1}-\mathbf{x}_t^*\|_2^2
			\nonumber\\
			&=
			\left\langle
			\mathbf{x}_t^*-\mathbf{x}_{t+1}^*,
			2\bar{\mathbf{x}}_{t+1}
			-\mathbf{x}_{t+1}^*
			-\mathbf{x}_t^*
			\right\rangle
			\nonumber\\
			&\leq
			2D
			\|\mathbf{x}_t^*-\mathbf{x}_{t+1}^*\|_2 .
			\label{eq:proof_drift_twostep}
		\end{align}
		The last inequality uses Cauchy--Schwarz and the embedded diameter bound:
		\[
		\left\|
		2\bar{\mathbf{x}}_{t+1}
		-\mathbf{x}_{t+1}^*
		-\mathbf{x}_t^*
		\right\|_2
		\leq
		\|\bar{\mathbf{x}}_{t+1}-\mathbf{x}_{t+1}^*\|_2
		+
		\|\bar{\mathbf{x}}_{t+1}-\mathbf{x}_t^*\|_2
		\leq 2D .
		\]
		
		Taking expectations and using the definition of the expected path length,
		\[
		\mathcal{P}_T
		=
		\mathbb{E}
		\left[
		\sum_{t=1}^{T-1}
		\|\mathbf{x}_t^*-\mathbf{x}_{t+1}^*\|_2
		\right],
		\]
		we obtain from~\eqref{eq:proof_telescope_twostep} and
		\eqref{eq:proof_drift_twostep}
		\begin{equation}
			\label{eq:proof_telescoped_bound_twostep}
			\mathbb{E}
			\left[
			\sum_{t=1}^{T}
			\left(
			\Phi_t(\mathbf{x}_t)
			-
			\Phi_t(\bar{\mathbf{x}}_{t+1})
			\right)
			\right]
			\leq
			D^2+2D\mathcal{P}_T .
		\end{equation}
		
		Substituting~\eqref{eq:proof_telescoped_bound_twostep} into
		\eqref{eq:proof_sum_twostep} gives
		\begin{align}
			\mathbb{E}\!\left[\mathrm{Regret}_T^{1/2}(\pi)\right]
			&\leq
			\frac{D^2+2D\mathcal{P}_T}{4\eta}
			+
			\frac{\eta T}{4}(G^2+\sigma^2)
			\nonumber\\
			&=
			\frac{D^2}{4\eta}
			+
			\frac{D\mathcal{P}_T}{2\eta}
			+
			\frac{\eta T}{4}(G^2+\sigma^2),\nonumber
		\end{align}
		which proves~\eqref{eq:regret_master}.
		
		Finally, define
		$C_1
		\triangleq
		\frac{D(D+2\mathcal{P}_T)}{4},
		C_2
		\triangleq
		\frac{T(G^2+\sigma^2)}{4}$.
		Then~\eqref{eq:regret_master} becomes
		\[
		\mathbb{E}\!\left[\mathrm{Regret}_T^{1/2}(\pi)\right]
		\leq
		\frac{C_1}{\eta}+C_2\eta .
		\]
		The minimizing constant step size is
		\[
		\eta^*
		=
		\sqrt{\frac{C_1}{C_2}}
		=
		\sqrt{
			\frac{D(D+2\mathcal{P}_T)}
			{T(G^2+\sigma^2)}
		}.
		\]
		Substituting $\eta^*$ gives
		\begin{align}
			\mathbb{E}\!\left[\mathrm{Regret}_T^{1/2}(\pi)\right]
			\leq
			2\sqrt{C_1C_2}
			=
			\frac{1}{2}
			\sqrt{
				D(D+2\mathcal{P}_T)\,T(G^2+\sigma^2)
			}. \nonumber
		\end{align}
		This proves~\eqref{eq:regret_exact} and completes the proof.
	\end{proof}
	
	\subsection{Explicit Parameter Bounds and Refinement of \cref{thm:regret}}\label{app:proof_bounds}
	
	\begin{proposition}[Explicit Bounds]
		Let $N_{\max} = \max_t |\mathcal{N}_t|$, $N^a_{\max} = \max_{t,i} |\mathcal{A}_{i,t}|$, and let $B$ bound the marginal gains. Then $D \leq \sqrt{2 N_{\max}}$, $G \leq B\sqrt{N_{\max} N^a_{\max}}$, and $\sigma \leq B\sqrt{N_{\max} N^a_{\max}}$, yielding
		\begin{equation}
			\mathbb{E}\left[ \mathrm{Regret}_T^{1/2}(\pi) \right] = O\left( B N_{\max} \sqrt{N^a_{\max} T} \cdot (1 + \mathcal{P}_T) \right).
		\end{equation}
	\end{proposition}
	
	\begin{proof}
		We derive explicit bounds for each constant in terms of the problem dimensions.
		
		\paragraph{(1) Diameter bound $D \leq \sqrt{2 N_{\max}}$.}
		The partition matroid polytope decomposes as a Cartesian product over agents:
		\begin{equation}
			\mathcal{P}(\mathcal{I}_t) = \prod_{i \in \mathcal{N}_t} \Delta_i, \quad \text{where} \quad \Delta_i = \left\{ \mathbf{x}_i \in [0,1]^{|\mathcal{A}_{i,t}|} : \sum_{a \in \mathcal{A}_{i,t}} x_{(i,a)} \leq 1 \right\}. \nonumber
		\end{equation}
		The diameter of each simplex $\Delta_i$ is achieved between two vertices, e.g., $\mathbf{e}_a$ and $\mathbf{e}_{a'}$ (unit vectors corresponding to deterministic action choices):
		$\max_{\mathbf{x}_i, \mathbf{y}_i \in \Delta_i} \|\mathbf{x}_i - \mathbf{y}_i\|_2 = \|\mathbf{e}_a - \mathbf{e}_{a'}\|_2 = \sqrt{2}$. 
		Since the overall polytope is a Cartesian product, the squared diameter is the sum of squared per-agent diameters:
		\begin{equation}
			D_t^2 = \max_{\mathbf{x}, \mathbf{y} \in \mathcal{P}(\mathcal{I}_t)} \|\mathbf{x} - \mathbf{y}\|_2^2 = \sum_{i \in \mathcal{N}_t} 2 = 2 |\mathcal{N}_t| \leq 2 N_{\max}. \nonumber
		\end{equation}
		Therefore, $D \triangleq \max_t D_t \leq \sqrt{2 N_{\max}}$.
		
		\paragraph{(2) Gradient bound $G \leq B\sqrt{N_{\max} N^a_{\max}}$.}
		By \cref{lem:gradient}, the partial derivative of the PME satisfies
		\begin{equation}
			\frac{\partial \tilde{f}_t}{\partial x_{(i,a)}}(\mathbf{x};\,s_t) 
			= \mathbb{E}_{A^{-i} \sim \mathcal{D}^{-i}(\mathbf{x})} 
			\left[ F_t\big((i,a) \mid A^{-i};\,s_t\big) \right]. \nonumber
		\end{equation}
		By \cref{ass:submodular}(4), the marginal gain satisfies
		$0 \leq F_t(e \mid A;\,s_t) \leq B$ for all $A \in \mathcal{I}_t$,
		$e \in \Omega_t$, and \emph{all} $s_t \in \mathcal{S}$.
		Therefore, each partial derivative satisfies
		$0 \leq \frac{\partial \tilde{f}_t}{\partial x_{(i,a)}}(\mathbf{x};\,s_t)
		\leq B$ uniformly over $s_t \in \mathcal{S}$.
		The gradient norm is thus bounded uniformly:
		\begin{align}
			\|\nabla \tilde{f}_t(\mathbf{x};\,s_t)\|_2^2 
			&= \sum_{i \in \mathcal{N}_t} \sum_{a \in \mathcal{A}_{i,t}} 
			\left( \frac{\partial \tilde{f}_t}{\partial x_{(i,a)}}
			(\mathbf{x};\,s_t) \right)^2 
			\leq N_{\max} N^a_{\max} B^2,
			\quad \forall\, s_t \in \mathcal{S}. \nonumber
		\end{align}
		Therefore,
		\begin{equation}
			G \;\triangleq\; \max_t \max_{\mathbf{x} \in \mathcal{F}_t}
			\sup_{s_t \in \mathcal{S}}
			\|\nabla \tilde{f}_t(\mathbf{x};\,s_t)\|_2
			\;\leq\; B\sqrt{N_{\max} N^a_{\max}}. \nonumber
		\end{equation}
		
		\paragraph{(3) Variance bound 
			$\sigma \leq B\sqrt{N_{\max} N^a_{\max}}$.}
		By \cref{lem:gradient},
		$(\mathbf{g}_t)_{(i,a)} = F_t((i,a)\mid A_t^{-i};s_t)$
		is an unbiased estimate of
		$\partial \tilde{f}_t / \partial x_{(i,a)}$,
		with $(\mathbf{g}_t)_{(i,a)} \in [0,B]$.
		We bound the mean squared error directly using the 
		identity $\|\mathbf{u}\|^2 = \sum_{(i,a)} u_{(i,a)}^2$, 
		which requires no independence assumption:
		\begin{align}
			\mathbb{E}\!\left[\|\mathbf{g}_t - \nabla\tilde{f}_t
			(\mathbf{x}_t)\|_2^2 \mid \mathbf{x}_t\right]
			&= \sum_{(i,a) \in \Omega_t}
			\mathbb{E}\!\left[\left((\mathbf{g}_t)_{(i,a)} - 
			\frac{\partial \tilde{f}_t}{\partial x_{(i,a)}}
			\right)^{\!2} \mid \mathbf{x}_t\right] 
			\nonumber \\
			&\leq \sum_{(i,a) \in \Omega_t}
			\mathbb{E}\!\left[(\mathbf{g}_t)_{(i,a)}^2 
			\mid \mathbf{x}_t\right] \nonumber \\
			& \leq \sum_{(i,a) \in \Omega_t} B^2
			\leq N_{\max} N^a_{\max} B^2,
			\nonumber
		\end{align}
		where the first inequality uses 
		$\mathrm{Var}[X] \leq \mathbb{E}[X^2]$ for 
		nonnegative $X$, and the second uses the bound 
		$(\mathbf{g}_t)_{(i,a)} \in [0,B]$.
		
		\paragraph{(4) Explicit regret bound.}
		Substituting the bounds $D \leq \sqrt{2N_{\max}}$, 
		$G \leq B\sqrt{N_{\max}N^a_{\max}}$, 
		$\sigma \leq B\sqrt{N_{\max}N^a_{\max}}$
		into \cref{thm:regret} with optimal step size
		$\eta^* = \sqrt{D(D+2\mathcal{P}_T)/[T(G^2+\sigma^2)]}$:
		\begin{align}
			\mathbb{E}\!\left[\mathrm{Regret}_T^{1/2}(\pi)\right]
			&= O\!\left((D+\mathcal{P}_T)\sqrt{T(G^2+\sigma^2)}\right) \nonumber\\
			&= O\!\left(\left(\sqrt{N_{\max}}+\mathcal{P}_T\right)
			\cdot\sqrt{T}\cdot B\sqrt{N_{\max}N^a_{\max}}\right) \nonumber\\
			&= O\!\left(BN_{\max}\sqrt{N^a_{\max}T}
			+ B\mathcal{P}_T\sqrt{N_{\max}N^a_{\max}T}\right) \nonumber\\
			&= O\!\left(BN_{\max}\sqrt{N^a_{\max}T}\cdot(1+\mathcal{P}_T)\right).
			\nonumber
		\end{align}
		This completes the proof.
	\end{proof}
	
	\subsection{Proof of \cref{cor:sublinear}}
	
	\textbf{Corollary~\ref{cor:sublinear}} (Sublinear Regret).
	\emph{If $D$, $G$, and $\sigma$ are uniformly bounded and the path length satisfies $\mathcal{P}_T=o(T)$, then
		$\limsup_{T\to\infty}
		\frac{1}{T}
		\mathbb{E}\!\left[\mathrm{Regret}_T^{1/2}(\pi)\right]
		\le 0$.
		Equivalently, the policy is asymptotically $1/2$-optimal on average: $\liminf_{T\to\infty}\frac{1}{T}\sum_{t=1}^T
		\mathbb{E}_{\pi}\!\left[F_t(A_t;\,s_t)\right]
		\;\geq\; \frac{1}{2}\cdot
		\liminf_{T\to\infty}\frac{1}{T}\sum_{t=1}^T
		\mathbb{E}_{s_t\sim\mu_t^{\pi}}\!\left[\mathrm{OPT}_t(s_t)\right]$.}
	
	\begin{proof}
		By \cref{thm:regret} with the optimized constant step size
		\[
		\eta^*
		=
		\sqrt{
			\frac{D(D+2\mathcal{P}_T)}
			{T(G^2+\sigma^2)}
		},
		\]
		we have
		\begin{equation}
			\mathbb{E}\!\left[\mathrm{Regret}_T^{1/2}(\pi)\right]
			\le
			\frac{1}{2}
			\sqrt{D(D+2\mathcal{P}_T)\,T(G^2+\sigma^2)}. \nonumber
		\end{equation}
		Dividing both sides by $T$ gives
		\begin{equation}
			\frac{1}{T}
			\mathbb{E}\!\left[\mathrm{Regret}_T^{1/2}(\pi)\right]
			\le
			\frac{1}{2}
			\sqrt{
				\frac{D(D+2\mathcal{P}_T)(G^2+\sigma^2)}{T}
			}.\nonumber
		\end{equation}
		Since $D$, $G$, and $\sigma$ are uniformly bounded, there exists a constant
		$C>0$, independent of $T$, such that
		\begin{equation}
			\label{eq:app_cor_average_rate}
			\frac{1}{T}
			\mathbb{E}\!\left[\mathrm{Regret}_T^{1/2}(\pi)\right]
			\le
			C
			\sqrt{
				\frac{1+\mathcal{P}_T}{T}
			}.
		\end{equation}
		The condition $\mathcal{P}_T=o(T)$ implies
		$ \frac{1+\mathcal{P}_T}{T}\to 0$.
		Hence the right-hand side of~\eqref{eq:app_cor_average_rate} converges to zero, and therefore
		\begin{equation}
			\label{eq:app_cor_limsup}
			\limsup_{T\to\infty}
			\frac{1}{T}
			\mathbb{E}\!\left[\mathrm{Regret}_T^{1/2}(\pi)\right]
			\le 0.
		\end{equation}
		This shows that the average dynamic $1/2$-regret has a vanishing upper bound.
		
		It remains to derive the average-performance statement. By \cref{def:regret},
		\begin{equation}
			\mathbb{E}\!\left[\mathrm{Regret}_T^{1/2}(\pi)\right]
			=
			\sum_{t=1}^T
			\left(
			\frac{1}{2}\,
			\mathbb{E}_{s_t\sim\mu_t^\pi}
			\!\left[\mathrm{OPT}_t(s_t)\right]
			-
			\mathbb{E}_{\pi}\!\left[F_t(A_t;s_t)\right]
			\right). \nonumber
		\end{equation}
		Rearranging and dividing by $T$ yields
		\begin{equation}
			\frac{1}{T}\sum_{t=1}^T
			\mathbb{E}_{\pi}\!\left[F_t(A_t;s_t)\right]
			=
			\frac{1}{2T}\sum_{t=1}^T
			\mathbb{E}_{s_t\sim\mu_t^\pi}
			\!\left[\mathrm{OPT}_t(s_t)\right]
			-
			\frac{1}{T}
			\mathbb{E}\!\left[\mathrm{Regret}_T^{1/2}(\pi)\right].\nonumber
		\end{equation}
		Taking $\liminf$ on both sides and using
		$\liminf_T (a_T-b_T)
		\ge
		\liminf_T a_T-\limsup_T b_T$,
		we obtain
		\begin{align}
			&\liminf_{T\to\infty}
			\frac{1}{T}\sum_{t=1}^T
			\mathbb{E}_{\pi}\!\left[F_t(A_t;s_t)\right]
			\nonumber\\
			&\quad\ge
			\frac{1}{2}
			\liminf_{T\to\infty}
			\frac{1}{T}\sum_{t=1}^T
			\mathbb{E}_{s_t\sim\mu_t^\pi}
			\!\left[\mathrm{OPT}_t(s_t)\right]
			-
			\limsup_{T\to\infty}
			\frac{1}{T}
			\mathbb{E}\!\left[\mathrm{Regret}_T^{1/2}(\pi)\right]
			\nonumber\\
			&\quad\ge
			\frac{1}{2}
			\liminf_{T\to\infty}
			\frac{1}{T}\sum_{t=1}^T
			\mathbb{E}_{s_t\sim\mu_t^\pi}
			\!\left[\mathrm{OPT}_t(s_t)\right],  \nonumber
		\end{align}
		where the last inequality follows from~\eqref{eq:app_cor_limsup}. This proves the claimed asymptotic $1/2$-optimality on average.
	\end{proof}
	
	\subsubsection{Path-Length-Independent Step Size}	
		\begin{corollary}[Path-Length-Independent Step Size]
			\label{cor:pt_free}
			Under the assumptions of \cref{thm:regret}, choosing
			$\eta=D/\sqrt{T(G^2+\sigma^2)}$, independently of $\mathcal{P}_T$, yields
			\begin{equation}
				\mathbb{E}\!\left[\mathrm{Regret}^{1/2}_T(\pi)\right]
				\leq
				\frac{D+\mathcal{P}_T}{2}\sqrt{T(G^2+\sigma^2)}.
			\end{equation}
			If $D$, $G$, and $\sigma$ are uniformly bounded, this bound is sublinear when $\mathcal{P}_T=o(\sqrt{T})$.
		\end{corollary}
		
		\begin{proof}
			The result follows by substituting the prescribed step size into~\eqref{eq:regret_master}. Dividing the resulting bound by $T$ establishes the stated sufficient condition.
		\end{proof}
			
		\subsection{Proof of \cref{cor:episodic}}	
		Let $\pi^\star\in\arg\max_{\pi'}J(\pi')$, and recall that $\Delta_T^{\mathrm{occ}}(\pi)$ is defined in \Cref{cor:episodic}. By \cref{lem:equivalence} and the tower property,
        $J(\pi)=\sum_{t=1}^{T}
        \mathbb{E}_{s_t\sim\mu_t^\pi}
        [\widetilde f_t(\mathbf{x}_t(s_t);s_t)]$.
		Since
		$J^\star\leq
		\sum_{t=1}^{T}
		\mathbb{E}_{s_t\sim\mu_t^{\pi^\star}}
		[\mathrm{OPT}_t(s_t)]$,
		\cref{def:regret} yields
		\begin{align}
			\mathbb{E}[J(\pi)]
			&=
			\frac{1}{2}
			\sum_{t=1}^{T}
			\mathbb{E}_{s_t\sim\mu_t^\pi}
			[\mathrm{OPT}_t(s_t)]
			-\mathbb{E}[\mathrm{Regret}_T^{1/2}(\pi)]
			\nonumber\\
			&\geq
			\frac{1}{2}J^\star
			-\frac{1}{2}\mathbb{E}[\Delta_T^{\mathrm{occ}}(\pi)]
			-\mathbb{E}[\mathrm{Regret}_T^{1/2}(\pi)].
		\end{align}
		Applying \cref{thm:regret} with its optimized step size
		gives~\eqref{eq:episodic_bound}.
		Consequently, if $D$, $G$, and $\sigma$ are
		uniformly bounded,
		$\mathcal{P}_T=o(T)$, and
		$\mathbb{E}[\Delta_T^{\mathrm{occ}}(\pi)]=o(T)$, then
		\begin{equation}
			\frac{\mathbb{E}[J(\pi)]}{T}
			\geq
			\frac{J^\star}{2T}-o(1).
		\end{equation}

	\section{Algorithm}
	\label{app:architecture}
	\subsection{Policy Architecture}
	We introduce two policy architectures that implement the PME and satisfy the partition matroid constraint. Both output a factorized categorical policy via a masked softmax, ensuring each agent cannot select more than one action. \textbf{SubMAPG-G} employs a Graph Neural Network (GNN) encoder. It is permutation-invariant and accommodates variable numbers of agents and tasks, making it suitable for open environments. \textbf{SubMAPG-M} employs a Multi-Layer Perceptron (MLP) encoder with a fixed-dimensional input, offering an efficient alternative for scenarios with fixed agent sets. The architectures adhere to two principles derived from our theoretical framework:
	\paragraph{Categorical Policy for Matroid Constraints.}
	The partition matroid constraint requires mutual exclusion within each agent's action set: each agent must select exactly one action. Standard independent Bernoulli relaxations violate this constraint by permitting infeasible multi-action assignments with positive probability. We enforce feasibility via a factorized categorical policy with masked softmax. For agent $i$, the policy over its feasible action set $\mathcal{A}_{i,t}$ is Eq.\eqref{eq:masked_softmax_main}.
	This guarantees that $\sum_{a \in \mathcal{A}_{i,t}} \pi_\theta^i(a \mid o_{i,t}) = 1$, placing the policy on the boundary of the matroid polytope $\mathcal{P}(\mathcal{I}_t)$ where idle probability is zero. This guarantees that the induced marginals $\mathbf{x}(\theta)$ reside within the matroid polytope $\mathcal{P}(\mathcal{I}_t)$, satisfying the feasibility condition in \cref{lem:masked_softmax} and enabling valid PME gradient estimation.
	
	\paragraph{Local Context for Unbiased Gradient Estimation.}
	\cref{lem:gradient} shows that computing the unbiased gradient estimate requires each agent to evaluate its marginal contribution $F_t(A_t) - F_t(A_{t}^{-i})$. We provide this context through two mechanisms. GNN encoder performs $K$ rounds of message passing over a dynamic communication graph $\mathcal{E}_t$, giving each agent a $K$-hop receptive field that captures relevant neighborhood information. MLP encoder directly concatenates features of the $k$ nearest neighboring agents and tasks into a fixed-dimensional observation vector. Both mechanisms supply the necessary context for decentralized credit assignment while respecting communication constraints.
	
	\subsection{Network Architecture}
	\label{subsec:heads_rewards}
	
	We implement the policy $\pi_{\theta}$ and value function $V_{\phi}$ using an actor-critic framework with separate but structurally identical 2-layer encoders, as exemplified by the SubMAPG-G architecture in \autoref{fig:network}.
	
	\begin{figure}[htbp]
		\centering
		\includegraphics[width=0.8\columnwidth]{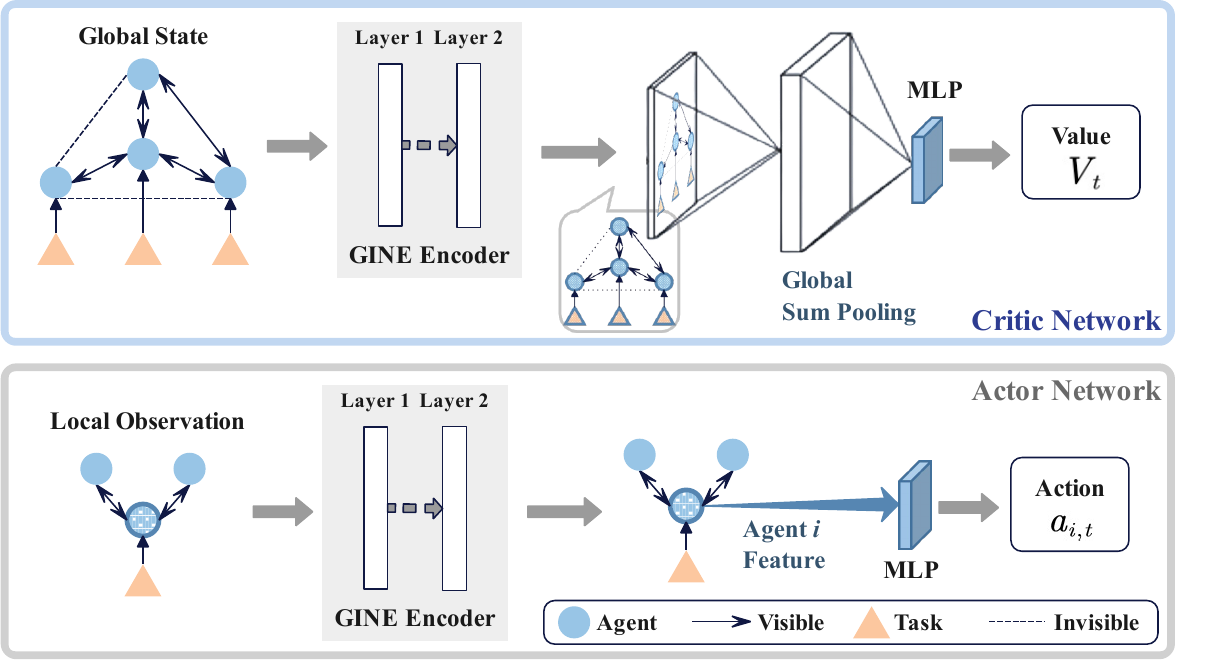}
		\caption{The Actor-Critic architecture of SubMAPG-G. The \textbf{Critic network} aggregates information from the global state graph via sum pooling to estimate the value $V_t$, while the \textbf{Actor network} extracts ego-centric features from the local observation to generate the action $a_{i,t}$.}
		\label{fig:network}
	\end{figure}
	
	\paragraph{Reward.}
	Across all experiments, SubMAPG constructs agent-wise rewards from the
	task-specific submodular utility $F_t$ using the difference reward $F_t\big((i,a_{i,t})\mid A_t^{-i}\big) = F_t(A_t)-F_t(A_t^{-i})$, which corresponds to the selected-action marginal contribution used in the stagewise policy-gradient estimator of \cref{lem:spg}. Its coordinate-wise counterpart is the unbiased PME marginal-gradient estimator in \cref{lem:gradient}. For information coverage, the stagewise utility is the covered information value
	\begin{equation}
		\label{eq:app_cov_utility}
		F_t^{\mathrm{cov}}(A_t)
		=
		\sum_{v\in \cup_{i\in\mathcal{N}_t}D_{i,t}}\rho(v),
	\end{equation}
	where $D_{i,t}$ denotes the set of cells covered by agent $i$ at round $t$
	and $\rho(v)$ is the information density. 
	For an agent-action pair $e = (i,a)$, let $\mathbf{p}_{i,t}$
	denote the predicted next position of agent $i$ after executing
	action $a$ from state $s_t$. Define the per-pair tracking score
	\[
	w_{j,t}(i,a)
	\triangleq
	\max\!\left(
	0,\,
	\frac{r_{\mathrm{sen}} - \|\mathbf{p}_{i,t}^{a} - \mathbf{q}_{j,t}\|}
	{r_{\mathrm{sen}}}
	\right).
	\]
	The tracking utility is
	\begin{equation}
		\label{eq:app_track_utility}
		F_t^{\mathrm{trk}}(A_t)
		=
		\frac{1}{|\mathcal{N}_t||\mathcal{M}_t|}
		\sum_{j \in \mathcal{M}_t}
		\max_{(i,a) \in A_t}
		w_{j,t}(i,a),
	\end{equation}
	with the convention $\max_{\emptyset} = 0$.
	Here $\mathcal{M}_t$ is the active target set and $\mathbf{q}_{j,t}$ denote the positions of target $j$.

	\paragraph{Actor and Critic Networks.}
	The actor processes the encoded representation $h_{i,t}$ to produce action logits, followed by the masked softmax~\eqref{eq:masked_softmax_main} over the feasible action set $\mathcal{A}_{i,t}$, yielding a categorical distribution that strictly enforces the partition matroid constraint. The critic estimates a centralized state value $V_t$ from aggregated agent embeddings via sum pooling.
	\paragraph{SubMAPG-G.} 
	The actor network for each agent $i$ operates on a local observation graph $\mathcal{G}_{i,t}$ that includes the ego agent, neighboring agents within communication radius $r_{\text{com}}$, and targets within sensing radius $r_{\text{sen}}$. Edges are bidirectional among agents and directed from observable targets to the ego agent. The critic network uses a fully connected global graph $\mathcal{G}_t$ containing all active agents and targets to ensure global state awareness. Both networks employ a 2-layer Graph Isomorphism Network with Edge features (GINE)~\citep{hu2020strategies} as the encoder.
	\paragraph{SubMAPG-M.} 
	For the actor, the input is a fixed-dimensional vector formed by concatenating relative features of the $k$ nearest agents (within $r_{\text{com}}$) and $k$ nearest targets (within $r_{\text{sen}}$). Missing entries are zero-padded to maintain consistent dimensionality. For the critic, the input is a global state vector containing absolute positions of all agents and targets. The encoder processes these inputs into latent representations for the subsequent actor and critic heads, which are structurally identical to those in the GNN variant. This design provides an alternative for environments with stable dimensions.
	
	\begin{remark}[Policy implementations and scope of theoretical results]
		\label{rem:nn_gap}
		\cref{prop:tabular_equivalence} establishes that tabular softmax policies preserve feasibility on $\mathcal{F}_t$ and induce PME marginal-ascent directions. SubMAPG-x directly implements the projected stochastic-gradient dynamics analyzed in Theorems~\ref{thm:stagewise} and~\ref{thm:regret}, as shown in \cref{alg:submapg-x}.
		SubMAPG-M and SubMAPG-G parameterize the PME marginals through MLP and GNN policies, respectively, enabling reusable observation-conditioned policies, with the GNN accommodating variable-size agent and task sets. Their marginal updates depend on the policy Jacobian $\mathbf{J}$ associated with $\theta\mapsto\mathbf{x}(\theta)$ and may differ from the Euclidean projected-gradient steps analyzed in Theorems~\ref{thm:stagewise}
		and~\ref{thm:regret}. Characterizing this discrepancy requires architecture-specific assumptions on $\mathbf{J}$ and is beyond the scope of this work. Accordingly, the theoretical guarantees apply directly to SubMAPG-x and provide a marginal-space foundation for SubMAPG-M and SubMAPG-G. Empirically, the neural variants achieve performance consistent with, and often exceeding, the conservative $1/2$ benchmark suggested by the marginal-space analysis.
	\end{remark}

	\section{Simulation Details and Additional Experiments}
	\label{app:experiments}
	
	This appendix provides implementation details and additional experimental results supporting the main findings in \cref{sec:experiments}.
	
	\subsection{Multi-Agent Information Coverage and Optimality Analysis}
	\label{app:optimal}
	
	\paragraph{Coverage Environment Setup.}
	The environment consists of a discrete $30 \times 30$ grid populated by up to $N=5$ agents operating over a training horizon of $T=100$. 
	At each step, agents jointly maximize the covered information value
	$F_t^{\mathrm{cov}}(A_t)=\sum_{v\in\cup_{i\in\mathcal{N}_t}D_{i,t}}\rho(v)$, as defined in~\eqref{eq:app_cov_utility}, where $D_{i,t}$ is the set of cells
	covered by agent $i$ and $\rho(v)\ge 0$ is the information density at cell $v$. This is a weighted coverage function with nonnegative weights, and is therefore normalized, monotone, and submodular in the selected agent-action set. Its marginal gains are uniformly bounded by $\sum_v \rho(v)$, so \cref{ass:submodular} holds for the coverage experiments.
	Following the setup~\citep{prajapat2024submodular}, the density field $\rho$ is generated from one of three distributions: Uniform, Bimodal Gaussian, or a Gaussian Process (GP). Each agent possesses a coverage radius $r_{\mathrm{cov}}=1$ and a communication range $r_{\mathrm{com}}=2$. Agents select actions from $\{\mathit{stay}, \mathit{right}, \mathit{up}, \mathit{left}, \mathit{down}\}$ and start from a random $5\times 5$ clustered region. All learning-based methods are trained in closed environments where all agents remain active throughout each 100-step episode. Evaluation is conducted in both closed and open environments over a horizon of $2000$ steps; in the open setting, two agents remain active for the full episode, while the other three agents randomly enter during the first half of the episode and leave after surviving for at least $20$ steps.

	\paragraph{Comparison with Multi-Agent Submodular RL.}
	In the absence of existing submodular RL methods directly designed for multi-agent submodular coverage, we construct a baseline by extending the single-agent SubPO~\citep{prajapat2024submodular} to the multi-agent setting, termed MASubPO. The multi-agent extension inherits the temporal marginal gain of SubPO, $r_t=F_t^{\mathrm{cov}}(A_t)-F_{t-1}^{\mathrm{cov}}(A_{t-1})$, and treats $r_t$ as a global team reward shared by all agents. In contrast, SubMAPG-M uses an agent-wise marginal contribution reward $F_t^{\mathrm{cov}}(A_t)-F_t^{\mathrm{cov}}(A_t^{-i})$, where $A_t^{-i}$ denotes the joint action set with agent $i$ removed. As shown in \autoref{fig:cov_train}, SubMAPG-M converges faster than MASubPO across all density fields. Moreover, \autoref{fig:cov_eval_close} and \autoref{fig:cov_eval_open} show that SubMAPG-M consistently outperforms MASubPO in both closed and open evaluations. These results indicate that
	agent-wise credit attribution provides a more effective learning signal for multi-agent coordination than the temporal team-level progress metric used by
	MASubPO.

	\begin{figure}[htbp]
		\centering
		\resizebox{0.95\textwidth}{!}{%
			\begin{minipage}{\textwidth}
				\begin{subfigure}{0.32\textwidth}
					\includegraphics[width=\linewidth]{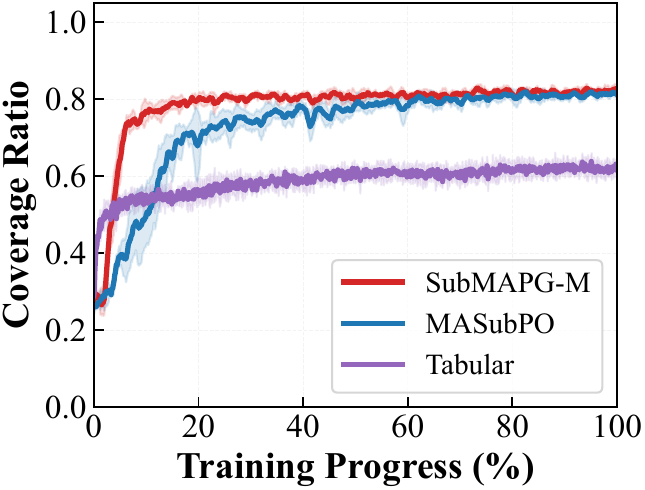}
					\caption{Uniform}
					\label{fig:cov_train_constant}
				\end{subfigure}
				\hfill
				\begin{subfigure}{0.32\textwidth}
					\includegraphics[width=\linewidth]{figures/cov_train/gauss.pdf}
					\caption{Gaussian Process}
					\label{fig:cov_train_gp}
				\end{subfigure}
				\hfill
				\begin{subfigure}{0.32\textwidth}
					\includegraphics[width=\linewidth]{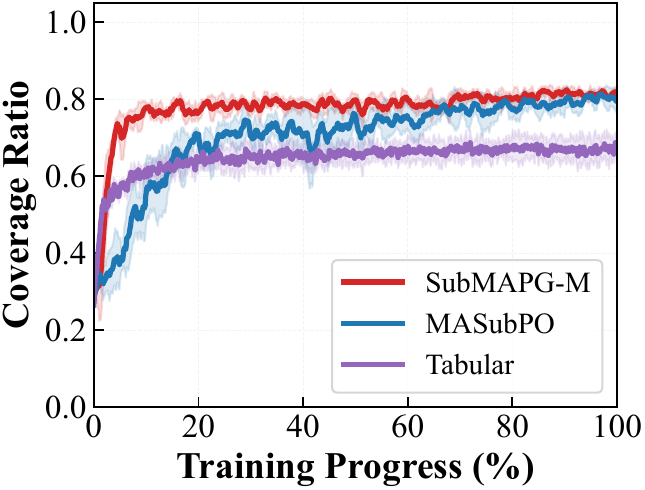}
					\caption{Bimodal Gaussian}
					\label{fig:cov_train_bimodal}
				\end{subfigure}
			\end{minipage}%
		}
		\caption{Training performance in the closed information coverage environment under
			three density fields. SubMAPG-M converges faster than MASubPO across all fields; Tabular is included as a tabular softmax baseline for \cref{prop:tabular_equivalence}. Results are averaged over 5 seeds; shaded regions indicate standard deviation.}
		\label{fig:cov_train}
	\end{figure}

	\paragraph{Comparison with Additional Baselines.}
	We further compare with MA-SPL, a policy-based continuous extension~\citep{zhang2025effective}; MA-OSEA, an MLE-based independent-Bernoulli method~\citep{zhang2025near}; and MAPPO, a standard MARL baseline~\citep{yu2022surprising}. Across closed and open evaluations, \autoref{fig:cov_eval_close}(a)--(f) and \autoref{fig:cov_eval_open}(a)--(f) show that SubMAPG-M and SubMAPG-x maintain high coverage and substantially outperform MAPPO and MASubPO in both closed and open environments. Although MA-SPL and MA-OSEA also achieve high coverage, SubMAPG-M generally attains smaller average utility gaps relative to OSG, while SubMAPG-x also substantially reduces the gap compared with MAPPO and MASubPO. These results demonstrate the effectiveness of the proposed SubMAPG framework in both closed and open environments.

	\paragraph{Optimality Gap and Theoretical Bound.}
	We evaluate the learned policy against greedy methods with known approximation behavior. Since the dynamic coverage problem does not admit a tractable exact optimal oracle, we use OSG~\citep{xu2023online}, an online submodular greedy method with a $1/2$-type guarantee, as the regret reference. We also report CSG-g, a centralized sequential greedy method with global information~\citep{nemhauser1978analysis,qu2019distributed}, as a strong global-information benchmark. SubMAPG-x directly implements the projected PME marginal-space updates analyzed in \cref{sec:theory}, while SubMAPG-M provides the neural network implementation. 
	As shown in \autoref{fig:cov_eval_close}(a)-(c), SubMAPG-M achieves stable coverage above OSG in the closed environment and approaches CSG-g under the Gaussian Process field. SubMAPG-x also achieves competitive coverage across the three density fields. For SubMAPG-M, the remaining gap to CSG-g under the other density fields is expected, as CSG-g uses global information whereas SubMAPG-M relies only on local communication and observations.
	In the open environment, \autoref{fig:cov_eval_open}(a)-(c) show that SubMAPG-M and SubMAPG-x remain robust under agent arrivals and departures. With a time-varying active-agent set, the one-step greedy decisions of CSG-g become less stable, and OSG even outperforms CSG-g in the Uniform and Gaussian Process fields. 
	In \autoref{fig:cov_eval_open}(d)-(f), SubMAPG-M maintains a lower
	average utility gap with respect to OSG or converges toward zero.
	The reported utility gap to OSG is an empirical benchmark for online coordination, rather than a direct measure of the dynamic regret in \cref{def:regret}, which is defined relative to the stagewise optimum $\mathrm{OPT}_t(s_t)$. Evaluating the regret requires the exact discrete optimum, while its theoretical upper bound additionally depends on the optimal PME path length $\mathcal{P}_T$. Computing these quantities requires solving the corresponding optimization problems across realized open-system states, which is intractable at the reported scale.
	The open-coverage protocol matches the slowly varying regime considered by \cref{cor:sublinear}: agents enter only during the
	first half of the episode and remain active for at least $20$ steps before departure. These experiments evaluate finite-horizon performance under gradual population changes, complementing the corollary's asymptotic guarantee.
	
	\begin{figure}[htbp]
		\centering
		\resizebox{0.95\textwidth}{!}{%
			\begin{minipage}{\textwidth}
				\begin{subfigure}{0.32\textwidth}
					\includegraphics[width=\linewidth]{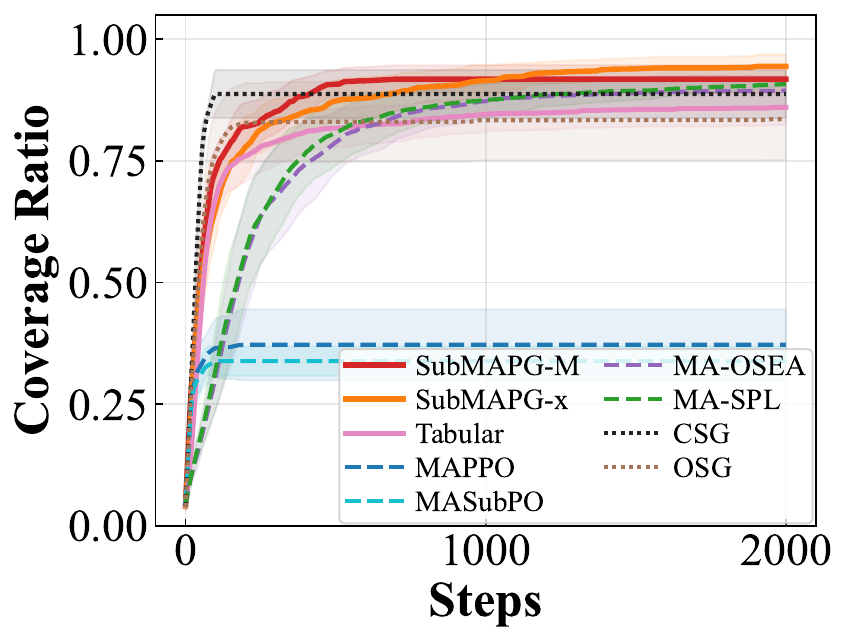}
					\caption{Uniform}
					\label{fig:cov_close_constant_ratio}
				\end{subfigure}
				\hfill
				\begin{subfigure}{0.32\textwidth}
					\includegraphics[width=\linewidth]{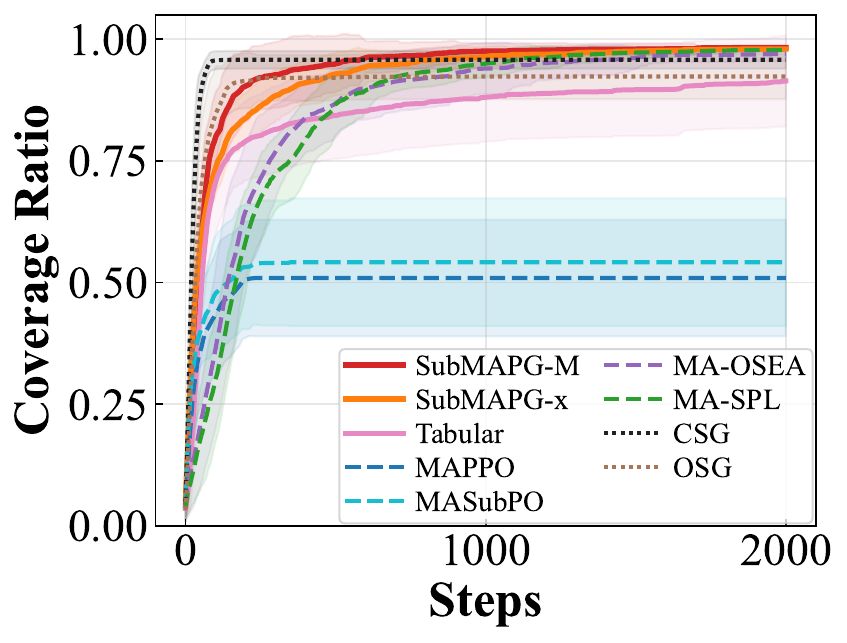}
					\caption{Gaussian Process}
					\label{fig:cov_close_gp_ratio}
				\end{subfigure}
				\hfill
				\begin{subfigure}{0.32\textwidth}
					\includegraphics[width=\linewidth]{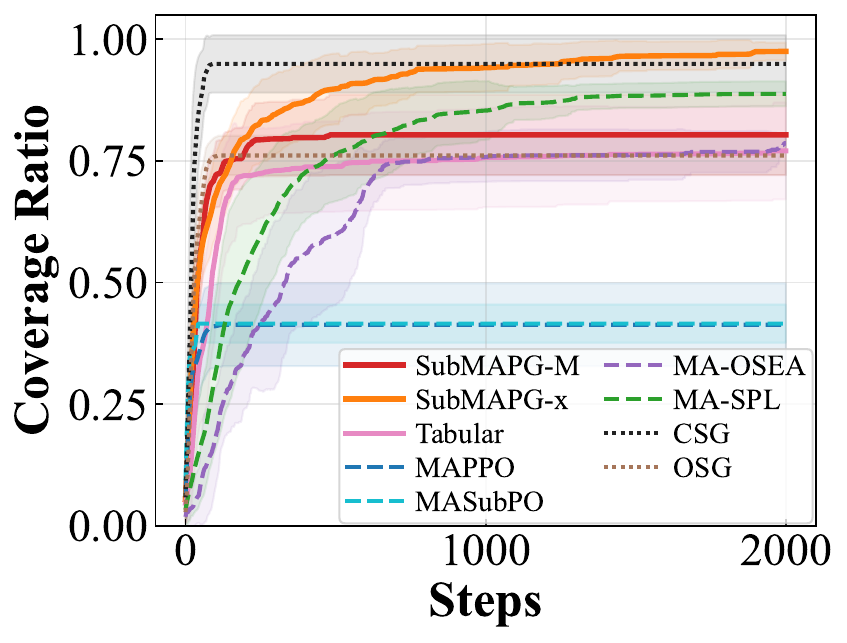}
					\caption{Bimodal Gaussian}
					\label{fig:cov_close_bimodal_ratio}
				\end{subfigure}
				
				\vspace{0.1cm}
				
				\begin{subfigure}{0.32\textwidth}
					\includegraphics[width=\linewidth]{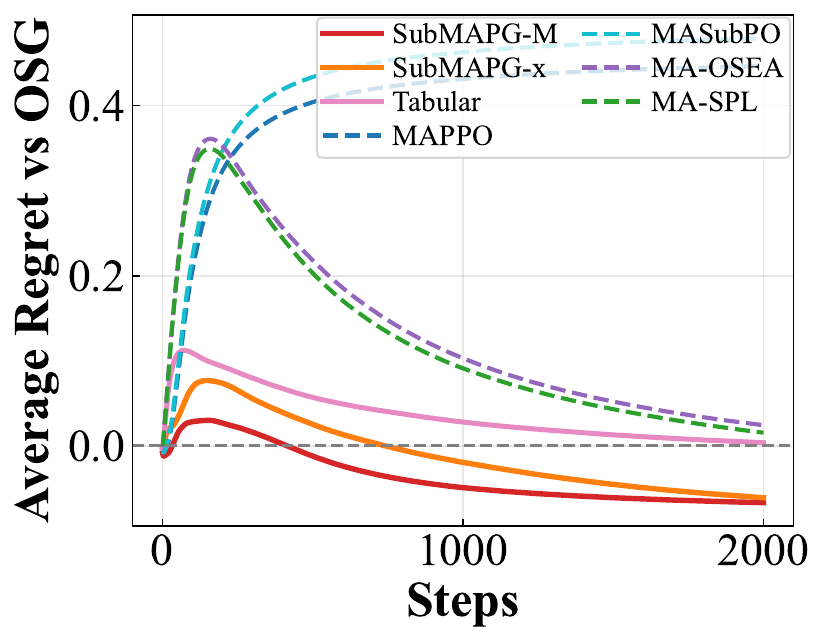}
					\caption{Uniform}
					\label{fig:cov_close_constant_regret}
				\end{subfigure}
				\hfill
				\begin{subfigure}{0.32\textwidth}
					\includegraphics[width=\linewidth]{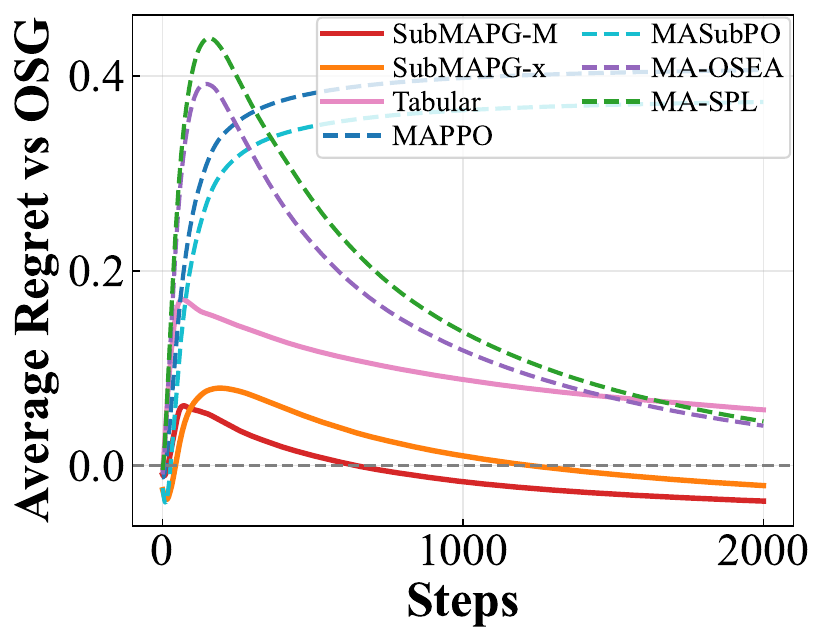}
					\caption{Gaussian Process}
					\label{fig:cov_close_gp_regret}
				\end{subfigure}
				\hfill
				\begin{subfigure}{0.32\textwidth}
					\includegraphics[width=\linewidth]{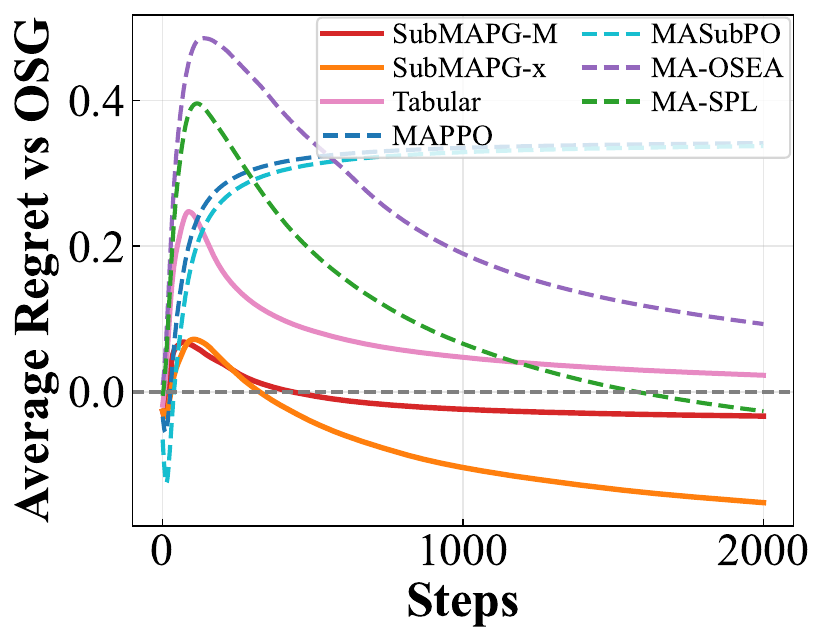}
					\caption{Bimodal Gaussian}
					\label{fig:cov_close_bimodal_regret}
				\end{subfigure}
			\end{minipage}%
		}
		\caption{Coverage performance and regret in the closed information coverage environment. Columns correspond to Uniform, Gaussian Process, and Bimodal Gaussian density fields. Top row: normalized covered information value $\sum_{v\in \cup_i D_{i,t}}\rho(v) / \sum_v \rho(v)$. Bottom row: average utility gap to OSG $\frac{1}{t}\sum_{\tau=1}^{t} \bigl(F_\tau^{\mathrm{OSG}}-F_\tau^{\mathrm{method}}\bigr)$. Curves report the mean over 20 independent evaluation rollouts; shaded regions indicate standard deviation.}
		\label{fig:cov_eval_close}
	\end{figure}
	
	\begin{figure}[htbp]
		\centering
		\resizebox{0.95\textwidth}{!}{%
			\begin{minipage}{\textwidth}
				\begin{subfigure}{0.32\textwidth}
					\includegraphics[width=\linewidth]{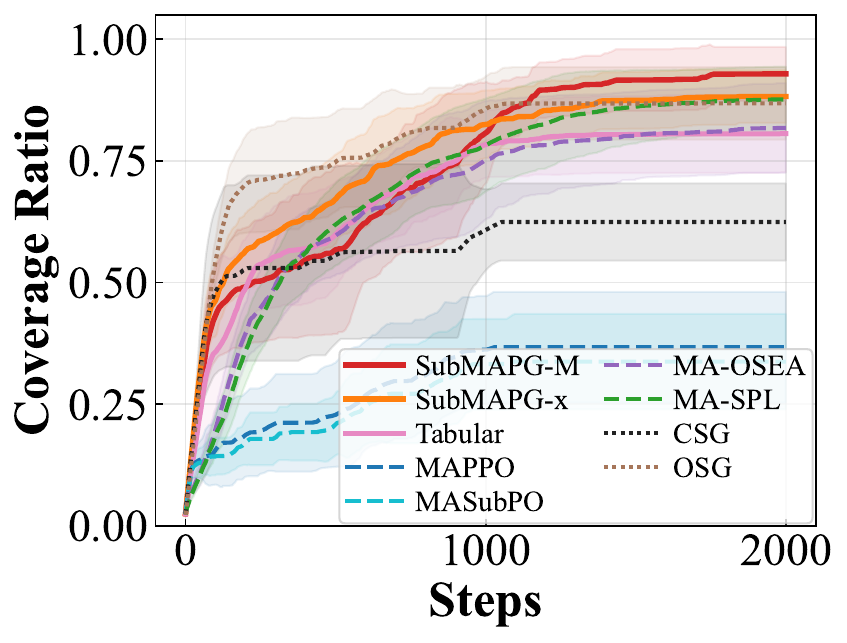}
					\caption{Uniform}
					\label{fig:cov_open_constant_ratio}
				\end{subfigure}
				\hfill
				\begin{subfigure}{0.32\textwidth}
					\includegraphics[width=\linewidth]{figures/cov_eval/open/gp_cov_ratio.pdf}
					\caption{Gaussian Process}
					\label{fig:cov_open_gp_ratio}
				\end{subfigure}
				\hfill
				\begin{subfigure}{0.32\textwidth}
					\includegraphics[width=\linewidth]{figures/cov_eval/open/bimodal_cov_ratio.pdf}
					\caption{Bimodal Gaussian}
					\label{fig:cov_open_bimodal_ratio}
				\end{subfigure}
				
				\vspace{0.1cm}
				
				\begin{subfigure}{0.32\textwidth}
					\includegraphics[width=\linewidth]{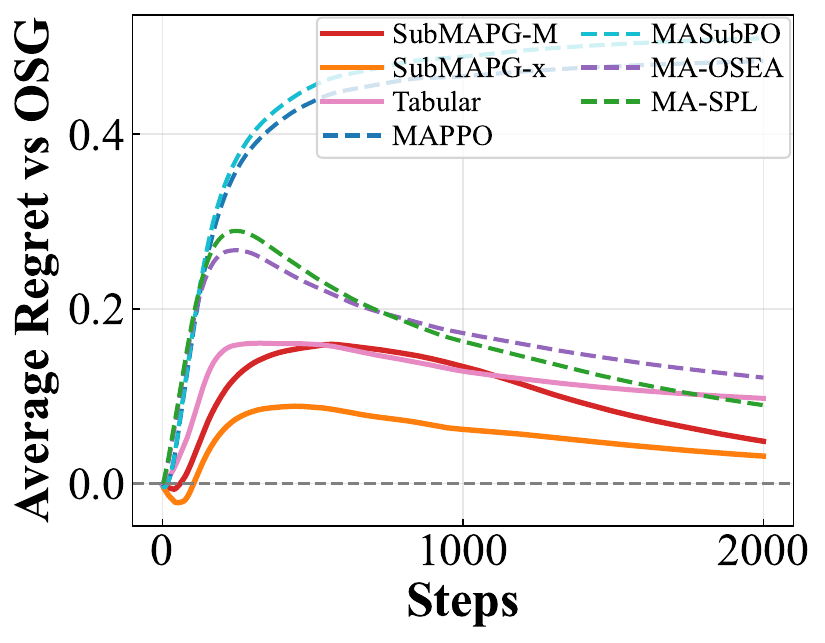}
					\caption{Uniform}
					\label{fig:cov_open_constant_regret}
				\end{subfigure}
				\hfill
				\begin{subfigure}{0.32\textwidth}
					\includegraphics[width=\linewidth]{figures/cov_eval/open/gp_avg_regret.pdf}
					\caption{Gaussian Process}
					\label{fig:cov_open_gp_regret}
				\end{subfigure}
				\hfill
				\begin{subfigure}{0.32\textwidth}
					\includegraphics[width=\linewidth]{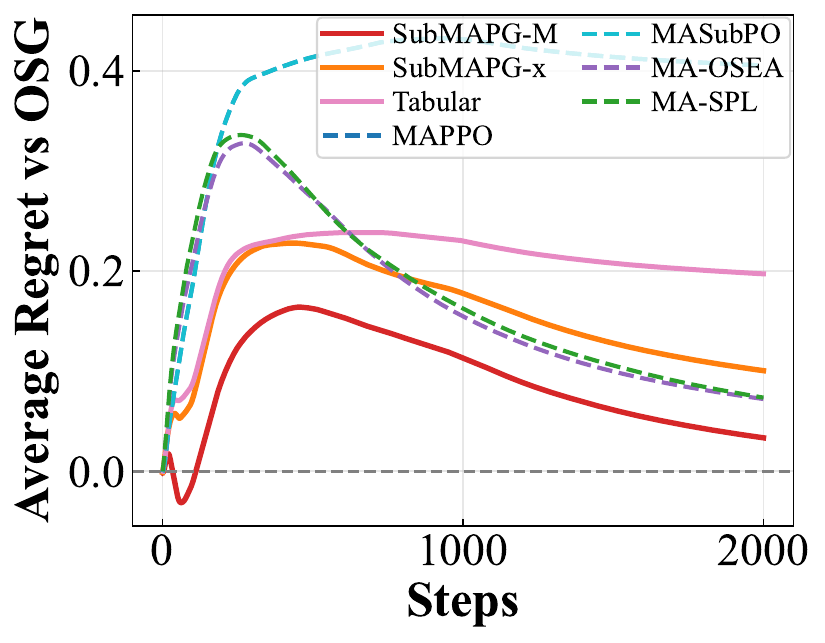}
					\caption{Bimodal Gaussian}
					\label{fig:cov_open_bimodal_regret}
				\end{subfigure}
			\end{minipage}%
		}
		\caption{
			Coverage performance and regret in the open information coverage environment with agent arrivals and departures. Columns correspond to
			Uniform, Gaussian Process, and Bimodal Gaussian density fields. Top row: normalized covered information value $\sum_{v\in \cup_i D_{i,t}}\rho(v)\big/\sum_v \rho(v)$. Bottom row: average utility gap to OSG $\frac{1}{t}\sum_{\tau=1}^{t}\bigl(F_\tau^{\mathrm{OSG}}-F_\tau^{\mathrm{method}}\bigr)$. Curves report the mean over 20 independent evaluation rollouts; shaded
			regions indicate standard deviation.}
		\label{fig:cov_eval_open}
	\end{figure}
	
	\paragraph{Tabular Softmax Equivalence.} We include a Tabular baseline to isolate the PME-based tabular softmax update in \cref{prop:tabular_equivalence}, without neural function approximation. The baseline uses independent tabular softmax policies~\citep{mei2020on,zhang2022on} with the same agent-wise difference reward as SubMAPG-M. In \autoref{fig:cov_train}, training curves are plotted against normalized training progress for visual comparison, since Tabular is trained with a larger episode budget than the neural policies. Although Tabular eventually learns reasonable policies, it typically requires around 20,000 episodes to fully converge, whereas SubMAPG-M and MASubPO converge within 5,000 episodes. This difference highlights the practical benefit of replacing tabular parameters with a neural encoder while keeping the same difference-reward update. Since Tabular samples actions stochastically during training but selects the largest-probability action at evaluation time, its evaluation performance is better than what the training curve alone suggests. However, as shown in \autoref{fig:cov_eval_close} and \autoref{fig:cov_eval_open}, Tabular remains weaker than SubMAPG-M due to its limited state representation. These results support \cref{prop:tabular_equivalence} by showing
	that, in the tabular softmax setting, the agent-wise difference reward induces a PME-compatible feasible ascent update, while neural function approximation
	improves sample efficiency and generalization in practice.
	
	\subsection{Multi-Agent Target Tracking and Performance Analysis}
	\label{app:train_env}
	\paragraph{Tracking Environment Setup.}
	Agents and targets move across a $100 \times 100\si{\meter^2}$ continuous 2D domain.
	Each agent follows unicycle kinematics with constant speed $v_a = \SI{1.0}{\meter/\second}$:
	\begin{equation}\label{eq:dynamics}
		\begin{aligned}
			p^x_{i,t+1} &= p^x_{i,t} + v_a \Delta t \cos \psi_{i,t}, \\
			p^y_{i,t+1} &= p^y_{i,t} + v_a \Delta t \sin \psi_{i,t}, \\
			\psi_{i,t+1} &= \psi_{i,t} + \omega_{i,t} \Delta t,
		\end{aligned}
	\end{equation}
	where $\mathbf p_{i,t}=(p^x_{i,t},p^y_{i,t})$ is the position of agent $i$, $\psi_{i,t}$ is its heading, $\Delta t=1$ is the time-step interval, and $\omega_{i,t}\in[-\pi/6,\pi/6]\si{\radian/\second}$ is selected from 12 uniformly spaced steering actions. Agents sense targets within $r_{\text{sen}} = 10\si{\meter}$ and communicate with other agents within $r_{\text{com}} = 25\si{\meter}$. Targets move at speed $v_m =\SI{0.25}{\meter/\second}$ with three motion patterns: Static targets remain stationary at initial positions, Linear targets move in a fixed initial direction with constant velocity, Random targets change direction randomly with 5\% probability each step while maintaining constant speed. 
	
	We use the tracking utility $F_t^{\mathrm{trk}}(A_t)$ defined
	in~\eqref{eq:app_track_utility}, which rewards each active target
	according to the closest selected agent-action pair within the sensing
	radius. For any fixed round $t$ with fixed active sets $\mathcal{N}_t$
	and $\mathcal{M}_t$, the normalization factor in
	$F_t^{\mathrm{trk}}$ is constant with respect to the selected joint
	action. 
	
	Each target-wise term has the form $\max_{e\in A_t} w_{j,t}(e)$ for a nonnegative score function $w_{j,t}(e)\in[0,1]$. Such maximum-of-score set functions are normalized, monotone, and submodular; nonnegative sums and positive rescaling preserve these properties. Since each sensing score is clipped to $[0,1]$, the marginal gains are uniformly bounded. Therefore, $F_t^{\mathrm{trk}}$ satisfies
	\cref{ass:submodular} in the tracking experiments.
	
	We train all agents in closed environments with a fixed number of 12 agents and 12 targets per episode, each episode lasting 200 steps. During training, all targets follow linear motion patterns to facilitate policy learning. We set the actor learning rate to $3\times10^{-4}$, the critic learning rate to $5\times10^{-4}$, the PPO epoch to 8, the hidden size to 128, the entropy coefficient to 0.012, and enable advantage normalization.
	
	\paragraph{Open Environment Experiments.}
	Models are evaluated in open environments where agents and targets dynamically arrive and depart. The initial set of agents and targets persists for the entire episode, while additional entities spawn at random steps uniformly sampled from $[1,0.75T]$ with $T=2500$ steps. The environment supports up to 12 agents and 12 targets. Among the 12 targets, the motion pattern distribution is set to
	Static:Linear:Random = $1{:}3{:}8$. We evaluate three initial population sizes: 4 agents/4 targets, 8 agents/8 targets, and 10 agents/10 targets. These configurations test robustness across different initial densities before the system reaches its maximum open-system capacity. As shown in \autoref{fig:appendix_open_init_pop}, SubMAPG-G consistently achieves strong cumulative utility across all three settings and maintains agent-target distances comparable to centralized greedy methods while clearly outperforming OSG, indicating robust performance under varying initial population sizes. Unlike in coverage, MA-SPL and MA-OSEA achieve substantially lower cumulative utility and larger agent-target distances than SubMAPG-G across all three open-tracking settings.
	As in the coverage experiments, these results should be interpreted as empirical evidence of robustness under open-system dynamics rather than a direct numerical verification of \cref{cor:sublinear}. The theoretical condition depends on the path length $\mathcal{P}_T$ of the optimal continuous PME solutions, which is not tractable to compute in the tracking environment.
	The evaluation protocol is designed to be consistent with the slowly varying regime: initial agents and targets persist for the full episode, additional entities arrive before $0.75T$, and each new entity has a minimum lifespan of $400$ steps.
	
	\begin{figure*}[t]
		\centering
		\setlength{\tabcolsep}{6pt}
		\renewcommand{\arraystretch}{1.0}
		
		\begin{tabular}{ccc}
			\begin{subfigure}[t]{0.31\linewidth}
				\centering
				\includegraphics[width=\linewidth]{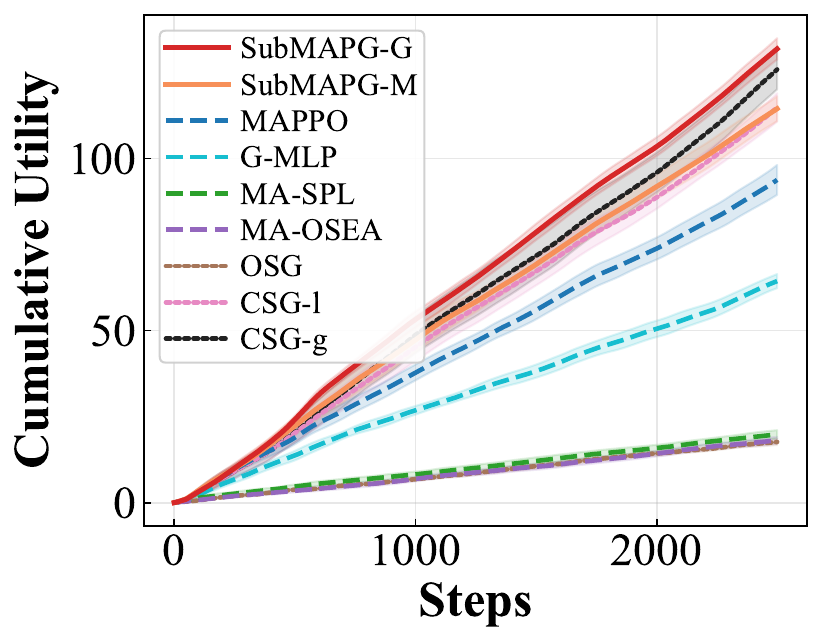}
				\caption{$N_0=M_0=4$}
				\label{fig:init_4a_4t_cum_app}
			\end{subfigure}
			&
			\begin{subfigure}[t]{0.31\linewidth}
				\centering
				\includegraphics[width=\linewidth]{figures/track_eval/openness/init_8a_8t_cum.pdf}
				\caption{$N_0=M_0=8$}
				\label{fig:init_8a_8t_cum_app}
			\end{subfigure}
			&
			\begin{subfigure}[t]{0.31\linewidth}
				\centering
				\includegraphics[width=\linewidth]{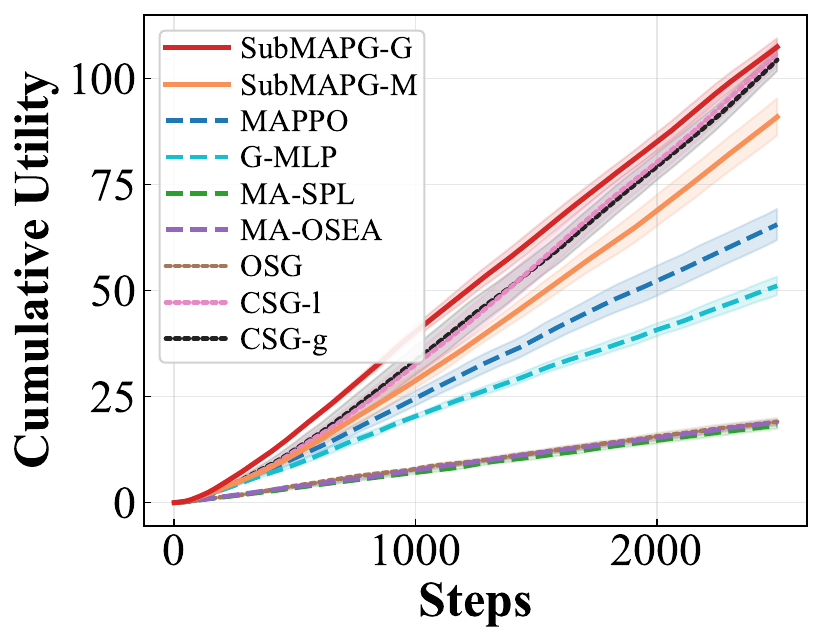}
				\caption{$N_0=M_0=10$}
				\label{fig:init_10a_10t_cum_app}
			\end{subfigure}
			\\[0.35em]
			\begin{subfigure}[t]{0.31\linewidth}
				\centering
				\includegraphics[width=\linewidth]{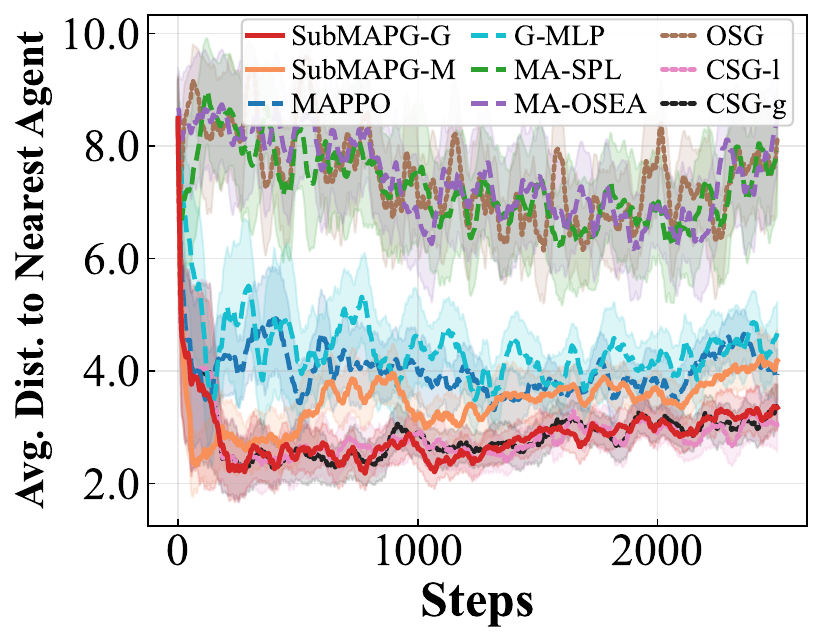}
				\caption{$N_0=M_0=4$}
				\label{fig:init_4a_4t_dist_app}
			\end{subfigure}
			&
			\begin{subfigure}[t]{0.31\linewidth}
				\centering
				\includegraphics[width=\linewidth]{figures/track_eval/openness/init_8a_8t_dist.pdf}
				\caption{$N_0=M_0=8$}
				\label{fig:init_8a_8t_dist_app}
			\end{subfigure}
			&
			\begin{subfigure}[t]{0.31\linewidth}
				\centering
				\includegraphics[width=\linewidth]{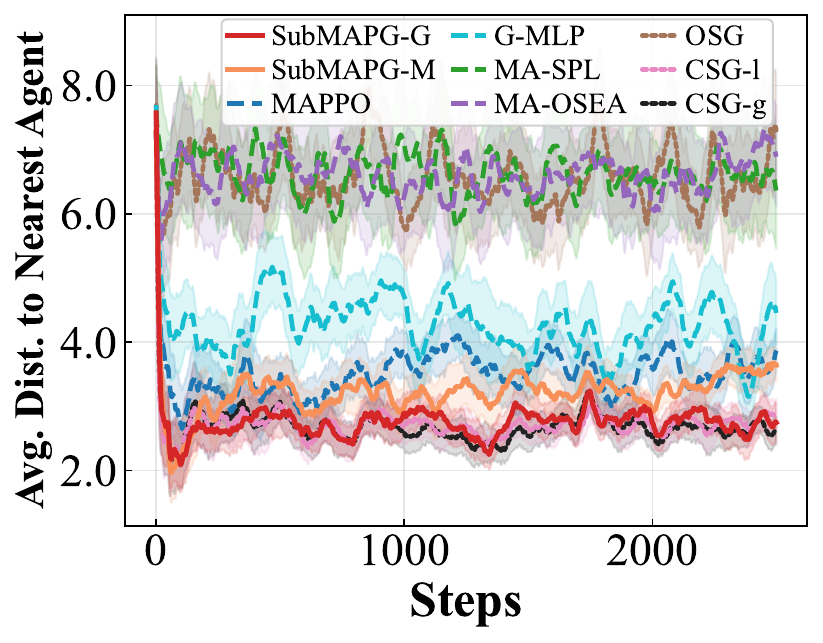}
				\caption{$N_0=M_0=10$}
				\label{fig:init_10a_10t_dist_app}
			\end{subfigure}
		\end{tabular}
		
		\caption{Tracking performance in open environments under different initial population sizes. Columns correspond to $N_0=M_0=4$, $N_0=M_0=8$, and $N_0=M_0=10$, where $N_0$ and $M_0$ denote the initial numbers of active agents and targets respectively. Top row: cumulative utility $\sum_{\tau=1}^{t} F_\tau$. Bottom row: average distance from each target to its nearest agent. Curves report the mean over 5 seeds; shaded regions indicate one standard
			deviation.
		}
		\label{fig:appendix_open_init_pop}
	\end{figure*}

	\paragraph{Reward Mechanism: Difference Reward Versus Global Reward.}
	We compare reward design and policy optimization using baselines with matched encoders. G-MLP and G-GNN use the SubMAPG policy-gradient update with the shared global reward $r_t=F_t^{\mathrm{trk}}(A_t)$, whereas SubMAPG-M and SubMAPG-G use the agent-wise difference reward computed as
    $F_t^{\mathrm{trk}}(A_t)-F_t^{\mathrm{trk}}(A_t^{-i})$. MAPPO(MLP) and MAPPO(GNN) use the same global reward with MAPPO's PPO-based policy optimization.
	As illustrated in \autoref{fig:diff-global}, MAPPO(MLP) improves over G-MLP, while SubMAPG-M converges faster and achieves higher cumulative utility than MAPPO(MLP), highlighting the effectiveness of difference rewards for decentralized coordination. This advantage is consistently corroborated by results in the coverage task (see \autoref{fig:cov_train}). With GNN encoders, only SubMAPG-G converges effectively, whereas MAPPO(GNN) and G-GNN fail to achieve effective convergence.
	
	\begin{figure}[htbp]
		\centering
		\includegraphics[width=0.5\linewidth]{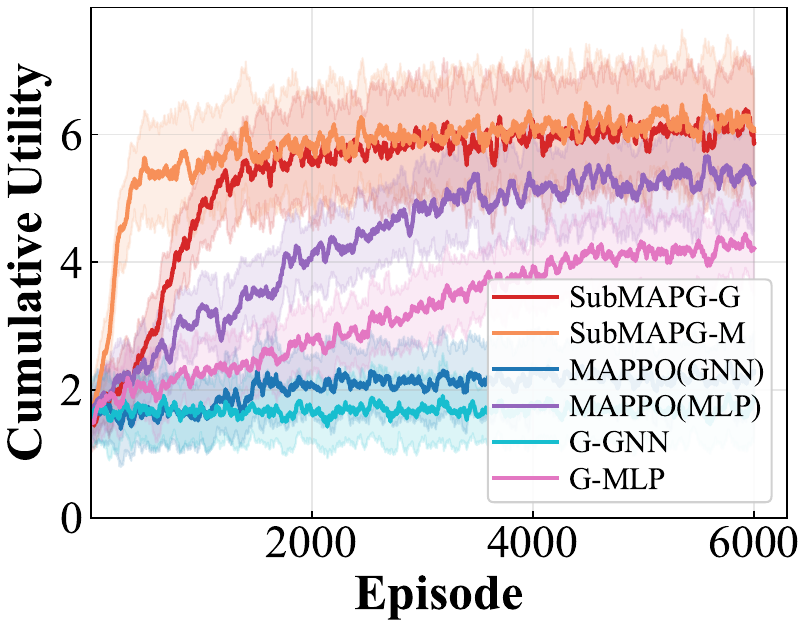}
		\caption{
			Reward mechanism and policy optimization ablation in the tracking task. SubMAPG-M and SubMAPG-G use the agent-wise difference reward computed as
            $F_t^{\mathrm{trk}}(A_t)-F_t^{\mathrm{trk}}(A_t^{-i})$, G-MLP and G-GNN use the shared global utility reward
			$r_t=F_t^{\mathrm{trk}}(A_t)$ with the SubMAPG policy-gradient update. MAPPO(MLP) and MAPPO(GNN) use the same global reward with PPO-based
			policy optimization. Comparisons use matched encoders, and only
			SubMAPG-G converges effectively among the GNN-based methods.}
		\label{fig:diff-global}
	\end{figure}
	
	\paragraph{Scalability.}
	\begin{figure}[htbp]
		\centering
		\includegraphics[width=0.5\linewidth]{figures/track_train/scale.pdf}  
		\caption{
			Zero-shot scalability from 12 to 48 agents and targets using a policy trained on systems with up to 12 agents and targets. The map side length and sensing radius are scaled to keep the spatial density and sensing difficulty approximately constant. SubMAPG-G(scaled) also scales the communication radius $r_{\mathrm{com}}$, while SubMAPG-G(fixed) keeps $r_{\mathrm{com}}=\SI{25}{\meter}$.}
		\label{fig:scale}
	\end{figure}

	We evaluate zero-shot scalability by deploying policies trained with up to 12 agents and targets to systems with up to 48 agents and targets. To keep the spatial density and sensing difficulty approximately constant, we scale the map side length and sensing radius $r_{\mathrm{sen}}$ together. Without this scaling, larger maps would make targets substantially harder to observe regardless of the policy. As shown in \autoref{fig:scale}, both SubMAPG-G variants retain 50-65\% of CSG-l's cumulative utility despite the domain
	shift. Recall that CSG-l uses centralized coordination, whose advantage grows with the number of agents and targets. Notably, SubMAPG-G(fixed) with $r_{\mathrm{com}}=\SI{25}{\meter}$ performs comparably to the scaled communication variant, suggesting that the learned policy generalizes to larger systems without increasing communication range.

	\paragraph{Robustness Analysis.}
	
	In the following analysis, we consider a closed environment with 12 agents and 12 targets, and perform ablations on individual parameters of the environment to isolate their specific effects.
	
	\textit{Robustness to Target Motion Patterns.} 
	We assess robustness against diverse target behaviors using Static:Linear:Random ratios of 10:1:1, 5:4:3, and 1:3:8. As shown in \autoref{fig:ratio}, SubMAPG-G rivals or exceeds the centralized CSG-g, and SubMAPG-M remains competitive with CSG-l. Although the performance gap narrows as stochasticity increases in \autoref{fig:1_3_8_cum}, both variants maintain strong tracking capabilities. Furthermore, SubMAPG-M’s consistent superiority over G-MLP in all configurations validates the efficacy of our difference reward formulation.
	
	\begin{figure}[htbp]
		\centering
		\resizebox{0.95\textwidth}{!}{%
			\begin{minipage}{\textwidth}
				\begin{subfigure}{0.32\textwidth}
					\includegraphics[width=\linewidth]{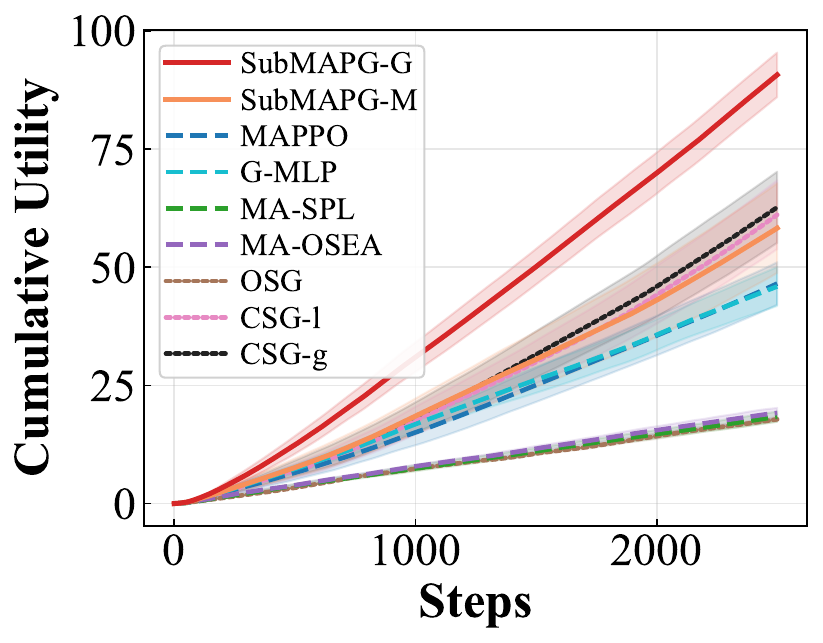}
					\caption{10:1:1}
					\label{fig:10_1_1_cum}
				\end{subfigure}
				\hfill
				\begin{subfigure}{0.32\textwidth}
					\includegraphics[width=\linewidth]{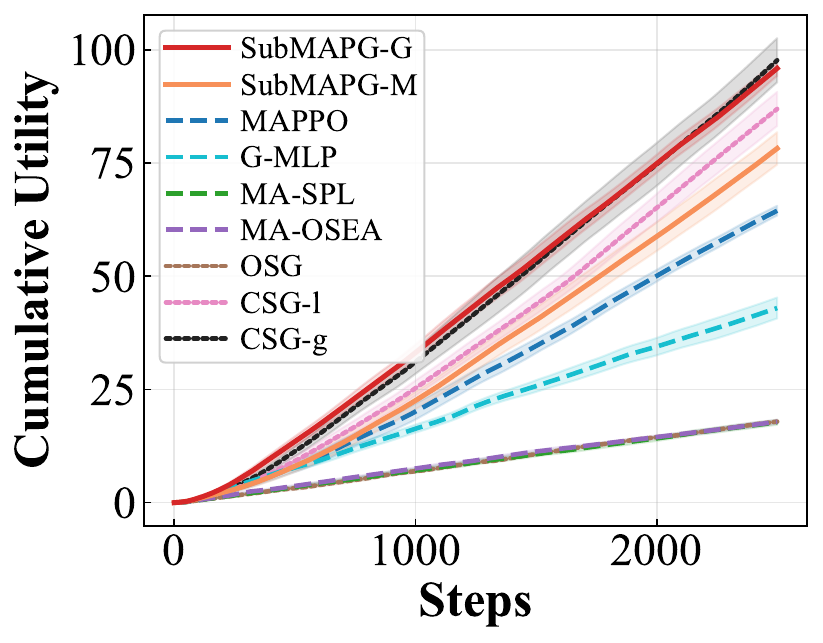}
					\caption{5:4:3}
					\label{fig:5_4_3_cum}
				\end{subfigure}
				\hfill
				\begin{subfigure}{0.32\textwidth}
					\includegraphics[width=\linewidth]{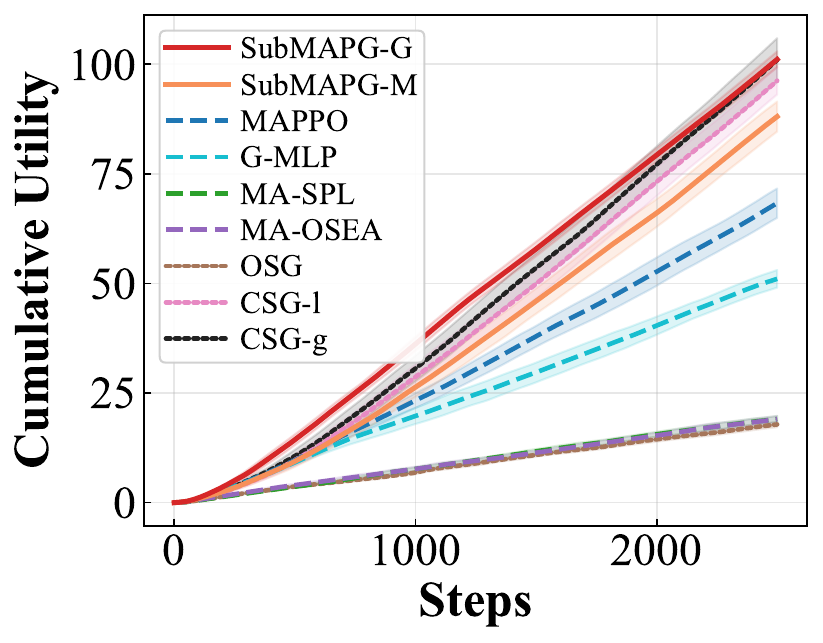}
					\caption{1:3:8}
					\label{fig:1_3_8_cum}
				\end{subfigure}
				
				\vspace{0.1cm}
				
				\begin{subfigure}{0.32\textwidth}
					\includegraphics[width=\linewidth]{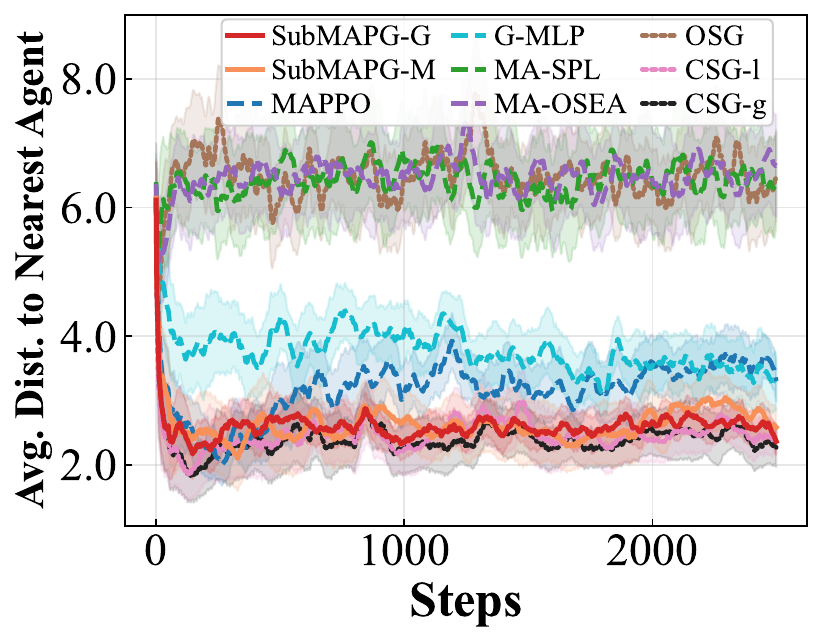}
					\caption{10:1:1}
					\label{fig:10_1_1_dist}
				\end{subfigure}
				\hfill
				\begin{subfigure}{0.32\textwidth}
					\includegraphics[width=\linewidth]{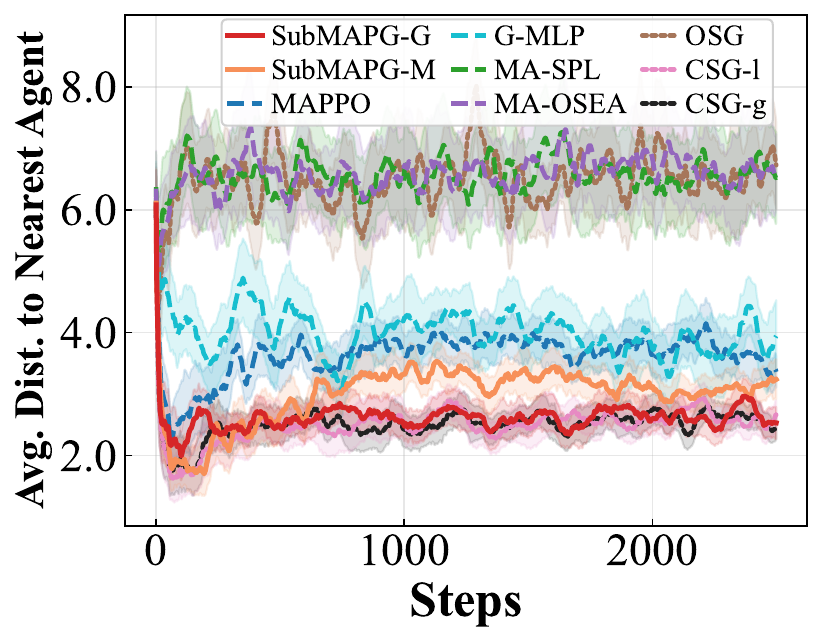}
					\caption{5:4:3}
					\label{fig:5_4_3_dist}
				\end{subfigure}
				\hfill
				\begin{subfigure}{0.32\textwidth}
					\includegraphics[width=\linewidth]{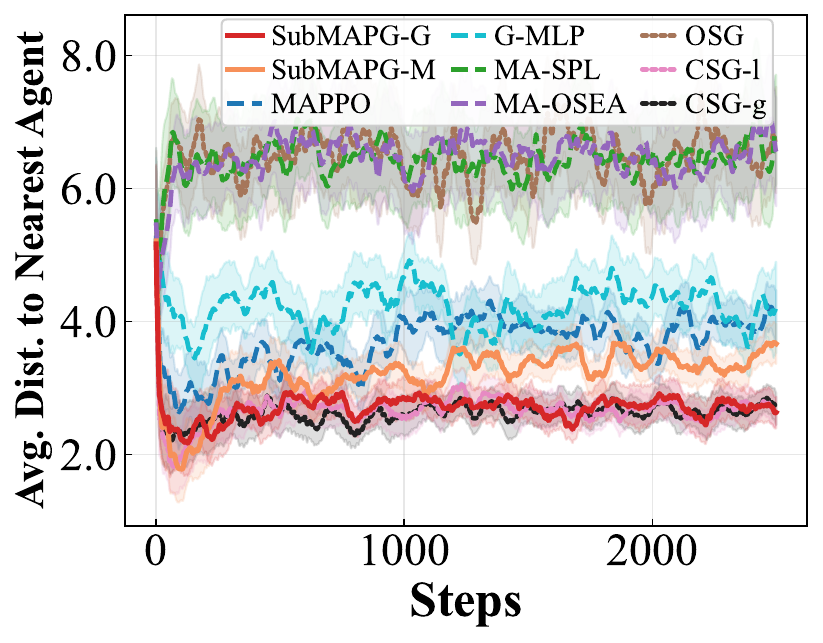}
					\caption{1:3:8}
					\label{fig:1_3_8_dist}
				\end{subfigure}
			\end{minipage}%
		}
		\caption{Tracking performance under different target-motion compositions. Columns correspond to Static:Linear:Random ratios of 10:1:1, 5:4:3,
			and 1:3:8. Top row: cumulative utility $\sum_{\tau=1}^{t} F_\tau$. Bottom row: average distance from each target to its nearest agent.
			Curves report the mean over 5 seeds; shaded regions indicate one standard deviation.}
		\label{fig:ratio}
	\end{figure}
	
	\textit{Robustness to Communication Range.}
	We vary $r_{\mathrm{com}}\in\{5,15,35\}$\,m while fixing $r_{\mathrm{sen}}=10$\,m. As shown in \cref{fig:comm}, both SubMAPG variants benefit from increased communication range, but exhibit different sensitivities. At a restricted range of $r_{\mathrm{com}}=5$\,m, limited information exchange degrades the performance of learning-based methods. Comparing the results across communication ranges in
	\cref{fig:comm}, SubMAPG-G recovers near-greedy performance, whereas SubMAPG-M shows continued performance gains as $r_{\mathrm{com}}$ increases from $15$\,m to $35$\,m. These results indicate that the GNN-based variant is more resilient to limited communication.
	
	\begin{figure}[htbp]
		\centering
		\resizebox{\textwidth}{!}{%
			\begin{minipage}{0.95\textwidth}
				\begin{subfigure}{0.32\textwidth}
					\includegraphics[width=\linewidth]{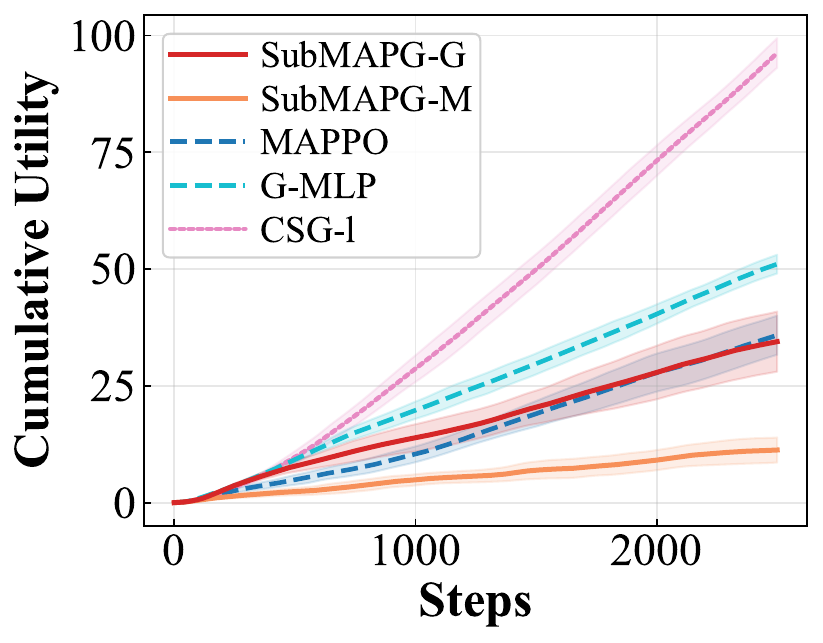}
					\caption{$r_{\mathrm{com}}=\SI{5}{\meter}$}
					\label{fig:com_5_cum}
				\end{subfigure}
				\hfill
				\begin{subfigure}{0.32\textwidth}
					\includegraphics[width=\linewidth]{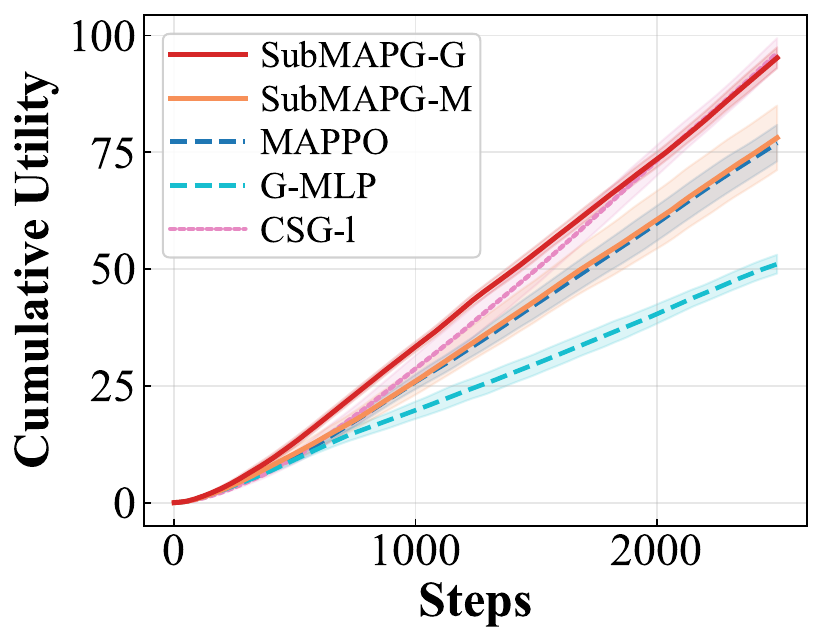}
					\caption{$r_{\mathrm{com}}=\SI{15}{\meter}$}
					\label{fig:com_15_cum}
				\end{subfigure}
				\hfill
				\begin{subfigure}{0.32\textwidth}
					\includegraphics[width=\linewidth]{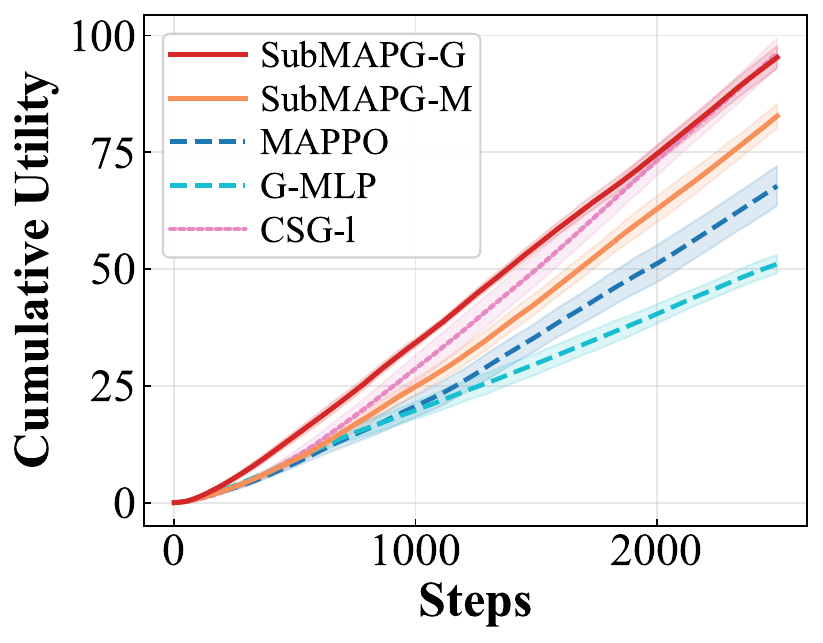}
					\caption{$r_{\mathrm{com}}=\SI{35}{\meter}$}
					\label{fig:com_35_cum}
				\end{subfigure}   
				
				\vspace{0.1cm}
				
				\begin{subfigure}{0.32\textwidth}
					\includegraphics[width=\linewidth]{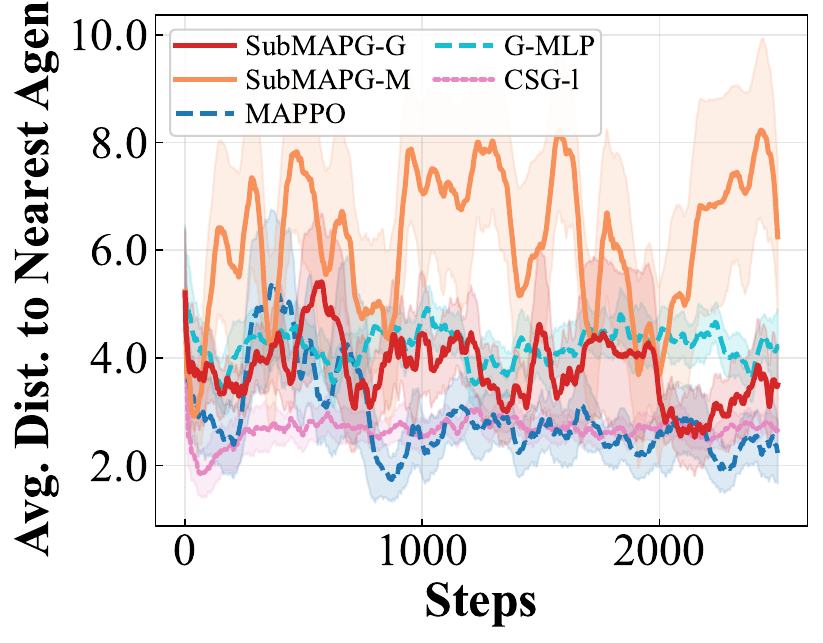}
					\caption{$r_{\mathrm{com}}=\SI{5}{\meter}$}
					\label{fig:com_5_dist}
				\end{subfigure}
				\hfill
				\begin{subfigure}{0.32\textwidth}
					\includegraphics[width=\linewidth]{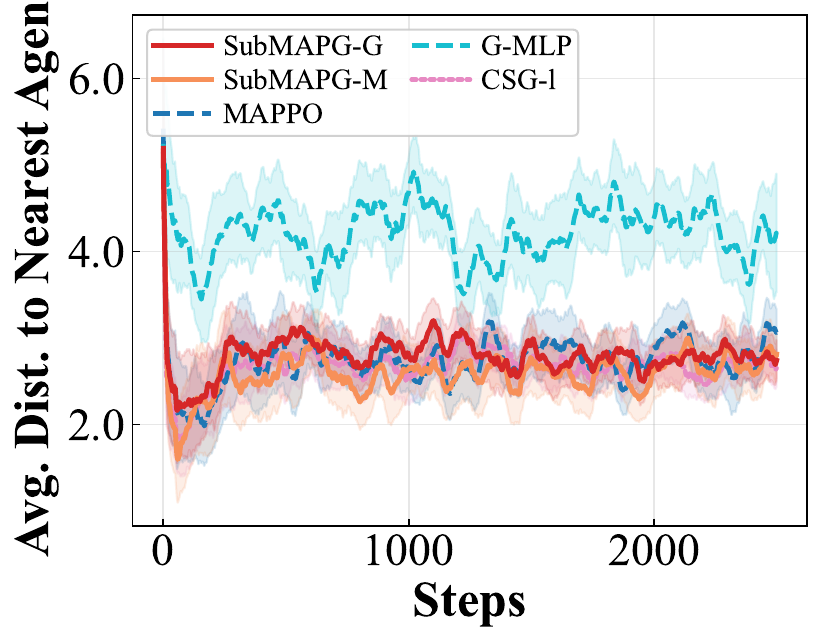}
					\caption{$r_{\mathrm{com}}=\SI{15}{\meter}$}
					\label{fig:com_15_dist}
				\end{subfigure}
				\hfill
				\begin{subfigure}{0.32\textwidth}
					\includegraphics[width=\linewidth]{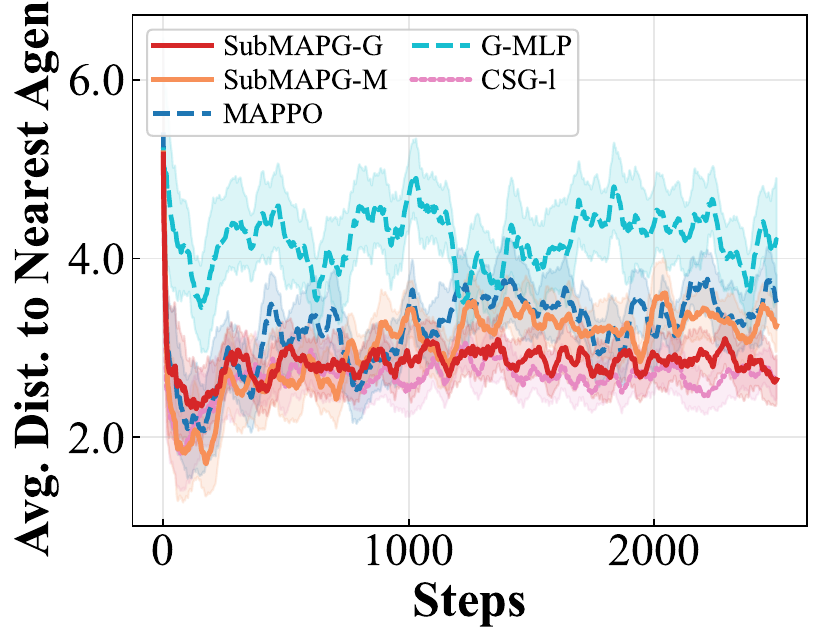}
					\caption{$r_{\mathrm{com}}=\SI{35}{\meter}$}
					\label{fig:com_35_dist}
				\end{subfigure}
			\end{minipage}%
		}
		\caption{Tracking performance under different communication ranges with fixed sensing radius $r_{\mathrm{sen}}=\SI{10}{\meter}$. Columns correspond to $r_{\mathrm{com}}=\SI{5}{\meter}$, $\SI{15}{\meter}$, and $\SI{35}{\meter}$. Top row: cumulative utility $\sum_{\tau=1}^{t} F_\tau$. Bottom row: average distance from each target
			to its nearest agent. Curves report the mean over 5 seeds; shaded regions indicate one standard deviation.}
		\label{fig:comm}
	\end{figure}
	
	
	
\end{document}